\documentclass[11pt]{article}

\usepackage{bbm}
\usepackage{tikz}
\usepackage{quantikz}
\usepackage{orcidlink}
\usepackage{graphicx}
\usepackage{stmaryrd}
\usepackage{soul,xcolor}
\setstcolor{red}

\usetikzlibrary{calc,decorations.pathreplacing}

\usepackage{authblk} 

\usepackage[utf8]{inputenc}
\usepackage{lipsum}
\usepackage{adjustbox}

\usepackage{amsmath, amssymb,amsfonts,amsthm,mathtools}
\usepackage{xspace,graphicx,relsize,bm,bbm}
\usepackage{soul} 

\usepackage{parskip}  

\usepackage{booktabs}
\usepackage[table]{xcolor}
\usepackage{pifont}
 \usepackage{makecell}

\usepackage{libertine}
\usepackage{libertinust1math}
\usepackage{dsfont}
\usepackage[T1]{fontenc}
\usepackage{enumitem}

\usepackage[
    backend=biber,
    style=alphabetic,
    sorting=anyt,
    minalphanames=3,
    maxalphanames=3,
    maxnames=99,
    backref=true
    ]{biblatex}
\DefineBibliographyStrings{english}{%
  backrefpage = {page},
  backrefpages = {pages},
}

\usepackage{sepfootnotes}
\newendnotes{x}
\renewcommand\xnotesize\normalsize

\usepackage{hyperref}
\usepackage[margin=1.75cm]{geometry}
\definecolor{linkcol}{rgb}{0.0,0.55,0.7}
\definecolor{citecol}{rgb}{0.0, 0.6, 0.45}
\definecolor{urlcol}{rgb}{0.7, 0.0, 0.55}
\hypersetup{
	colorlinks,
	linkcolor={linkcol},
	citecolor={citecol},
	urlcolor={urlcol}
}

\usepackage{url}
\usepackage{subcaption}
\usepackage{mleftright}
\usepackage{hyperref}
\usepackage{multirow}
\usepackage{physics}  

\usepackage{algorithm}
\usepackage{algpseudocodex}[indLines = true,italicComments = false]

\usepackage{cleveref}

\usepackage{authblk}  

\def\01{\{0,1\}}

\newcommand{\poly}{\mathrm{poly}}

\newcommand{\iu}{\mathrm{i}}

\newtheoremstyle{mydefinitionsty}
{10pt}
{10pt}
{}
{}
{}
{}
{.5em}
{\textbf{\thmname{#1}~\thmnumber{#2}:  }\thmnote{(#3)}}
\theoremstyle{mydefinitionsty}
\newtheorem{definition}{Definition}[section]
\newtheorem{remark}[definition]{Remark}
\newtheorem{informal-definition}[definition]{Informal Definition}

\newtheorem{assumption}[definition]{Assumption}

\newtheoremstyle{myproblemsty}
{10pt}
{10pt}
{}
{}
{}
{}
{.5em}
{\textbf{\thmname{#1}~\thmnumber{#2}:  }\thmnote{(#3)}\newline}
\theoremstyle{myproblemsty}

\newtheoremstyle{mythmsty}
{10pt}
{10pt}
{\itshape}
{}
{}
{}
{.5em}
{\textbf{\thmname{#1}~\thmnumber{#2}:  }\thmnote{(#3)}}
\theoremstyle{mythmsty}

\newtheorem{theorem}[definition]{Theorem}
\newtheorem{lemma}[definition]{Lemma}
\newtheorem{corollary}[definition]{Corollary}

\newtheorem{question}[definition]{Question}
\newtheorem{construction}[definition]{Construction}

\numberwithin{equation}{section}

\definecolor{highlightrow}{RGB}{225, 245, 238}

\title{Instantiating Microcrypt: \\Obstacles and opportunities via tailored state certification}

\author[1]{Jose Carrasco\,\thanks{jose.carrasco@fu-berlin.de}}
\author[1,2,3]{Jens Eisert}
\author[4, 12]{Soumik Ghosh}
\author[11,12]{Dominik Hangleiter}
\author[5,6,7]{Nicky Kai~Hong~Li} 
\author[8,9,10]{Ryan Sweke\,\thanks{rsweke@aims.ac.za}}
\affil[1]{Dahlem Center for Complex Quantum Systems, Freie Universit\"{a}t Berlin, 14195 Berlin, Germany}
\affil[2]{Helmholtz-Zentrum Berlin für Materialien und Energie, 14109 Berlin, Germany}
\affil[3]{Fraunhofer Heinrich Hertz Institute, 10587 Berlin, Germany}
\affil[4]{Center for Theoretical Physics, Massachusetts Institute of Technology, 77 Massachusetts Ave, Cambridge, MA 02139, USA}
\affil[5]{Technische Universit\"{a}t Wien, Atominstitut, Stadionallee 2, 1020 Vienna, Austria}
\affil[6]{Vienna Center for Quantum Science and Technology, TU Wien, 1020 Vienna, Austria}
\affil[7]{Institute for Quantum Optics and Quantum Information (IQOQI), Austrian Academy of Sciences, Boltzmanngasse 3, 1090 Vienna, Austria}
\affil[8]{African Institute for Mathematical Sciences (AIMS), South Africa}
\affil[9]{Department of Mathematical Sciences, Stellenbosch University, Stellenbosch 7600, South Africa}

\affil[10]{National Institute for Theoretical and Computational Sciences (NITheCS), South Africa}
\newcommand{\simons}{%
Simons Institute for the Theory of Computing, University of California at Berkeley, USA}
\newcommand{\ethz}{%
Institute for Theoretical Physics, ETH Z\"urich, Switzerland
}
\affil[11]{\ethz}
\affil[12]{\simons}
\date{\today}

\begin{document}
\maketitle

\begin{abstract} Recent work has introduced the Hamiltonian phase state (HPS) assumptions, which postulate that Hamiltonian phase states can be used to instantiate pseudorandom and one-way state generators~\cite{bostanci2025efficientquantumpseudorandomnesshamiltonian}. Additionally, it has been conjectured that these assumptions can be true, even if one-way functions do not exist. This is exciting, because if true, then the HPS assumptions provide a route to the instantiation of cryptography which is genuinely in Microcrypt -- the world of (quantum) cryptographic protocols and primitives which can exist even if one-way functions do not. In this work we falsify this conjecture, by proving that if the HPS assumptions are true, then one-way functions exist. While this removes the possibility of instantiating genuine Microcrypt protocols and primitives with Hamiltonian phase states, it shows that the HPS assumptions provide novel inherently quantum assumptions for the construction of classical cryptography. The technical contribution that allows us to do this is a method for the construction of one-way puzzles from one-way state generators via tailored "measure first, ask later" state certification protocols. This generalizes prior constructions of one-way puzzles from one-way state generators via classical shadows and allows us to relate properties of the resulting one-way puzzle to properties of the state certification protocol used in the construction. Specifically, if the state certification protocol admits efficient classical post-processing then one obtains an efficiently verifiable one-way puzzle, and if the state certification protocol can be efficiently classically simulated in a certain sense, then one obtains a classical one-way puzzle, which implies one-way functions. Indeed, the latter observation allows us to prove that the HPS assumptions imply one-way functions, by exploiting properties of recent state certification protocols for phase states. The former observation provides a new toolbox for the construction of efficiently verifiable one-way puzzles by exploiting tailored state certification protocols for pseudorandom and one-way state generators.
\end{abstract}

\newpage

\tableofcontents

\newpage

\section{Introduction}\label{s:introduction}

The existence of \emph{one-way functions} (OWFs) is a minimal assumption for classical cryptography. More specifically, all meaningful classical cryptographic primitives imply OWFs, and thus OWFs are necessary (but not sufficient) for classical cryptography. Remarkably, however, over the last few years, evidence has emerged that meaningful cryptography may be possible in a \textit{quantum} world even if $\mathsf{P}=\mathsf{NP}$ (and hence OWFs do not exist). More specifically:

\begin{enumerate}
\item A series of breakthrough works ~\cite{Kretschmer_2021,Kretschmer_2023,Kretschmer_2025} have culminated in the construction of oracle worlds in which $\mathsf{P}=\mathsf{NP}$, but quantum cryptographic primitives such as quantum computable OWFs and pseudorandom state generators (PRSG) do exist. 
\item Khurana and Tomer have shown that one can construct one-way puzzles (OWPs) from ``quantum advantage assumptions''
for quantum random sampling schemes, which can be true even if $\mathsf{P}=\mathsf{NP}$~\cite{khurana2025founding}, under extremely mild
complexity-theoretic assumptions.
\end{enumerate}
In addition, a variety of works have also shown how to construct non-trivial quantum cryptographic protocols from such inherently quantum primitives, including variants of bit-commitments, secure multiparty computation and digital signatures amongst others~\cite{ananth2022cryptography, Prabhanjan,Morimae_2022, grilo2025quantum, khurana2024commitmentsquantumonewayness, ananth2023pseudorandomstringspseudorandomquantum, morimae2024onewaynessquantumcryptography,QCCCCrypto,Kretschmer_2025}.

In line with the ``world building'' tradition of theoretical cryptography, the name \textit{Microcrypt} has been given to the world in which OWFs do not exist, but genuinely quantum cryptographic primitives such as quantum computable OWFs, OWPs and PRSGs (amongst others) do exist. Motivated by how unexpected and exciting this world is, 
recent years have witnessed significant effort to map this world~\cite{microcryptzoo}, i.e.:
\begin{enumerate}
\item To propose and define Microcrypt primitives, understand which primitives imply which, and which complexity-theoretic conditions are necessary for a primitive to exist.
\item To understand which cryptographic protocols can be built from which Microcrypt primitives.
\item To propose \textit{concrete instantiations} of Microcrypt primitives under suitable \textit{quantum} assumptions, which are plausibly independent of OWFs.
\end{enumerate}
Given this work, we now understand that Microcrypt has a hierarchical structure -- i.e., it consists of multiple distinct ``subworlds''. Each such subworld is defined by a necessary upper bound on the computational power of quantum computers, and has an associated candidate minimal primitive -- i.e., a primitive which can be constructed from every other primitive in the subworld, and is thus necessary for its existence. More specifically, as discussed and illustrated in Ref.~\cite{goldin2024countcrypt}, we have the following Microcrypt hierarchy:
\begin{enumerate}
    \item \textbf{Quantumania:} The cryptographic primitives which are broken if $\mathsf{BQP}=\mathsf{QCMA}$. Nearly every primitive in Quantumania implies the existence of \textit{efficiently-verifiable} OWPs (EV-OWPs)~\cite{khurana2024commitmentsquantumonewayness,khurana2025founding}, which can therefore be considered as a minimal primitive for this world. This world is particularly interesting, given that it contains a variety of cryptographic protocols which can be executed with local quantum computation but only classical communication -- i.e., QCCC cryptography~\cite{QCCCCrypto}. Additionally, Ref.~\cite{Kretschmer_2025} has recently provided an oracle relative to which $\mathsf{P}=\mathsf{NP}$, yet quantum computable trapdoor OWFs (which imply EV-OWPs) exist, thereby giving evidence that Quantumania may indeed be non-empty even if OWFs do not exist.
\item \textbf{Countcrypt:} The cryptographic primitives which are not necessarily broken if $\mathsf{BQP}=\mathsf{QCMA}$, but are broken if $\mathsf{BQP}=\mathsf{PP}$. Amongst others, this world contains PRSGs~\cite{PRSGdefinition}, one-way state generators (OWSGs) and standard OWPs~\cite{khurana2024commitmentsquantumonewayness,khurana2025founding} (which appear to be a minimal primitive for this world). 
\item \textbf{Nanocrypt:} The cryptographic primitives which may exist even if $\mathsf{BQP}=\mathsf{PP}$, and for which efficiently generated, statistically far-apart, \emph{computationally indistinguishable} 
(EFI) pairs~\cite{BCQ2023} appear to be a minimal primitive.
\end{enumerate}
In order to realize any of the cryptography possible in any subworld 
of Microcrypt, one requires proposals for \textit{concrete instantiations} of the relevant minimal primitive.  Of course, any such concrete instantiation will come with its own assumptions, and importantly, for the instantiation to be genuinely in Microcrypt, its assumptions should \textit{not} imply the existence of OWFs. With the goal of providing such concrete instantiations under inherently quantum assumptions, Ref.~\cite{bostanci2025efficientquantumpseudorandomnesshamiltonian} recently proposed the \textit{Hamiltonian phase state (HPS) assumptions}. Specifically, these assumptions postulate that ensembles of Hamiltonian phase states provide concrete PRSGs and OWSGs, which can then be used to construct OWPs~\cite{khurana2024commitmentsquantumonewayness}. Importantly, Ref.~\cite{bostanci2025efficientquantumpseudorandomnesshamiltonian} also provided some preliminary evidence that the HPS assumptions may be true, even if OWFs do not exist. This is exciting, as if this is the case, then the HPS assumptions indeed provide a clear foundation for the instantiation of the Countcrypt subworld of Microcrypt.  However, given that the HPS assumptions are new and untested, the natural question that we address in this work is the following:
\begin{center}
\textit{Can the Hamiltonian phase state assumptions be true if one-way functions do not exist?}
\end{center}
We resolve this question in the \textit{negative}, showing that if the Hamiltonian phase state assumptions are true, then OWFs exist. This unfortunately provides a clear obstacle for instantiating Microcrypt, by showing that the HPS assumptions are not sufficient for the instantiation of genuine Microcrypt cryptography. It does however mean that the HPS assumptions provide a novel \textit{inherently quantum} assumption for the construction of classical cryptography, complementing recent work aiming precisely to provide such assumptions~\cite{YihuiLSN1,YihuiLSN2}.

The primary technical contribution that allows us to resolve the question above is a method for the construction of OWP variants from OWSGs via "measure first, ask later" state certification protocols for the set of states defining the OWSG. This generalizes and abstracts the existing construction of OWPs from OWSGs via classical shadows~\cite{khurana2024commitmentsquantumonewayness}. Importantly, it also allows us to relate properties of the resulting OWP to properties of the state certification protocol that was used. 
In particular, if the state certification measurement protocol can be classically simulated in a specific sense, then the resulting OWP is a \textit{classical} OWP, which then implies OWFs~\cite{khurana2024commitmentsquantumonewayness,khurana2025founding}. Indeed, this  fact is what allows us to prove that the HPS assumptions imply OWFs, by exploiting properties of existing state certification protocols for (Hamiltonian) phase states~\cite{HuangPreskillSoleimanifar2025}.

Importantly, however, the above method for constructing OWPs from OWSGs and state certification protocols also allows us to show that if the state certification protocol admits efficient classical post-processing, then the resulting OWP is \textit{efficiently verifiable}. This provides a new toolbox for constructing \textit{efficiently verifiable} OWPs, and may open up new opportunities for identifying concrete ensembles of states with which to instantiate Quantumania.

\subsection{Our contributions}\label{ss:contributions}

Motivated by the above question, we make a variety of contributions, which can be loosely categorized into those providing \textit{obstacles} and those providing \textit{opportunities} (hopefully) for the concrete instantiation of Microcrypt.

\subsubsection{Obstacles: Hamiltonian phase state assumptions imply one-way functions}\label{sss:contributions-HPSimplyowf}

We postpone formal definitions to Section~\ref{ss:HPS-prelim}, but informally, the \textit{Hamiltonian phase state (HPS)} assumptions, recently introduced in Ref.~\cite{bostanci2025efficientquantumpseudorandomnesshamiltonian}, are the following:

\begin{assumption}[Hamiltonian phase state assumptions \cite{bostanci2025efficientquantumpseudorandomnesshamiltonian}] $\,\,$
\begin{enumerate}
\item \textbf{Decision HPS assumption} (informal version of Assumption~\ref{ass:decision-HPS}): There exists a set of Hamiltonian phase states, and an efficiently sampleable distribution over this set, which can be used to construct a pseudorandom state generator.
\item \textbf{Search HPS assumption}  (informal version of Assumption~\ref{ass:search-HPS}): There exists a set of Hamiltonian phase states, and an efficiently sampleable distribution over this set, which can be used to construct a one-way state generator.
\end{enumerate}
\end{assumption}

\begin{figure}
    \centering
     \includegraphics[width=\linewidth]{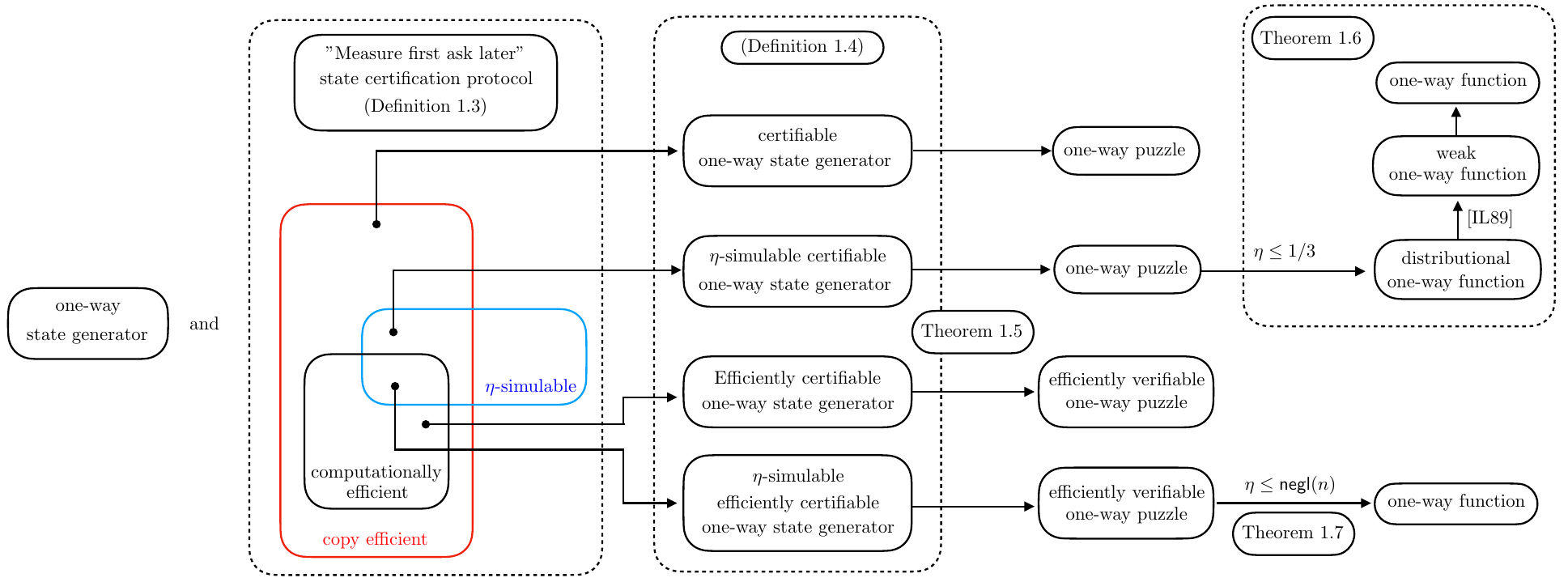} 
    \caption{Summary of the main definitions and constructions developed in this work (where $n$ is taken as the security parameter of the OWSG). Theorem~\ref{thm:HPS-imply-owf-intro} below is proven by showing that there exists an $\eta$-simulable and computationally efficient "measure first, ask later" state certification protocol for the OWSG obtained via the Search HPS assumption, and therefore one can construct a OWF by the implications of either Theorem~\ref{thm:owf-from-sim-cert-owsg-intro} or Theorem~\ref{thm:owf-from-sim-cert-owsg-direct-intro}. As we discuss in Section~\ref{ss:towards-simpler}, Theorem~\ref{thm:owf-from-sim-cert-owsg-direct-intro} yields a simpler OWF, by virtue of not passing through a distributional OWF.} 
    \label{fig:summary} 
\end{figure}

We note that the Search HPS assumption is a weaker assumption, in that the Decision HPS assumption implies the Search HPS assumption. Given that there are a wide variety of cryptographic primitives and protocols that can be built from PRSGs and OWSGs~\cite{microcryptzoo,goldin2024countcrypt}, the HPS assumptions provide a route for the \textit{concrete instantiation} of these primitives and protocols via a family of quantum states that can be easily prepared experimentally via \textit{instantaneous quantum polynomial-time} (IQP) circuits. Additionally, Ref.~\cite{bostanci2025efficientquantumpseudorandomnesshamiltonian} has provided some preliminary evidence that the HPS assumptions are \textit{independent of the existence of OWFs} -- i.e., the HPS assumptions could be true, even if OWFs do not exist. If correct, then the concrete protocols and primitives constructed via the HPS assumptions would indeed be in Microcrypt.

In this work, we establish that, contrary to prior evidence, the HPS assumptions are \textit{not} independent of OWFs. Specifically, in Section~\ref{s:HPS-imply-one-way}, we prove the following:

\begin{theorem}[One-way function from Hamiltonian phase state assumptions]\label{thm:HPS-imply-owf-intro} If the Search HPS assumption is true, then one can explicitly construct a quantum-secure one-way function.
\end{theorem}
Given that the Decision HPS assumption implies the Search HPS assumption, an immediate corollary of the above result is that one can also construct one-way functions from the Decision HPS assumption. Unfortunately, Theorem~\ref{thm:HPS-imply-owf-intro} shows that one cannot use the HPS assumptions to obtain cryptography in Microcrypt. However, given that we are able to explicitly construct a OWF from either HPS assumption, it also shows that the HPS assumptions provide a novel, inherently quantum assumption for the construction of classical cryptography. As such, Theorem~\ref{thm:HPS-imply-owf-intro} simultaneously provides an \textit{obstacle} for the concrete instantiation of Microcrypt, and an \textit{opportunity} for the instantiation of Minicrypt from novel inherently quantum assumptions. This complements recent work aimed at providing inherently quantum constructions for cryptography~\cite{YihuiLSN1,YihuiLSN2}, and opens up a wide variety of open questions and directions for future research, which we discuss in more detail in Section~\ref{ss:discussion-questions}.

At a high level, our proof of Theorem~\ref{thm:HPS-imply-owf-intro} is enabled by a new toolbox for constructing \textit{variants} of OWPs from OWSGs and copy efficient "measure first, ask later" state certification protocols for the set of states defining the OWSG, which is illustrated in Figure~\ref{fig:summary}. In particular, this method abstracts and generalizes an existing construction of OWPs from OWSGs and classical shadows~\cite{khurana2024commitmentsquantumonewayness}, and as illustrated in Figure~\ref{fig:summary}, allows us to relate properties of the resulting OWP to properties of the state certification protocol used in the construction. Importantly, when the state certification protocol admits classically efficient post-processing then the resulting OWP is \textit{efficiently verifiable}. If the state certification protocol is classically simulable in a certain sense, then the OWP can be used to construct a OWF. With this in hand, we prove Theorem~\ref{thm:HPS-imply-owf-intro} by showing that existing state certification protocols for phase states~\cite{HuangPreskillSoleimanifar2025} satisfy all the necessary efficiency and simulability properties for (a) the construction of a OWP from the OWSG obtained from the HPS assumptions, and (b) the construction of an explicit OWF from this OWP. We stress that while the prior construction of OWPs from OWSGs via classical shadows~\cite{khurana2024commitmentsquantumonewayness} allows one to obtain a OWP from the OWSG obtained from the Search HPS assumption, it is not clear if one can use this construction to then obtain a OWF from this OWP.

Given that our method for the construction of OWP variants from OWSGs and tailored state certification protocols provides a potential route for the instantiation of EV-OWPs (and therefore Quantumania), we discuss these techniques in more detail in Section~\ref{sss:contributions-towardsQCCC} below.

\subsubsection{Opportunities: Efficiently verifiable one-way puzzles via tailored state certification protocols}\label{sss:contributions-towardsQCCC}

As mentioned above, apart from allowing us to prove Theorem~\ref{thm:HPS-imply-owf-intro}, our technique for constructing OWP variants from OWSGs and state certification protocols, illustrated in Figure~\ref{fig:summary}, also provides a new toolbox for constructing \textit{efficiently-verifiable} OWPs, and therefore a potential route towards instantiating Quantumania. Central to this toolbox is the notion of a "measure first, ask later" state certification protocol for a set of states $\{|\phi_k\rangle\}$.

\begin{definition}[``Measure first, ask later'' state certification protocol for $\{ |\phi_k\rangle\}$]\label{def:measure-first-ask-later-intro}
Let $\{|\phi_k\rangle\,|\, k\in \mathbb{K}\}$ be a set of $n$-qubit states. Consider a QPT algorithm $\mathcal{M}(|\psi\rangle^{\otimes t})\rightarrow s\in\{0,1\}^{g(t,n)}$ and let $\mathcal{F}:\{0,1\}^*\times\{0,1\}^*\rightarrow \{\mathsf{Accept},\mathsf{Reject}\}$ be such that $\mathcal{F}(\tilde{k},\cdot) = \mathsf{reject}$ for all $\tilde{k}\notin \mathbb{K}$. We say that $(\mathcal{M},\mathcal{F})$ is a ``measure first, ask later'' state certification protocol for $\{|\phi_k\rangle\}$ from $t = t(n,\epsilon,\delta,|\mathbb{K}|)$ copies if, for all states $|\psi\rangle$ and $(\epsilon,\delta)\in (0,1)$:
\begin{equation}
\underset{s\gets \mathcal{M}(|\psi\rangle^{\otimes t})}{\mathrm{Pr}}\left[\forall \, k\in \mathbb{K} \, \begin{cases} \mathcal{F}(k,s)=\mathsf{Accept}\text{ if } |\psi\rangle = |\phi_k\rangle \\  \mathcal{F}(k,s)=\mathsf{Reject}\text{ if } |\langle \psi|\phi_k\rangle|^2 < 1-\epsilon \end{cases}\right] > 1-\delta.
\end{equation}
We say that $(\mathcal{M},\mathcal{F})$ is:
\begin{enumerate}
\item \textit{Copy efficient} if $t=\mathrm{poly}(n,\epsilon^{-1},\log \delta^{-1},\log |\mathbb{K}|)$ is sufficient.
\item \textit{Computationally efficient} if it is copy efficient and the function  $\mathcal{F}$ is computationally efficient. 
\item $\eta$-simulable if there exists a classical PPT algorithm $\mathcal{M}_C(k,1^n,1^t)\rightarrow s\in\{0,1\}^{g(t,n)}$ such that
\begin{equation}
d_\mathrm{TV}(\mathcal{M}_C(k,1^n,1^t),\mathcal{M}(|\phi_k\rangle^{\otimes t})) \leq \eta(|k|,n,t)
\end{equation}
for all $k\in\mathbb{K}$ and $n,t\in\mathbb{N}$, where again $\mathcal{M}_C(k,1^n,1^t)$ and $\mathcal{M}(|\phi_k\rangle^{\otimes t})$ are understood as distributions over $\{0,1\}^{g(t,n)}$.
\end{enumerate}
\end{definition}

Informally, $(\mathcal{M},\mathcal{F})$ is a ``measure first, ask later'' state certification protocol for the set of states $\{|\phi_k\rangle\,|\,k\in\mathbb{K}\rangle\}$ if, when $s\,\gets {\mathcal{M}}(\ket \psi^{\otimes t})$, the quantity $\mathcal{F}(k,s)$ allows one to decide whether $|\psi\rangle=|\phi_k\rangle$, for all $k\in \mathbb{K}$ and all states $|\psi\rangle$. We note that by exploiting existing classical shadow protocols for the set of observables $\{|\phi_k\rangle\langle \phi_k|\}$, one can immediately construct a \textit{copy efficient} ``measure first, ask later'' state certification protocol for \textit{any} set of states $\{|\phi_k\rangle\}$~\cite{Huang_2020}. However, we stress that this standard classical shadow based state certification protocol will typically be neither computationally efficient nor $\eta$-simulable for meaningful values of $\eta$. Indeed, the idea behind the definition above is to provide an abstraction of such "classical-shadow-type" state certification protocols, and to make clear the properties beyond copy-efficiency which, as we show below, will be useful for cryptographic constructions.

With this in hand, we can then define variants of \textit{certifiable} OWSGs in a natural way:

\begin{definition}[Certifiable one-way state generators (informal version of Definition~\ref{def:certifiable-microcrypt-primitives})] Given a one-way state generator $(\mathsf{KeyGen},\mathsf{StateGen},\mathsf{Ver})$ (defined formally in Definition~\ref{def:owsg}) with output states $\{|\phi_k\rangle\,|\,k\in \mathbb{K}\}$, and a ``measure first, ask later'' state certification protocol $(\mathcal{M},\mathcal{F})$,  we say that the tuple of one-way state generator and state certification protocol $((\mathsf{KeyGen},\mathsf{StateGen},\mathsf{Ver}),(\mathcal{M},\mathcal{F}))$ is:
\begin{enumerate}
\item A \textit{certifiable} one-way state generator if $(\mathcal{M},\mathcal{F})$ is copy efficient.
\item An \textit{$\eta$-simulable certifiable} one-way state generator if $(\mathcal{M},\mathcal{F})$ is certifiable and $\eta$-simulable.
\item An \textit{efficiently certifiable} one-way state generator if $(\mathcal{M},\mathcal{F})$ is computationally efficient. 
\item An \textit{$\eta$-simulable, efficiently certifiable} one-way state generator if $(\mathcal{M},\mathcal{F})$ is efficiently certifiable and $\eta$-simulable.
\end{enumerate}
\end{definition}
We note that one can easily give an analogous definition for certifiable PRSG (see Definition~\ref{def:certifiable-microcrypt-primitives}), and prove that standard output length certifiable PRSGs imply certifiable OWSGs (see Theorem~\ref{thm:cert-prsg-imply-cert-owsg}). Additionally, we stress again that \textit{any} OWSG is a certifiable OWSG when equipped with the ``measure first, ask later'' state certification protocol provided by standard classical shadows~\cite{Huang_2020}. However, this OWSG will typically fail to be either computationally efficient or $\eta$-simulable for meaningful values of $\eta$. With this established, we then show that certifiable OWSGs can be used to construct variants of OWPs. Specifically, we prove the following:

\begin{theorem}[One way puzzles from certifiable one-way state generators]\label{thm:owp-from-cowsg-intro} $\,$
\begin{enumerate}
\item Given a certifiable one-way state generator, one can construct a one-way puzzle (Theorem~\ref{thm:owp-from-certifiable-owsg}).
\item Given an efficiently certifiable one-way state generator, one can construct an efficiently verifiable one-way puzzle (Corollary~\ref{cor:ev-owp-from-eff-certifiable-owsg}).
\end{enumerate}
\end{theorem}

The constructions used to prove Theorem~\ref{thm:owp-from-cowsg-intro} generalize and abstract Khurana and Tomer's construction of OWPs via OWSGs and classical shadows~\cite{khurana2024commitmentsquantumonewayness}. Indeed, our contribution is to show that this prior construction works for \textit{any} copy-efficient ``measure first, ask later'' state certification protocol{\textemdash}not just the one derived from classical shadows{\textemdash}and that, if this protocol is in fact \textit{computationally efficient}, then the resulting OWP is efficiently verifiable. This is interesting because it allows us to obtain \textit{variants} of OWPs by using state certification protocols \textit{tailored} to the set of states defining the OWSG. Indeed, our hope is that this can be used as a tool for translating progress in state certification into progress on the proposal of explicit candidates for \textit{efficiently verifiable} OWPs -- i.e.,  for the instantiation of Quantumania. As discussed in Section~\ref{ss:related-work} below, there has recently been significant progress in the development of state certification protocols~\cite{HuangPreskillSoleimanifar2025,gupta2025singlequbitmeasurementssufficecertify,Li_2026}, and Theorem~\ref{thm:owp-from-cowsg-intro} shows that if such progress yields computationally efficient ``measure first, ask later'' state certification protocols for the output states of existing candidate OWSGs (or PRSGs), then one in fact obtains a candidate \textit{efficiently verifiable} OWP.

There is however an important caveat! As we show in Theorem~\ref{thm:owf-from-sim-cert-owsg-intro} below, if the certifiable OWSG has a classical key generation algorithm, and is also $1/3$-simulable then the resulting OWP can be used to construct a (quantum-secure) OWF. As such, the OWP is by definition not in Microcrypt! 
\begin{theorem}[OWFs from $1/3$-simulable certifiable OWSGs (informal version of Theorem~\ref{thm:owf-via-sim-certOWSG})]\label{thm:owf-from-sim-cert-owsg-intro}
Given a $\frac{1}{3}$-simulable certifiable one-way state generator, with a classical key generation algorithm, one can construct a quantum-secure one-way function.
\end{theorem}

We prove Theorem~\ref{thm:owf-from-sim-cert-owsg-intro} by showing that $1/3$-simulability of the certifiable OWSG implies the existence of an efficient approximate \textit{classical} algorithm for simulating the \textit{quantum} key/puzzle sampling algorithm of the OWP resulting from the construction used to prove Theorem~\ref{thm:owp-from-cowsg-intro}. Given this, it follows from existing results that
\begin{enumerate}
\item One can construct a quantum-secure \textit{distributional} OWF from the OWP~\cite{khurana2025founding}.
\item One can then construct a quantum-secure \textit{weak} OWF from the distributional OWF~\cite{impagliazzo1989one,kashefi2007statisticalzeroknowledgequantum}, from which one can construct a OWF~\cite{goldreich2001foundations,Radian_2019}.
\end{enumerate}

Taken together, we see that Theorem~\ref{thm:owp-from-cowsg-intro} and Theorem~\ref{thm:owf-from-sim-cert-owsg-intro} provide a new opportunity for the construction of EV-OWPs which are genuinely in Microcrypt, via any OWSG which:
\begin{enumerate}
\item Is itself not constructed from OWFs (i.e.,  relies on a purely quantum assumption which is plausibly independent of OWFs).
\item Admits a computationally efficient ``measure first, ask later'' state certification protocol, which is \textit{not} also $1/3$-simulable.
\end{enumerate}
In Section~\ref{s:towards-concrete-instantiations} we discuss a variety of candidate OWSGs and the extent to which they may or may not satisfy the criteria above. However, as already mentioned in Section~\ref{sss:contributions-HPSimplyowf}, we also stress that when a OWSG fails to provide a OWP that is genuinely in Microcrypt by virtue of being simulable, this implies that the quantum assumptions used for the OWSG provide new assumptions for \textit{classical} cryptography, by virtue of the fact that the OWSG can be used to construct OWFs. As mentioned before, the obstacle for instantiating Microcrypt becomes an opportunity for instantiating Minicrypt from novel inherently quantum assumptions.

\subsubsection{Towards simpler one-way functions from the Search HPS assumption}\label{ss:towards-simpler}

Given Theorem~\ref{thm:owf-from-sim-cert-owsg-intro},  it is clear that to prove Theorem~\ref{thm:HPS-imply-owf-intro} it is sufficient to provide a "measure first, ask later" state certification protocol for Hamiltonian phase states, which is both copy efficient \textit{and} $1/3$-simulable. However, as we discussed in the previous section, this allows one to directly construct a distributional OWF, which can be compiled into a weak OWF, which can then be compiled into a OWF, with each step adding complexity to the construction of the OWF.  With the goal of providing a simpler concrete OWF construction from the Search HPS assumption, we show that given an $\eta$-simulable \textit{efficiently} certifiable OWSG, for an $\eta$ which is negligible in a specific sense, then one can directly construct a simple quantum-secure OWF, without going via distributional OWFs. In other words, we leverage both the increased accuracy of the measurement simulation algorithm, and the computational efficiency of the post-processing functions, to obtain a simpler OWF. 

\begin{theorem}[Simple OWFs from simulable and efficiently certifiable OWSGs (informal version of Theorem~\ref{thm:owf-via-sim-effcertOWSG})]\label{thm:owf-from-sim-cert-owsg-direct-intro}
Given an $\eta$-simulable efficiently certifiable one-way state generator with classical key generation algorithm, one can directly construct a quantum-secure one-way function, without going via a distributional one-way function, whenever $\eta$ is such that $\eta(|k|,n,t)$ is negligible with respect to $n$, for all $|k|$ and $t$ that are at most polynomial in $n$.
\end{theorem}
With the above in mind, to prove Theorem~\ref{thm:HPS-imply-owf-intro} we don't simply prove the sufficient statement that there exists a copy efficient \textit{and} $1/3$-simulable state certification protocol for Hamiltonian phase states. Instead, we prove the stronger statement that the existing state certification protocol from Ref.~\cite{HuangPreskillSoleimanifar2025} is in fact an $\eta$-simulable computationally efficient state certification protocol, for negligible $\eta$ in the appropriate sense, which allows us to invoke Theorem~\ref{thm:owf-from-sim-cert-owsg-direct-intro} to construct a OWF from the Search HPS assumption.

\subsection{Related work}\label{ss:related-work}

Below we describe a variety of existing and active research directions, which intersect with the goals and contributions of this work.

\textbf{Microcrypt:} Our work  contributes to the rapidly growing literature on understanding and characterizing Microcrypt. As this area has become too extensive to survey completely here, we refer the reader to the \textit{Microcrypt Zoo}~~\cite{microcryptzoo} for a broad overview. Nevertheless, we highlight the following points to place our work in context:
\begin{enumerate}
\item Over the past years, a wide variety of inherently quantum cryptographic primitives -- such as EFI pairs~\cite{BCQ2023}, variants of PRSGs~\cite{PRSGdefinition, brakerski2020scalablepseudorandomquantumstates}, pseudorandom unitaries~\cite{ma2025constructrandomunitaries}, OWSGs~\cite{Morimae_2022}, and (efficiently verifiable) OWPs~\cite{khurana2024commitmentsquantumonewayness, QCCCCrypto,PKE-deletion} amongst others -- have been proposed and studied. Our work introduces two new primitives, namely certifiable and efficiently certifiable OWSGs, and shows their utility for the construction of \textit{variants} of OWPs via state certification protocols tailored to the underlying OWSG. Apart from allowing us to prove Theorem~\ref{thm:HPS-imply-owf-intro}, this is particularly interesting as it provides a new route for the construction of \textit{efficiently verifiable} OWPs, the minimal primitive of Quantumania~\cite{goldin2024countcrypt}. Previous work has shown that EV-OWPs can be constructed from quantum pseudorandom generators~\cite{,QCCCCrypto} (which can be constructed from \textit{logarithmic output length} PRSGs~\cite{ananth2023pseudorandomstringspseudorandomquantum}, which in turn can be constructed from quantum computable pseudorandom functions~\cite{Kretschmer_2025}) as well as from a variety of (non-interactive) QCCC primitives~\cite{QCCCCrypto}, and our work adds a new method to this toolbox.

\item In order to actually implement any cryptographic protocol, one requires a \textit{concrete instantiation} of the primitive on which the protocol is built. To this end, there has been a wide variety of work aimed at providing candidate concrete instantiations of Microcrypt primitives, from a wide variety of alternative assumptions~\cite{metger2024simpleconstructionslineardepthtdesigns, brakerski2019pseudorandomquantumstates, Chen_2024, bostanci2025efficientquantumpseudorandomnesshamiltonian,fefferman2025hardnesslearningquantumcircuits}. However, we note that many of these constructions assume the existence of (quantum-secure) OWFs as a starting point! As Microcrypt is particularly interesting due to its potential existence even if OWFs do not exist, one ideally wants candidate constructions from assumptions which are potentially independent of OWFs. Prior to our work, the Hamiltonian phase state assumptions~\cite{bostanci2025efficientquantumpseudorandomnesshamiltonian} were precisely such assumptions, and one of our main contributions is to show that these assumptions are in fact \textit{not} independent of the existence of OWFs, as originally hoped.
\end{enumerate}

\textbf{Cryptography from (quantum) hardness of learning:} There is a long line of work on both learning theoretic characterizations of cryptographic primitives, as well as concrete instantiations of cryptographic protocols and primitives from \textit{hardness of learning} assumptions -- such as learning parities with noise, and learning with errors~\cite{regevOnLattices2009}.  Inspired by this, and motivated by the growing number of inherently quantum cryptographic protocols and primitives, recent work has attempted to understand the extent to which inherently \textit{quantum} learning theoretic assumptions can be used to both characterize and instantiate (quantum) cryptography. Indeed, the Hamiltonian phase state assumptions~\cite{bostanci2025efficientquantumpseudorandomnesshamiltonian}, are precisely an assumption on the hardness of learning Hamiltonian phase states (or equivalently, random IQP circuits) in a specific sense, and one of our primary contributions is to show that while these assumptions \textit{cannot} be used to instantiate genuine Microcrypt primitives, they are sufficient for the construction of OWFs, and therefore provide a means for the concrete instantiation of Minicrypt (at least).

However, a wide variety of other quantum learning theoretic assumptions have also been recently proposed and explored in the context of quantum cryptography. Similar in spirit to the HPS assumptions are the \textit{computational no-learning} and \textit{computational no-cloning} assumptions introduced in Ref.~\cite{fefferman2025hardnesslearningquantumcircuits}, which posit the hardness of learning and cloning the output states of sufficiently deep \textit{brickwork} random quantum circuits from copies of the output state of the circuit. Specifically, the authors of Ref.~\cite{fefferman2025hardnesslearningquantumcircuits} have shown that this assumption is sufficient for the construction of OWSGs, quantum commitments and digital signatures. We note that these assumptions differ from the HPS assumptions central to this work in both the families of circuits considered, and the details of the learning task which is assumed to be hard. While not strictly learning theoretic, we also mention that Ref.~\cite{khurana2025founding} has also explored the cryptographic potential of random quantum circuits, by showing that the standard ``quantum advantage assumptions'' -- namely, the $\#\mathsf{P}$-hardness of estimating the output probabilities of a random quantum circuit from its classical description~\cite{Hangleiter_2023} -- is sufficient for the construction of OWPs. 

Moving away from random quantum \textit{circuits}, Refs.~\cite{hiroka2024computationalcomplexitylearningefficiently,cojocaru2026equivalenceaveragecasehardnesslearning} have given characterizations of both EFI pairs and OWSGs in terms of the average-case complexity of learning (efficiently generatable) random quantum \textit{states}. Additionally, Ref.~\cite{hiroka2025hardnessquantumdistributionlearning} has provided a characterization of OWP in terms of the average-case hardness of proper quantum \textit{distribution} learning. Ref.~\cite{Arapinis_2021} has also studied the possibility of building quantum cryptography from \textit{hardware} assumptions via quantum Physically Unclonable Functions. Finally, we mention the recent series of works~\cite{YihuiLSN1,YihuiLSN2,KhesinLSN3} that have proposed and studied the \textit{learning stabilizers with noise} assumption -- which at a high level posits the average-case hardness of decoding random quantum stabilizer codes and generalizes the well studied learning parities with noise assumption -- and shown that this assumption is sufficient for instantiating Cryptomania -- i.e., this inherently quantum assumption is useful for \textit{classical} cryptography. Indeed, our proof that the HPS assumptions imply OWFs, place the HPS assumptions on a similar footing to that of the learning stabilizers with noise assumptions.

\textbf{State certification:} State certification protocols are central to the results and contributions of this work. Indeed, our hope is that by using the notion of certifiable OWSGs as a tool, the community can more easily translate progress on quantum state certification into concrete instantiations of Microcrypt primitives. To aid with this, its helpful to put (significant) recent progress on state certification ~\cite{HuangPreskillSoleimanifar2025,gupta2025singlequbitmeasurementssufficecertify,coladangelo2026powerbasesrobustcopyoptimal,coladangelo2026robustquantumstatecertification} into context. In particular, the following is a list of criteria with respect to which different state certification protocols are often compared and contrasted (adapted and extended from the comparitive table in Figure 1 of Ref.~\cite{coladangelo2026robustquantumstatecertification}), together with a brief discussion of the relevance of a given criterion for our applications here:

\begin{enumerate}
\item Set of states: The set of states for which the state certification protocol is guaranteed to work for. In our setting, we require the protocol to work for the output states of the candidate OWSG we would like to use as input to our constructions.
\item Target dependence of measurements: Whether or not the protocol only requires \textit{target-independent} measurements and target-state dependent classical post-processing, as opposed to target-dependent measurement protocols. In the latter case, one can also distinguish between adaptive and non-adaptive measurement strategies. Borrowing from Ref.~\cite{elben2023randomized} we have called the former protocols "measure first, ask later".  Here, we require "measure first, ask later" state certification protocols.
\item Copy complexity: The number of measurements required in the worst case. We only require polynomial copy complexity.
\item Measurement complexity: The complexity of the largest measurement required. Again, we only require efficient measurements.
\item Oracle model: The type of oracle access to the target state which is assumed. For example, many protocols assume the ability to query specific amplitudes of the target state in the computational and/or Hadamard basis. For the construction of \textit{certifiable} OWSGs (and therefore OWP), the oracle model is not relevant to us. However, for the construction of \textit{efficiently certifiable} OWSGs (and therefore \textit{efficiently verifiable} OWPs) we require that the oracle can be \textit{efficiently simulated} classically, for all target states of interest.
\item Computational complexity: Combined time complexity of the entire protocol, including any classical post-processing of measurement outcomes, assuming oracle access to the target state. As above, for the construction of certifiable OWSGs this is not relevant for us. However, for the construction of efficiently certifiable OWSGs, we require protocols which are computationally efficient \textit{including} the simulation of the oracle model. 
\item Robustness: The extent to which the state certification protocol can be made \textit{tolerant}. For our purposes even \textit{non-tolerant} state certification protocols suffice.
\end{enumerate}
With the above in mind, we note that the paradigmatic "measure first, ask later" state certification protocol is that of classical shadows and its variants~\cite{Huang_2020,HelsenWalter2023,GrewalIngram2024,KohGrewal2022,ZhaoRubinMiyake2021,WanHugginsLeeBabbush2023, BertoniEtAl2024ShallowShadows,Conrad_2026,Becker_2024}. Indeed, global Clifford shadows provide a state certification protocol which works for \textit{all} target states, and requires only polynomially many efficient measurements~\cite{Huang_2020} -- which is precisely why Khurana and Tomer are able to use such classical shadows to construct OWPs from \textit{any} OWSG~\cite{khurana2024commitmentsquantumonewayness}. Unfortunately however, global Clifford shadows are \textit{not} computationally efficient for arbitrary target states, and therefore cannot be used out of the box to obtain \textit{efficiently verifiable} OWPs from arbitrary OWSGs.  While variants of classical shadows can be made computationally efficient for some sets of structured states (such as stabilizer states~\cite{Huang_2020}, fermionic Gaussian states~\cite{ZhaoRubinMiyake2021,WanHugginsLeeBabbush2023} and bosonic Gaussian states~\cite{Becker_2024}) it is not clear whether such states can be used for the construction of OWSGs.

With the view to identifying state certification protocols which satisfy \textit{all} the criteria for the construction of an \textit{efficiently certifiable} OWSG (especially for the HPS derived OWSG we are particularly concerned with) we note that a significant amount of recent work ~\cite{HuangPreskillSoleimanifar2025,gupta2025singlequbitmeasurementssufficecertify,coladangelo2026powerbasesrobustcopyoptimal,coladangelo2026robustquantumstatecertification} has gone into developing novel state certification protocols, which try to simultaneously optimize as many as the above criteria as possible (we refer to Ref.~\cite{coladangelo2026robustquantumstatecertification} for a detailed comparison). These protocols are also computational efficient for any set of target states for which the assumed oracle model can be efficiently simulated. As a result, at least for the OWSG we are primarily concerned with here, which has Hamiltonian phase states as output states, the only protocol which is both "measure first, ask later" \textit{and} admits computationally efficient post-processing is that of Ref.~\cite{HuangPreskillSoleimanifar2025} (as a result of the fact that the oracle which is assumed can be simulated efficiently for phase states). While the measurement complexity, copy complexity and robustness of this protocol is not as good as the more recent state certification protocols developed in~\cite{,coladangelo2026powerbasesrobustcopyoptimal,coladangelo2026robustquantumstatecertification} this is not a concern for us, given that we do not require any robustness, and polynomial copy and measurement complexity suffices. Finally, we note that a state certification protocol which satisifes all our desired criterion for phase states is also developed implicitly in Ref.~\cite{park2026samplehardwareefficientfidelityestimation}, however for simplicity of presentation we focus on the protocol from Ref.~\cite{HuangPreskillSoleimanifar2025}.

\subsection{Discussion and open questions}\label{ss:discussion-questions}

We highlight the following open questions and directions arising from our contributions here:

\textbf{Concrete instantiations of efficiently certifiable OWSGs:} One of the primary contributions of this work is to provide a concrete set of sufficient conditions for a state ensemble, in terms of its learnability and certifiability, for the construction of efficiently certifiable OWSGs, and therefore for efficiently verifiable OWPs. Given this, perhaps the most natural open question is whether one can identify explicit ensembles of quantum states which are (a) pseudorandom under a plausible conjecture, and (b) admit a computationally efficient ``measure first, ask later'' state certification protocol which is \textit{not} also simulable. Of course, we are particularly interested in state ensembles which are plausibly pseudorandom \textit{even if} OWFs do not exist -- i.e.,  state ensembles which are not constructed assuming the existence of a OWF, and whose pseudorandomness would not imply OWFs. Before this work, Hamiltonian phase states were precisely such a candidate ensemble, and our hope is that by understanding precisely the requirements in terms of certifiably, new candidate ensembles for the construction of EV-OWPs can be identified. We discuss a variety of existing candidates and their shortcomings in Section~\ref{s:towards-concrete-instantiations}.

\textbf{QCCC cryptography via efficiently certifiable OWSGs:} Quantumania is particularly interesting because it contains a wide variety of QCCC cryptographic protocols and primitives~\cite{QCCCCrypto,Kretschmer_2025}. With this in mind, are there QCCC cryptographic protocols which can be constructed directly from efficiently certifiable OWSGs? Could one construct interactive QCCC bit commitments~\cite{ananth2023pseudorandomstringspseudorandomquantum}, or QCCC public-key encryption, directly from efficiently certifiable OWSGs? Said simply, can we show that the primitives that we introduce in this work are \textit{useful}?

\textbf{Separations and implications between Quantumania primitives:} Refs.~\cite{Behera_2025,owsgevowpseperations} proved the existence of an oracle world in which OWPs exist, but OWSGs do not exist, showing that unlike classical Minicrypt, the world of Countcrypt \textit{cannot} be collapsed to a single minimal primitive~\cite{QCCCCrypto}. Does Quantumania also have "subworlds"? More specifically, can one prove an analogous black box separation between \textit{efficiently certifiable} OWSGs and \textit{efficiently verifiable} OWPs? More generally, how are efficiently certifiable OWSGs related to other Quantumania primitives such as quantum computable OWFs~\cite{Kretschmer_2025}, logarithmic depth output PRSGs, and quantum pseudorandom generators~\cite{ananth2023pseudorandomstringspseudorandomquantum}? 

\textbf{Refinements of certifiable primitives:} One can straightforwardly adapt our definition of a ``measure first, ask later'' state certification protocol (Definition~\ref{def:measure-first-ask-later-intro}) to define constrained notions such as a single-copy or non-adaptive ``measure first, ask later'' state certification protocols. This subsequently allows us to define single-copy and non-adaptive versions of certifiable PRSGs and OWSGs. Intuitively, since single-copy and non-adaptive state certification protocols for a given set of states are \textit{harder} to construct than their multi-copy and adaptive counterparts, it may be the case that single-copy and non-adaptive versions of the efficiently certifiable primitives we have introduced here are \textit{more powerful} than the unconstrained versions which allow multi-copy and adaptive measurement protocols for certification. Said another way, it is by now well established that in the context of learning and testing, multi-copy and adaptive measurements provide a powerful resource which enables exponential separations (see e.g., Ref.~\cite{chen2022exponential}). Do such separations manifest cryptographically? Can one construct cryptographic protocols using single-copy/non-adaptive certifiable primitives that one cannot using the unconstrained certifiable primitives we define here? 

\textbf{Evidence for efficiently-certifiable OWSGs without OWFs:} The major motivation for studying many Microcrypt primitives comes from evidence that they can exist, even if $\mathsf{P}=\mathsf{NP}$ and OWFs do not exist~\cite{Kretschmer_2021,Kretschmer_2023,Kretschmer_2025}. Can we provide similar evidence that efficiently certifiable OWSGs can exist even if OWFs do not exist? From Ref.~\cite{Kretschmer_2025} we know that there exists an oracle world in which $\mathsf{P}=\mathsf{NP}$  but quantum computable OWFs, and therefore also logarithmic output length PRSGs, quantum pseudorandom generators and EV-OWPs exist. As such, one may be able to answer this question by better understanding the implications between existing Quantumania primitives, as per the previous question.

\textbf{Classical cryptography from quantum assumptions:} Given the fact that the HPS assumptions are sufficient for the construction of an explicit OWF (Theorem~\ref{thm:HPS-imply-owf-intro}), one can instantiate all cryptography in \textit{Minicrypt} under the HPS assumption. An immediate natural question is whether one could in fact instantiate cryptographic primitives and protocols in \textit{Cryptomania} under the HPS assumptions? To this end, an immediate direction would be to try to construct suitable \textit{trapdoor} OWFs from the HPS assumptions. Additionally, our results show that one can construct OWFs from any OWSG which admits a simulable and computationally efficient state certification protocol. Could one construct other OWSGs under inherently quantum assumptions, which also admit such state certification protocols, and can therefore be used to construct OWFs under inherently quantum assumptions?

\textbf{Independence and relation of HPS assumptions from classical assumptions:} As has been mentioned before, one of the primary reasons that Theorem~\ref{thm:HPS-imply-owf-intro} is interesting, is because it shows that the HPS assumptions provide a novel \textit{inherently quantum} assumption for the construction of classical cryptography. However, unlike existing and well studied assumptions like "Learning with Errors"~\cite{regevOnLattices2009} the HPS assumptions have not been tested in any way, and their relation to existing assumptions is largely unclear. With this in mind, if one is to take the HPS assumptions seriously as a foundation for classical cryptography, then significantly more effort is needed to investigate both the plausibility of these assumptions, and their potential relation to existing known assumptions.

\textbf{Quantumania via state certification protocols with quantum post-processing:} We note that our definition of a computationally efficient "measure first, ask later" state certification protocol implicitly requires the post-processing functions to be classically efficient to compute. This is done to ensure that the constructions using such protocols lead to EV-OWPs as normally understood -- i.e. OWPs with classically efficient verification algorithms. However, if one was to allow for efficient deterministic \textit{quantum} verification algorithms in the definition of an EV-OWP, then in principle any copy-efficient state certification protocol with efficient \textit{quantum} postprocessing algorithms would suffice. We believe it is an interesting direction to understand the power of such "measure first, ask later" state certification protocols with \textit{quantum} post-processing functions, and their relation to quantum computable OWFs and the instantiation of concrete QCCC cryptographic protocols like one-time digital signatures or interactive quantum bit commitments.

\subsection{A short story on the highs and lows of developing this work}\label{ss:historical-context}

This work grew out of the observation that one could construct \textit{efficiently verifiable} OWPs from OWSGs, if one had a computationally efficient ``measure first, ask later'' state certification protocol for the output states of the OWSG. With this in mind, we went looking for such a state certification protocol for Hamiltonian phase states, knowing that they were conjectured to provide a OWSG (plausibly independent of OWFs), with the hope of providing a concrete candidate instantiation of an EV-OWP, and therefore of a QCCC 
one-time digital signature, which was plausibly genuinely in Microcrypt. In particular, a recent experimental work implemented an "almost" QCCC digital signature~\cite{niroula2026digitalsignaturesclassicalshadows}, which was just lacking efficient verifiability, by virtue of using standard classical shadows. We know that if we could find a computationally efficient ``measure first, ask later'' state certification protocol tailored for Hamiltonian phase states, this would give us an explicit EV-OWP (under the HPS assumptions), from which we would be able to construct a genuine and experimentally feasible QCCC one-time digital signature, answering an open question from Ref.~\cite{fefferman2025hardnesslearningquantumcircuits}.

We then found the state certification protocol for (Hamiltonian) phase states from Ref.~\cite{park2026samplehardwareefficientfidelityestimation} and happily put together the pieces, writing a draft whose main claim was the construction of an explicit candidate for an EV-OWP, and therefore experimentally feasible QCCC one-time digital signature, which was plausibly independent of OWFs, under the HPS assumptions. In other words, a plausible concrete and explicit instantiation of Quantumania! 

After uploading the manuscript to the arXiv, and sharing the completed draft among friends and colleagues, Matthias Caro asked us why we hadn't used the state certification protocol from Ref.~\cite{HuangPreskillSoleimanifar2025}? Quickly we realized that one \textit{could} of course also use that state certification protocol, and that if we did this, the resulting OWP sampling algorithm would be \textit{classically} simulable, and imply OWFs. In essence, we \textit{did} have a construction of an efficiently verifiable OWP under the HPS assumptions, but one could construct a OWF from this OWP, and therefore \textit{the HPS assumptions implied OWFs}. Nothing we had written in our original draft was wrong, except the foundational assumption! We pulled the paper from the arXiv before it was announced (luckily), and the result is the work you are now reading.

\subsection{Structure of this work}\label{ss:structure}

 We start in Section~\ref{s:preliminaries} by introducing all relevant notation, existing definitions and assumptions. Following this, we provide in Section~\ref{ss:state certification-prelim} formal definitions for the "measure first, ask later" state certification protocols which are central to this work. Building on this, we then define a variety of \textit{certifiable} Microcrypt primitives in Section~\ref{s:certifiable-microcrypt}. With this established we then prove in Section~\ref{s:owp-from-cert-owsg} that one can build variants of OWPs from certifiable OWSGs. We then show in Section~\ref{s:owf-from-simulable-owsg} that if a OWSG is \textit{simulable} in addition to being certifiable, then one can construct variants of OWFs from the OWPs constructed from the OWSG. Given this, we can then finally prove Theorem~\ref{thm:HPS-imply-owf-intro} in Section~\ref{s:HPS-imply-one-way}. We conclude in Section~\ref{s:towards-concrete-instantiations} with a discussion of candidate state ensembles for instantiating Quantumania via the tools developed in this work.

 \subsection*{AI Disclosure} The first version of this manuscript, containing all the essential ideas, constructions, proofs and method of presentation, was obtained without the assistance of generative AI. Both Claude Opus 4.8 (Max) and ChatGPT 5.6 Sol (Ultra) were then used to generate a detailed and critical referee reports, with a focus on verifying mathematical correctness. These report identified a variety of subtle mathematical issues, which were fixed via interaction with both models. This materially affected details in the proofs of Corollary~\ref{cor:classical-sampling-implies-qsecDOWF}, Theorem~\ref{thm:owp-from-certifiable-owsg}, Theorem~\ref{thm:owf-via-sim-certOWSG} (as it relies on Corollary~\ref{cor:classical-sampling-implies-qsecDOWF}), Theorem~\ref{thm:owf-via-sim-effcertOWSG} and Lemma~\ref{lem:eff-approx-Bernoulli}. The authors verified the correctness and originality of all content including references.

 \subsection*{Acknowledgements} We are grateful for helpful and inspiring conversations with Carlos Cid, Janek Denzler, David Elkouss, Bill Fefferman, Elies Gil-Fuster, Tommaso Guaita, Manuel Goul\~{a}o, Jonas Haferkamp, Zephrina Aniska Be, Mina Doosti and Martina Onetti. JC thanks the support of Berlin Quantum. NKHL acknowledges support from  the Austrian Science Fund (FWF) (\href{https://doi.org/10.55776/P36478}{10.55776/P36478}), the Austrian Federal Ministry of Education, Science and Research via the Austrian Research Promotion Agency (FFG) through the flagship project FO999921415 (Vanessa-QC) funded by the European Union{\textemdash}NextGenerationEU, the European Research Council (Consolidator grant `Cocoquest' 101043705), and the Croucher Foundation. RS thanks the Alexander von Humboldt foundation for their support, under the German Research Chair program at the African Institutes for Mathematical Sciences. JE has been funded by the BMFTR (QuSol, Hybrid++), Berlin Quantum, the Munich Quantum Valley, the Quantum Flagship
(Millenion, PasQuans2), 
the DFG (CRC 183, SPP 2514), the Clusters of Excellence (ML4Q, MATH+),
and the European Research Council (DebuQC). Part of this work was carried out while JC and JE visited the African Institute for Mathematical Sciences (AIMS), Cape Town, for
the ``1st AIMS Workshop on the Theory of Quantum Learning Algorithms'' (2025).
DH was supported by
a Simons postdoctoral fellowship through DOE QSA and NSF QLCI Grant No. 2016245, and
by the Swiss National Science Foundation through Ambizione Grant No. 223764.

\section{Preliminaries}\label{s:preliminaries}

Throughout this work we use the notation $x\gets \mathcal{X}$ to denote $x$ sampled uniformly from the set $\mathcal{X}$, and the notation $x\gets D$ to denote $x$ sampled from a distribution $D$.  Given some function $f:\mathcal{X}\rightarrow\mathcal{Y}$ we will use the notation $(f(x))_{x\gets \mathcal{X}}$ to denote the distribution over $\mathcal{Y}$ which one samples from by first sampling $x\gets\mathcal{X}$ and then outputting $f(x)$. Additionally, we use the notation $\mathsf{negl}$ to represent a negligible function -- i.e. any function $f:\mathbb{N}\rightarrow\mathbb{R}$ such that for every positive polynomial $p$ there exists an $N$ such that for all $n\geq N$ one has $f(n) < 1/p(n)$. We use the abbreviations QPT and PPT for ``quantum polynomial time'' and ``probabilistic polynomial time'' respectively.

\subsection{Classical cryptographic primitives}\label{ss:classical-primitives}

We begin by defining the classical cryptographic primitives relevant to this work. The first such primitives are (quantum-secure) one-way functions.

\begin{definition}[One-way function (OWF)]\label{def:one-way-function} Let $m:\mathbb{N}\rightarrow\mathbb{N}$ be some fixed polynomial. A family of efficiently computable functions $\{F_\lambda:\{0,1\}^{m(\lambda)}\rightarrow\{0,1\}^*\}_{\lambda\in\mathbb{N}}$ is:
\begin{enumerate}
\item A (quantum-secure) \textit{one-way function} if for all PPT (QPT) algorithms $\mathcal{A}$ there exists a negligible function $\mathsf{negl}$ such that
\begin{equation}
\underset{\substack{x\gets\{0,1\}^{m(\lambda)} \\y\gets F_\lambda(x)}}{\mathrm{Pr}}\left[\mathcal{A}(1^\lambda,y) \in F_\lambda^{-1}(y)\right] \leq \mathsf{negl}(\lambda)
\end{equation}
for all sufficiently large $\lambda\in \mathbb{N}$. 
\item A (quantum-secure)  \textit{weak one-way function} if there exists a polynomial $p:\mathbb{N}\rightarrow\mathbb{N}$ such that for all (QPT) PPT algorithms $\mathcal{A}$
\begin{equation}
\underset{\substack{x\gets\{0,1\}^{m(\lambda)} \\y\gets F_\lambda(x)}}{\mathrm{Pr}}\left[\mathcal{A}(1^\lambda,y) \notin F_\lambda^{-1}(y)\right] > \frac{1}{p(\lambda)}
\end{equation}
for all sufficiently large $\lambda\in \mathbb{N}$. 
\end{enumerate}
\end{definition}

In order to break a candidate OWF $\{F_\lambda\}$, on input $F_\lambda(x)$ the adversary needs to output a \textit{single} element of the preimage of $F_\lambda(x)$. Given a weak OWF, there exists a polynomial $p$ such that for all adversaries the failure probability is at least $1/p$ with respect to $x$ drawn uniformly randomly. Given a OWF, the success probability of any adversary is negligible. We note that weak OWFs imply OWFs~\cite{goldreich2001foundations} and that quantum-secure weak OWFs imply quantum-secure OWFs~\cite{kashefi2007statisticalzeroknowledgequantum,Radian_2019}.

The second classical cryptographic primitive relevant to this work is a \textit{distributional} one-way function, for which, on input $F_\lambda(x)$, the adversary has the harder task of sampling from the uniform distribution over the preimage of $F_\lambda(x)$. A function family is a distributional one-way function if there is an inverse polynomial lower bound on the accuracy achievable by any PPT adversary:

\begin{definition}[Distributional one-way function (D-OWF) \cite{impagliazzo1989one}]\label{def:dist-one-way-function} Let $m:\mathbb{N}\rightarrow\mathbb{N}$ be some fixed polynomial. A family of efficiently computable functions $\{F_\lambda:\{0,1\}^{m(\lambda)}\rightarrow\{0,1\}^*\}_{\lambda\in\mathbb{N}}$ is a \textit{distributional one-way function} if there exists a polynomial $p$ such that for all PPT algorithms $\mathcal{A}$,  one has
\begin{equation}
d_\mathrm{TV}\left((x,F_\lambda(x))_{x\gets\{0,1\}^{m(\lambda)}}, (\mathcal{A}(1^\lambda,F_{\lambda}(x)),F_\lambda(x))_{x\gets\{0,1\}^{m(\lambda)}}\right) \geq \frac{1}{p(\lambda)}
\end{equation}
all sufficiently large $\lambda\in \mathbb{N}$.
\end{definition}

In the above definition, only classical PPT adversaries have been considered. To define a \textit{quantum-secure} distributional one-way function, one insists that, on average with respect to function inputs, no QPT adversary can prepare a state close to a uniform superposition over preimages. To make this precise, given any two quantum states $\rho,\sigma$ let's denote by $F_{\mathrm{sq}}(\rho,\sigma) = \|\sqrt{\rho}\sqrt{\sigma}\|_1^2$ the squared fidelity between the two states. We then have: 
\begin{definition}[Quantum-secure distributional one-way function (Adapted from~\cite{kashefi2007statisticalzeroknowledgequantum})]\label{def:qsecDOWF} Let $m:\mathbb{N}\rightarrow\mathbb{N}$ be some fixed polynomial. A family of efficiently computable functions $\{f_\lambda:\{0,1\}^{m(\lambda)}\rightarrow\{0,1\}^*\}_{\lambda\in\mathbb{N}}$ is a quantum-secure \textit{distributional one-way function} if there exists a polynomial $p$ such that for all QPT algorithms $\mathcal{A}$ and all sufficiently large $\lambda$ one has that
\begin{align}
\overline{F}_{\mathcal{A}}(\lambda) \coloneqq  \frac{1}{2^{m(\lambda)}}\sum_{r\in\{0,1\}^{m(\lambda)}}  F_\mathrm{sq}\left(\sigma_{f_\lambda(r)}, \rho^{\mathcal{A}}_{f_\lambda(r)}\right) \leq 1-\frac{1}{p(\lambda)}
\end{align}
where  $\sigma_{y} =  |H_y\rangle\langle H_y|$ and $\rho^{\mathcal{A}}_{y} = \mathcal{A}(1^\lambda,|y\rangle\langle y|)$ with
\begin{equation}
|H_{y}\rangle = \frac{1}{\sqrt{|f_\lambda^{-1}(y)|}}\sum_{z\in f_\lambda^{-1}(y)}|z\rangle.
\end{equation}
\end{definition}

We note that the definition we have given above, where the adversary can output mixed quantum states, is a strengthening of the definition given in~\cite{kashefi2007statisticalzeroknowledgequantum} and implies their definition. As sampling from the preimage of $F_\lambda(x)$ is a harder task than outputting a single element of the preimage of $F_\lambda(x)$, any (quantum-secure) OWF is immediately also a (quantum-secure) distributional OWF. The other direction however does not hold. In particular, it is possible for a family of functions $\{F_\lambda\}$ to be a distributional OWF but \textit{not} a OWF. Perhaps surprisingly, however, Impagliazzo and Luby showed that given a distributional OWF, one can construct a OWF~\cite{impagliazzo1989one}. Particularly, we have the following result:

\begin{theorem}[D-OWFs imply OWFs (Lemma 1~\cite{impagliazzo1989one}) and Theorem 4.2.2~\cite{impagliazzo1992pseudo})]\label{thm:dowf-imply-owf} Given a distributional OWF, one can construct a OWF. 
\end{theorem}
We note that the construction underneath Theorem~\ref{thm:dowf-imply-owf} works by first constructing a \textit{weak} OWF from the distributional OWF, then constructing a OWF from the weak OWF. An analogous result for \textit{quantum-secure} distributional OWF was proven in Ref.~\cite{kashefi2007statisticalzeroknowledgequantum}:

\begin{theorem}[Quantum-secure D-OWFs imply quantum-secure OWFs (Theorem 5~\cite{kashefi2007statisticalzeroknowledgequantum})]\label{thm:qsecDWOFsimplyqsecOWFS} Given a quantum-secure distributional OWF, one can construct a quantum-secure OWF.
\end{theorem}
We note that in Ref.~\cite{kashefi2007statisticalzeroknowledgequantum} they actually only proved that given a quantum-secure distributional OWF one can construct a quantum-secure \textit{weak} OWF, however as mentioned, quantum-secure weak OWFs imply quantum-secure OWFS~\cite{Radian_2019}.

\subsection{Microcrypt primitives}\label{ss:microcrypt-primitives}

In this section, we provide definitions for all existing Microcrypt primitives that are relevant for this work. We begin with the definition of a \textit{one-way puzzle}, originally introduced by Khurana and Tomer in Ref.~\cite{khurana2024commitmentsquantumonewayness}. 

\begin{definition}[One-way puzzle (OWP)]\label{def:OWP}
A \emph{one-way puzzle} is a pair $({\sf Samp},{\sf Ver})$ consisting of a sampling algorithm ${\sf Samp}$ and a verification function ${\sf Ver}$. Here, ${\sf Samp}(1^\lambda)
\to
(k,s)$ is either a uniform quantum polynomial time (QPT) or probabilistic polynomial time (PPT) algorithm that on input of a security parameter $\lambda$ outputs a pair of classical strings; we call $k$ the {\em key} and $s$ the {\em puzzle}. The function ${\sf Ver}(k,s)=b$ takes a key and a puzzle as inputs and returns a bit $b\in\{0,1\}$. When ${\sf Ver}(k,s)=1$, we say that the pair $(k,s)$ is {\em accepted}. For all sufficiently large $\lambda$, these must satisfy the following properties:
\begin{enumerate}
    \item {\em Correctness}: Outputs of the sampling algorithm are accepted with high probability,
    \begin{equation}
    \underset{(k,s)\,\gets\,{\sf Samp}(1^\lambda)}{\rm Pr}\big[{\sf Ver}(k,s)=1\big]\ge1-{\sf negl}(\lambda)\,.
    \end{equation}
    \item {\em Security}: Given just a puzzle $s$, it is computationally hard to find a key $k$ such that the pair $(k,s)$ is accepted. That is, for all QPT algorithms $\mathcal A(1^\lambda,s)\to k'$ that take a puzzle as input and output a 
    key one has
    \begin{equation}
    \underset{\substack{
    (k,s)\,\gets\,{\sf Samp}(1^\lambda)\,\,\\
    k'\,\gets\,{\mathcal A}(1^\lambda,s)
    }}{\rm Pr}\big[{\sf Ver}(k',s)=1\big]\le{\sf negl}(\lambda)\,.
    \end{equation}
\end{enumerate}
\end{definition}

We note that the definition above allows for classical PPT or quantum QPT sampling algorithms. However, when the sampling algorithm is classical, it is known that one can construct a distributional OWF from the OWP:

\begin{theorem}[Classical sampling for one-way puzzles implies a distributional one-way function (Claim D.1~\cite{khurana2025founding})]\label{thm:KTowf} Let $m:\mathbb{N}\rightarrow\mathbb{N}$ be some fixed polynomial.  Given a one-way puzzle $(\mathsf{Samp},\mathsf{Ver})$, assume that there exists an efficient deterministic classical algorithm $\mathsf{Samp}_{\mathrm{C}}$ such that 
\begin{equation}
d_\mathrm{TV}\left(\mathsf{Samp}(1^\lambda),\mathsf{Samp}_{\mathrm{C}}(1^\lambda,r)_{r\gets\{0,1\}^{m(\lambda)}}\right) \leq 1/3,
\end{equation}
where $\mathsf{Samp}_{\mathrm{C}}(1^\lambda,r)_{r\gets\{0,1\}^{m(\lambda)}}=(k_\lambda(r),s_\lambda(r))_{r\gets\{0,1\}^{m(\lambda)}}$ is understood as the distribution over outputs of $\mathsf{Samp}_{\mathrm{C}}(1^\lambda,r)$ with respect to input bitstrings $r\in\{0,1\}^{m(\lambda)}$ drawn uniformly at random. Then, the family of functions $\{f_\lambda:\{0,1\}^{m(\lambda)}\rightarrow\{0,1\}^*\}_{\lambda\in\mathbb{N}}$ with $f_\lambda(r) = s_\lambda(r)$ is a distributional one-way function.
\end{theorem}\label{thm:classical-sampling-implies-owf}
In light of Theorem~\ref{thm:dowf-imply-owf}, the above means that one can construct OWFs from OWPs with classical sampling algorithms. The result cited above only explicitly proved the ability to construct a distributional OWF, however, the proof can be generalized to obtain a quantum-secure distributional OWF. In particular, we have the following corollary:

\begin{corollary}[Classical sampling for one-way puzzles implies quantum-secure distributional one-way function]\label{cor:classical-sampling-implies-qsecDOWF} Under the same conditions as Theorem~\ref{thm:KTowf}, the family of functions $\{f_\lambda\}_{\lambda\in\mathbb{N}}$ is in fact a quantum-secure distributional one-way function.
\end{corollary}

For completeness, we provide a proof of Corollary~\ref{cor:classical-sampling-implies-qsecDOWF} in Appendix~\ref{app:owp-to-doOWF}. In light of Theorem~\ref{thm:qsecDWOFsimplyqsecOWFS}, the above means that one can construct quantum-secure OWFs from OWPs with classical sampling algorithms. Given these results, from the perspective of Microcrypt, we are most interested in OWPs genuinely \textit{quantum} sampling algorithms. As discussed in the introduction, multiple recent works have given evidence that such \textit{quantum} OWP (i.e.,  with QPT sampling algorithm $\mathsf{Samp}$) can be constructed even if OWFs do not exist~\cite{khurana2025founding,Kretschmer_2025}. Additionally, OWPs are the minimal assumption for \textit{Countcrypt}~\cite{goldin2024countcrypt}.

While the standard definition of a OWP given in Definition~\ref{def:OWP} does not require the verification algorithm ${\sf Ver}$ to be efficient, in this work we will be particularly interested in OWP where the verification algorithm ${\sf Ver}$ \textit{is} efficient. In this case, we obtain an \textit{efficiently verifiable} OWP defined formally as follows.

\begin{definition}[Efficiently verifiable one-way puzzle (EV-OWP)]
An \emph{efficiently verifiable} one-way puzzle is a one-way puzzle $({\sf Samp},{\sf Ver})$ as above, with the added requirement that the verification function $\sf Ver$ is efficiently computable.
\end{definition}

Note that since ${\sf Samp}(1^\lambda)\to(k,s)$ is efficient, one implicitly has that $(k,s)\in\{0,1\}^{O(\poly(\lambda))}$. Therefore, ${\sf Ver}(k,s)=b$ is efficiently computable if the runtime is $O(\poly(\lambda))$ -- i.e., polynomial with respect to the security parameter. Additionally, as discussed at length in the introduction, the existence of EV-OWP is the minimal assumption for Quantumania~\cite{goldin2024countcrypt}.

With this, we proceed to define the notion of a \textit{pure one-way state generator}, initially introduced by Morimae and Yamakawa~\cite{Morimae_2022}.

\begin{definition}[Pure one-way 
state generator (OWSG)]\label{def:owsg}
A \emph{pure one-way state generator} is a triple $({\sf KeyGen},{\sf StateGen}, {\sf Ver})$ of QPT algorithms satisfying the following conditions. The sampling algorithm ${\sf KeyGen}(1^\lambda)\to k$ takes as input a security parameter $\lambda$ and outputs a classical key $k\in \mathbb{K}_\lambda \coloneqq \mathrm{supp}(\mathsf{KeyGen}(\lambda))$. The state generator algorithm ${\sf StateGen}(k)\to\ket{\phi_k}$ takes a key $k\in\mathbb{K}_\lambda$ as input and outputs an $n(\lambda)$-qubit pure quantum state $|\phi_k\rangle$. The verification algorithm ${\sf Ver}(k,\ket\phi)\to b$ takes as input a key and a pure quantum state $|\phi\rangle$ and returns a bit $b\in\{0,1\}$ as the result of applying the projective measurement $\{\Pi_k^1,\Pi_k^0\}=\{\ketbra{\phi_k}{\phi_k},\mathbbm{1}-\ketbra{\phi_k}{\phi_k}\}$ on the state $\ket\phi$, i.e., ${\rm Pr}[{\sf Ver}(k,\ket\phi)=b]=\langle\phi|\Pi_k^b|\phi\rangle$. For sufficiently large $\lambda$, these algorithms must satisfy the following properties:
\begin{enumerate}
    \item {\em Correctness}: Outputs of the sampling algorithm ${\sf KeyGen}(1^\lambda)\to k$ and the state generator algorithm ${\sf StateGen}(k)\to\ket{\phi_k}$ on input $k$ are such that the pair $(k,\ket{\phi_k})$ is accepted with high probability, 
    \begin{equation}
    \underset{
    \substack{
            k\,\gets\,{\sf KeyGen}(1^\lambda)\\
            \ket{\phi_k}\,\gets\,{\sf StateGen}(k)
    }
    }{\rm Pr}\Big[{\sf Ver}(k,\ket{\phi_k})\to1\Big]\ge1-{\sf negl}(\lambda)\,.
    \end{equation}
    \item {\em Security}: Given just copies of the states $\ket{\phi_k}$, it is 
    computationally hard to find a key $k'$ such that the pair $(k',\ket{\phi_k})$ passes verification. That is, for all QPT algorithms $\mathcal A(1^\lambda,\ket\phi^{\otimes t})\to k'$ 
    that, on input $t$ copies of a state, outputs a key, one has
    \begin{equation}
    \underset{
    \substack{
            k\,\gets\,{\sf KeyGen}(1^\lambda)\\
            \ket{\phi_k}\,\gets\,{\sf StateGen}(k)\\
            k'\,\gets\, {\mathcal A}(1^\lambda,\ket{\phi_k}^{\otimes t})
    }
    }{\rm Pr}\Big[{\sf Ver}\big(k',\ket{\phi_k}\big)\to1\Big]\le{\sf negl}(\lambda),
    \end{equation}
    for all $t=O(\poly(\lambda))$.
\end{enumerate}
\end{definition}

Before proceeding, a number of remarks concerning the above definition are in order.
\begin{enumerate}
        \item In principle, as discussed in Ref.~\cite{morimae2024onewaynessquantumcryptography},  one can generalize the definition above in two ways: (a) allowing the QPT algorithm ${\sf StateGen}(k)\to\rho_k$ to output mixed states and (b) allowing arbitrary QPT algorithms ${\sf Ver}(k,\rho)\to b$ where the bit $b$ need not be the result of a rank-$1$ projective measurement. 
        \item However, following Ref.~\cite{morimae2024onewaynessquantumcryptography}, we note that if all $\rho_k=\ketbra{\phi_k}{\phi_k}$ are pure, as we consider in the above definition, then one can always define ${\sf Ver}(k,\rho)\to b$ as the result of applying the projective measurement $\{\Pi_k^1,\Pi_k^0\}=\{\ketbra{\phi_k}{\phi_k},\mathbbm{1}-\ketbra{\phi_k}{\phi_k}\}$ on the state $\rho$. Therefore, for OWSG with pure-state outputs, as we consider here, our definition does not lose generality.
        \item As $\mathsf{KeyGen}$ is efficient we have $\mathbb{K}_\lambda\subseteq\{0,1\}^{O(\poly(\lambda))}$. Given this, it follows from the efficiency of $\mathsf{StateGen}$ that $n(\lambda) = O(\mathrm{poly}(\lambda))$. 
\end{enumerate}
        
Given that we only consider OWSG with pure outputs in this work, we will from now on drop the explicit \textit{pure} qualification. Additionally, we note that in their original work defining OWP~\cite{khurana2024commitmentsquantumonewayness}, Khurana and Tomer showed that OWSG are sufficient for the construction of OWP, by utilizing classical shadows to obtain a classical puzzle from the output of a OWSG. One of the main contributions of this work, in Section~\ref{s:owp-from-cert-owsg}, will be to generalize this construction to allow for arbitrary "measure first, ask later" state certification protocols in place of classical shadows. 

Finally, the last Microcrypt primitive relevant to us is the \textit{pseudorandom state generator}, originally defined by Ji, Liu and Song~\cite{PRSGdefinition}.
        
\begin{definition}[Pseudorandom state generator (PRSG)]\label{def:PRSG}
    A \emph{pseudorandom state generator} is  a pair of QPT algorithms $({\sf KeyGen},{\sf StateGen})$ consisting of a sampling algorithm ${\sf KeyGen}$ and a state generator algorithm ${\sf StateGen}$. The sampling algorithm ${\sf KeyGen}(1^\lambda)\to k$ takes as input a security parameter $\lambda$ and outputs a classical key $k\in \mathbb{K}_\lambda \coloneqq \mathrm{supp}(\mathsf{KeyGen}(\lambda))$. The state generator algorithm ${\sf StateGen}(k)\to\ket{\phi_k}$ takes a key $k\in\mathbb{K}_\lambda$ as input and outputs an $n(\lambda)$-qubit pure quantum state $|\phi_k\rangle$. For sufficiently large $\lambda$ and for all QPT algorithms ${\mathcal A}(1^\lambda, \ket\phi^{\otimes t})\to b$ that, on input $t$ copies of a state, output a bit $b\in\{0,1\}$, one has
    \begin{equation}
    \Bigg|\,\,\underset{
    \substack{
    k\,\gets\,{\sf KeyGen}(1^\lambda)\\
    \ket{\phi_k}\,\gets\,{\sf StateGen}(k)
    }
    }
    {\rm Pr}\big[{\mathcal A}(1^\lambda,\ket{\phi_k}^{\otimes t})=1\big]-\underset{\ket\psi\gets{\rm Haar}(n)}{\rm Pr}\big[{\mathcal A}(\ket{\psi}^{\otimes t})=1\big]\,\Bigg|\le{\sf negl}(\lambda),
    \end{equation}
for all $t=O(\poly(\lambda))$, where ${\rm Haar}(n)$ is the Haar distribution over $n$-qubit states. 
\end{definition}

As for a OWSG, given that $({\sf KeyGen},{\sf StateGen})$ are both efficient in the above definition, one again implicitly has that $\mathbb{K}_\lambda\subseteq \{0,1\}^{O(\poly(\lambda))}$ and $n=O(\poly(\lambda))$. As alluded to in the introduction, we will refer to the function $n(\lambda)$ as the \textit{stretch} of the PRSG. We note that it is helpful to distinguish the following regimes: 

\begin{enumerate}
\item $n(\lambda) = \omega(\log \lambda)$. This is the regime in which (a) Ji, Liu and Song were able to show that PRSGs can be built from (quantum-secure) OWF~\cite{PRSGdefinition}, (b) PRSG immediately yield OWSG~\cite{Cavalar_2025} which in turn yield OWP~\cite{khurana2025founding} and (c) in which Kretschmer gave a black box separation between PRSG and OWF~\cite{Kretschmer_2021} (giving evidence that such PRSG can exist even if OWFs do not).  We will call such PRSG \textit{standard} PRSG.
\item $n(\lambda) = c\log \lambda$ with $(c \geq 1)$. These are called \textit{short} or \textit{logarithmic length} PRSGs. Given that one can perform tomography on a quantum state of $O(\log \lambda)$ qubits to any desired precision $\epsilon$ in time $\poly(\lambda, \epsilon)$, short PRSG behave more like cryptographic objects with classical output. Using this insight, such PRSG can be used to build quantum pseudorandom generators (QPRG), which are \textit{pseudodeterministic} generators of pseudorandom bit strings~\cite{ananth2023pseudorandomstringspseudorandomquantum}, from which one can construct EV-OWP~\cite{QCCCCrypto}. Such short PRSG can also be constructed from (quantum-secure) OWFs~\cite{brakerski2020scalablepseudorandomquantumstates}, and as a result the recent black box separation between OWFs and \textit{quantum computable} OWFs gives evidence that short PRSG can also exist in a world without OWFs~\cite{Kretschmer_2025}.
\item $n(\lambda) \leq c\log \lambda$ for $c\in (0,1)$. Here, there exist $c$ such that PRSGs can exist \textit{unconditionally}~\cite{brakerski2020scalablepseudorandomquantumstates,Prabhanjan}.
\end{enumerate}
In this work, given that we are interested in novel constructions of EV-OWP from OWSGs, we focus on the regime of \textit{standard} PRSG. Finally, we note that one can also define a weaker notion of a \textit{single-copy secure} PRSG, in which the condition of Definition~\ref{def:PRSG} is only required to hold for $t=1$. We will, however, be concerned here with standard multi-copy secure PRSG.

\subsection{Hamiltonian phase states and Microcrypt}\label{ss:HPS-prelim}

In this section, we introduce the notion of phase states, and the specific set of phase states known as \textit{Hamiltonian phase states} (HPS). Doing this allows us to introduce the central assumptions in this work -- namely, the search and decision HPS assumptions recently introduced in Ref.~\cite{bostanci2025efficientquantumpseudorandomnesshamiltonian}. We start with the definition of phase states.

\begin{definition}[Phase states]\label{def:phase-states}
Let $n$ be a positive integer and consider arbitrary functions $f:\{0,1\}^n\to[0,2\pi)$ that map classical strings to angles. Every such {\em phase function} $f$ defines a {\em phase state} $|{\phi_f}\rangle=\sum_{{\bf x}}\alpha_f({\bf x})\ket{{\bf x}}$ on $n$ qubits with $\alpha_f({\bf x})=2^{-n/2}\exp(\iu f({\bf x}))$. The class ${\rm PS}_n$ of $n$-qubit phase states is defined as
\begin{equation}
{\rm PS}_n\coloneqq\Big\{|{\phi_f}\rangle\coloneqq\frac1{\sqrt{2^n}}\sum_{{\bf x}\in{\mathbb F}_2^n}\exp(\iu f({\bf x}))\ket{\bf x}\Big\}_{f:\{0,1\}^n\to[0,2\pi)}\,.
\end{equation}
There are uncountably many such states. We will be interested in restricted families of states associated to sets $\{f_k\}_{k\in{\mathbb K}}$ of phase functions that one can at least label by the elements $k\in{\mathbb K}$ of some numerable set. Given a family of functions $\{f_k\}_{k\in{\mathbb K}}$ with $f_k:\{0,1\}^n\to[0,2\pi)$ for all $k\in{\mathbb K}$, we define its corresponding set of phase states ${\rm PS}_n(\{f_k\}_{k\in{\mathbb K}})$ simply as
\begin{equation}
{\rm PS}_n(\{f_k\}_{k\in{\mathbb K}})\coloneqq\Big\{\ket{\phi_k}\coloneqq\frac1{\sqrt{2^n}}\sum_{{\bf x}\in{\mathbb F}_2^n}\exp(\iu f_k({\bf x}))\ket{\bf x}:\,k\in{\mathbb K}\Big\}\subset {\rm PS}_n.
\end{equation}
 \end{definition}

 Given this, we will primarily be interested in sets of phase states defined by \textit{efficiently computable} phase functions, defined as follows:
 
 \begin{definition}[Efficiently computable phase functions]\label{def:eff-computable-phase-functions} Given a set of phase functions $\{f_k\}_{k\in\mathbb{K}}$ with $f_k:\{0,1\}^n\rightarrow[0,2\pi)$, we say that $\{f_k\}_{k\in\mathbb{K}}$ is \textit{efficiently computable} if there exists a deterministic algorithm $\mathcal{A}$ which on input $k\in\mathbb{K}$, $x\in\{0,1\}^n$ and $1^l$ with $l\in \mathbb{N}$, runs in time $\mathrm{poly}(|k|,n,l)$ and outputs a dyadic rational $\mathcal A(k,x,l)=\hat{f}_k(x,l)$ satisfying  
\begin{equation}\label{eq:hat-f-def}
\bigl|\hat{f}_k(x,l)-f_k(x)\bigr|\leq \frac1{2^{l+1}},
\end{equation}
where $|k|$ denotes the bit-length of the encoding of index $k$.
 \end{definition}

With this established, we now define the specific set of phase states known as Hamiltonian phase states.

\begin{definition}[Hamiltonian phase states]\label{def:HPS}
Let $q,m,n$ be positive integers and $\Theta_q=\{0,2\pi/q,4\pi/q\ldots,2(q-1)\pi/q\}$. The class ${\rm HPS}_{q,m,n}$ of \emph{Hamiltonian phase states} (HPS) contains $n$-qubit states of the form $\ket{\phi(\boldsymbol\theta,\bf A)}=U({\boldsymbol\theta},{\bf A})\ket{+^n}$, where the unitary $U(\boldsymbol\theta,\bf A)$ is diagonal and given by
\begin{equation}
\label{HPS}
U({\boldsymbol\theta},{\bf A})=\exp\Big(\iu \sum_{j=1}^m\theta_j\,Z^{A_{j,1}}\otimes Z^{A_{j,2}}\otimes\cdots\otimes Z^{A_{j,n}}\Big)=\sum_{{\bf x}\in{\mathbb F}_2^n}\exp(\iu \sum_{j=1}^m\theta_j\,(-1)^{{\bf A}_{j}\cdot{\bf x}})\ket{\bf x}\!\!\bra{\bf x}
\end{equation}
for all $\boldsymbol\theta=(\theta_1,\ldots,\theta_m)\in\Theta_q^m$, and ${\bf A}=((A_{1,1},\ldots,A_{1,n}),\ldots,(A_{m,1},\ldots,A_{m,n}))\in{\mathbb F}_2^{mn}$. We define ${\mathbb K}_{q,m,n}\coloneqq\Theta_q^m\times{\mathbb F}_2^{mn}$ for the set of labels and write $k=({\boldsymbol\theta},{\bf A})\in{\mathbb K}_{q,m,n}$ for its elements. Therefore,
\begin{equation}\label{eq:HPS-phik}
    {\rm HPS}_{q,m,n}\coloneqq\Big\{\ket{\phi_k}\coloneqq\frac1{\sqrt{2^n}}\sum_{{\bf x}\in{\mathbb F}_2^n}\exp(\iu f_k({\bf x}))\ket{\bf x}\in({{\mathbb C}^2})^{\otimes n}:k\in{\mathbb K}_{q,m,n}\Big\}
\end{equation}
where, for all $k=({\boldsymbol\theta},{\bf A})\in{\mathbb K}_{q,m,n}$, we have
\begin{equation}\label{eq:def-fk}
f_{k}({\bf x})\coloneqq\sum_{j=1}^m\theta_j\,(-1)^{{\bf A}_{j}\cdot{\bf x}}=\sum_{j=1}^m\theta_j\,(-1)^{A_{j,1}\cdot{x_1}}\cdots(-1)^{A_{j,n}\cdot{x_n}}\,.
\end{equation}
\end{definition}

\begin{remark}[Preparing Hamiltonian phase states]\label{remark:efficient-preparation-HPS}  As discussed in detail in Section 2.4 of Ref.~\cite{bostanci2025efficientquantumpseudorandomnesshamiltonian}, states in the class ${\rm HPS}_{q,m,n}$ can be prepared via \textit{instantaneous quantum polynomial-time} (IQP) circuits with $\poly(n,m)$ gates acting on $\ket{0^n}$, independently of the value of $q$. Specifically, preparation of any state $|\phi_k\rangle \in {\rm HPS}_{q,m,n}$ only requires a single layer of Hadamard gates, followed by $\lceil m/n\rceil$ alternating layers of single-qubit $Z$ rotations and CNOT circuits. The commuting structure of these circuits makes Hamiltonian phase
states natural candidates for implementation in programmable quantum
platforms, including Rydberg-atom and superconducting architectures, as discussed in Appendix~\ref{app:experimental-implementation}.
\end{remark}

\begin{remark}[On parameter regimes for poly-size keys]\label{remark:HPS} For integer values of $q$, the class ${\rm HPS}_{q,m,n}$ contains $|{\mathbb K}_{q,m,n}|=q^m2^{nm}$ states that can be labeled by keys $k$ with $|k|= m(n+\lceil\log q\rceil)$ bits. In general, we will take the number of qubits $n$ as our security parameter for cryptographic applications and thus require that both $m$ and $\log q$ are at most $\poly(n)$, so that the class ${\rm HPS}_{q,m,n}$ can be labeled by keys with $\poly(n)$ bit-length encodings.
\end{remark}

\begin{remark}[Efficient computability of HPS phase functions]\label{remark:eff-computability-HPS} For any positive integers $q,m,n$, the set of phase functions $\{f_k\}_{k\in\mathbb{K}_{q,m,n}}$ defining the class $\mathrm{HPS}_{q,m,n}$ in Eq.~\eqref{eq:def-fk} is efficiently computable. More specifically, there exists an algorithm $\mathcal{A}$ which on input $(k,x,1^l)$ runs in time $\mathrm{poly}(n,m,\log q, l)$ and outputs $\mathcal{A}(k,x,l)=\hat{f}_k(x,l)$ satisfying Eq.~\eqref{eq:hat-f-def}. In the regime of Remark~\ref{remark:HPS} -- i.e. when both $m$ and $\log q$ are at most $\poly(n)$ -- the runtime of $\mathcal{A}$ on input $(k,x,l)$ is $\mathrm{poly}(n,l)$.
\end{remark}

We note that Hamiltonian phase states, and the IQP circuits which prepare them, were initially proposed as a class of states that is both easy to prepare and hard to simulate classically~\cite{shepherd2009temporally,bremner2011classical}. However from our perspective, our primary interest in Hamiltonian phase states stems from the recently introduced assumptions, motivated by the apparent hardness of \textit{learning} these states, that they can be used to construct both PRSG and OWSG. We start from the \textit{Decision} HPS assumption:

\begin{assumption}[Decision HPS assumption -- Definition 2 in Ref.~\cite{bostanci2025efficientquantumpseudorandomnesshamiltonian}]\label{ass:decision-HPS}
There exist functions $q(n) = 2^{O(\poly(n))}$ and $m(n) = \poly(n)$, and a classically efficiently sampleable distribution $\chi_{q,m,n}$ over $\mathbb K_{q,m,n}$, such that for all QPT algorithms $\mathcal A(\ket\phi^{\otimes t})\to b$ that, on input $t$ copies of an $n$-qubit state $|\psi\rangle$, output a bit $b\in\{0,1\}$, one has
\begin{equation}
\Bigg|\underset{k\,\gets\,\chi_{q,m,n}}{\rm Pr}\big[{\mathcal A}(\ket{\phi_k}^{\otimes t})=1\big]-\underset{\ket\psi\gets{\rm Haar}(n)}{\rm Pr}\big[{\mathcal A}(\ket{\psi}^{\otimes t})=1\big]\Bigg|\le{\sf negl}(n)
\end{equation}
for all $t=\poly(n)$, where ${\rm Haar}(n)$ is the Haar distribution over $n$-qubit states.
\end{assumption}

Note that an immediate corollary of the above assumption is that, under the Decision HPS assumption, the following construction yields a PRSG.

\begin{construction}[PRSG from Decision HPS]\label{con:PRSG-from-HPS} Given $m(n)$, $q(n)$ and a distribution $\chi_{q,m,n}$ which satisfy the Decision HPS assumption, define:
\begin{enumerate}
\item ${\sf KeyGen}(1^n)$: Sample a key $k\gets \chi_{q,m,n}$.
\item ${\sf StateGen}(k\in \mathbb{K}_{q,m,n})$: Output an $n$-qubit state $\ket{\phi_k}\in {\rm HPS}_{q,m,n}$ defined in Eqs.~\eqref{eq:HPS-phik}-\eqref{eq:def-fk}.
\end{enumerate}
\end{construction}
In particular, note that we have used $n$ as the security parameter, and since $\log|\mathbb K_{q,m,n}|=m(n+\log q)=\poly(n)$, ${\sf KeyGen}$ is efficient with respect to $n$.

With this established, we now introduce the \textit{Search} HPS assumption.

\begin{assumption}[Search HPS assumption -- Definition 1 in Ref.~\cite{bostanci2025efficientquantumpseudorandomnesshamiltonian}]\label{ass:search-HPS}
There exist functions $q(n) = 2^{O(\poly(n))}$ and $m(n) = \poly(n)$, and a classically efficiently sampleable distribution $\chi_{q,m,n}$ over $\mathbb K_{q,m,n}$, such that for all QPT algorithms $\mathcal A(\ket\phi^{\otimes t})\to k'$ that, on input $t$ copies of a $n$-qubit state, output a key $k'\in{\mathbb K}_{q,m,n}$, one has
\begin{equation}
\underset{\substack{
k\,\gets\,\chi_{q,m,n}\\
k'\,\gets\,{\mathcal A}(\ket{\phi_k}^{\otimes t})
}}{\mathbb E}\Big[|\braket{\phi_{k'}}{\phi_k}|^2\Big]\le{\sf negl}(n)
\end{equation}
for all $t=\poly(n)$, where $\mathbb E$ denotes the expectation value.
\end{assumption}

In this case, an immediate corollary of the Search HPS assumption is that the following construction yields a OWSG:

\begin{construction}[OWSG from Search HPS]\label{con:OWSG-from-HPS} 
Given $m(n)$, $q(n)$ and a distribution $\chi_{q,m,n}$ which satisfy the 
Search HPS assumption, define:
\begin{enumerate}
\item ${\sf KeyGen}(1^n)$: Sample a key $k\gets \chi_{q,m,n}$.
\item ${\sf StateGen}(k\in \mathbb{K}_{q,m,n})$: Output an $n$-qubit state $\ket{\phi_k}\in {\rm HPS}_{q,m,n}$.
\item $\mathsf{Ver}(k,\ket\phi)$: Apply the projective measurement $\{\Pi_k^1,\Pi_k^0\}=\{\ketbra{\phi_k}{\phi_k},\mathbbm{1}-\ketbra{\phi_k}{\phi_k}\}$ on the state $\ket\phi$ and output the result. More specifically:
\begin{equation}
\begin{aligned}
{\rm Pr}\big[{\sf Ver}(k,\ket\phi)\to1\big]&={\rm tr}(\Pi_k^1\ketbra{\phi}{\phi})=|\braket{\phi_k}{\phi}|^2\,,\\
{\rm Pr}\big[{\sf Ver}(k,\ket\phi)\to0\big]&={\rm tr}(\Pi_k^0\ketbra{\phi}{\phi})=1-|\braket{\phi_k}{\phi}|^2\,.
\end{aligned}
\end{equation}
\end{enumerate}
\end{construction}

\begin{remark}[Relation between HPS assumptions]
Note that the Decision HPS assumption (Assumption~\ref{ass:decision-HPS}) implies the Search HPS assumption (Assumption~\ref{ass:search-HPS}). Additionally, it can be that the functions $q(n)$ and $m(n)$ required to satisfy the Search HPS assumption (and therefore have a OWSG) are smaller than those required to satisfy the Decision HPS assumption (and therefore have a PRSG). Given that the Search HPS assumption is the weaker assumption, we focus from here on 
showing that the Search HPS assumption implies OWFs, which immediately yields the same conclusion for the Decision HPS assumption.
\end{remark}

\section{State certification}\label{ss:state certification-prelim}

The notion of a "measure first, ask later" state certification protocol is central to this work. In order to define this rigorously, and to provide intuition for the definition, we start with the standard definition of a state certification protocol for an unknown target state~\cite{badescu2017quantumstatecertification}.

\begin{definition}[State certification protocol for $|\phi\rangle$] We say that an algorithm $\mathcal{A}$ is a state certification protocol for an $n$-qubit state $|\phi\rangle$ from $t(n,\epsilon,\delta)$ copies if, for all states $|\psi\rangle$ and $(\epsilon,\delta)\in (0,1)$:
\begin{enumerate}
\item If $|\psi\rangle = |\phi\rangle$ then 
\begin{equation}
\mathrm{Pr}\left[\mathcal{A}\left(|\psi\rangle^{\otimes t}\right) = \mathsf{Accept}\right] \geq 1-\delta
\end{equation}
\item If $|\langle \psi|\phi\rangle|^2 < 1-\epsilon$ then 
\begin{equation}
\mathrm{Pr}\left[\mathcal{A}\left(|\psi\rangle^{\otimes t}\right) = \mathsf{Reject}\right] \geq 1-\delta
\end{equation}
\end{enumerate}
\end{definition}
We note that for \textit{any} fixed state $|\phi\rangle$ there is a trivial state certification protocol with \textit{constant} copy complexity with respect to $n$: Given $O(\epsilon^{-2}\log(\delta^{-1}))$ copies of $|\psi\rangle$ simply perform the projective measurement $\{|\phi\rangle\langle \phi|, \mathds{1}-|\phi\rangle\langle \phi|\}$ on each copy, use this to estimate $|\langle \psi|\phi\rangle|^2$ sufficiently accurately, then make the appropriate decision. In other words, when one is allowed to measure in a $|\phi\rangle$-dependent way -- i.e.,  in a way which \textit{depends on the target state} -- then state certification is trivial. This observation makes clear that one could also consider the refined notion of a ``measure first, ask later'' state certification protocol. In such a protocol we insist that all measurements are \textit{target state independent}, and that only the classical post-processing of measurement outcomes is target state dependent. More specifically, we have the following definition:

\begin{definition}[``Measure first, ask later'' state certification protocol for $|\phi\rangle$]\label{def:mfal-for-a-single-state} Given an $n$-qubit state $|\phi\rangle$, consider a QPT algorithm $\mathcal{M}(|\psi\rangle^{\otimes t})\rightarrow s\in\{0,1\}^{g(t,n)}$ \textit{which does not depend} on $|\phi\rangle$, and let $F_{|\phi\rangle}:\{0,1\}^*\times\{0,1\}^*\times\{0,1\}^*\rightarrow \{\mathsf{Accept},\mathsf{Reject}\}$ be a \textit{classical} function  which can depend on $|\phi\rangle$. We say that $(\mathcal{M},F_{|\phi\rangle})$ is a ``measure first, ask later'' state certification protocol for $|\phi\rangle$ from $t(n,\epsilon,\delta)$ copies if, for all states $|\psi\rangle$ and $\epsilon,\delta\in (0,1)$ with finite binary encodings:
\begin{enumerate}
\item If $|\psi\rangle = |\phi\rangle$ then 
\begin{equation}
\underset{s\gets \mathcal{M}(|\psi\rangle^{\otimes t})}{\mathrm{Pr}}\left[F_{|\phi\rangle}(s,\epsilon,\delta) = \mathsf{Accept}\right] \geq 1-\delta
\end{equation}
\item If $|\langle \psi|\phi\rangle|^2 < 1-\epsilon$ then 
\begin{equation}
\underset{s\gets \mathcal{M}(|\psi\rangle^{\otimes t})}{\mathrm{Pr}}\left[F_{|\phi\rangle}(s,\epsilon,\delta) = \mathsf{Reject}\right] \geq 1-\delta
\end{equation}
\end{enumerate}
We say that $(\mathcal{M},F_{|\phi\rangle})$ is 
\begin{enumerate}
\item \textit{Copy efficient} if $t=\mathrm{poly}(n,\epsilon^{-1},\log \delta^{-1})$ is sufficient,
\item  \textit{Computationally efficient from $\mathsf{O}$-access} if it is copy efficient and $F_{|\phi\rangle}$ is computationally efficient, given access to an oracle $\mathsf{O}(|\phi\rangle)$ which provides some notion of classical access to $|\phi\rangle$ (for example, an oracle which when queried on $x$ provides the amplitude $\langle x|\phi\rangle$).
\item \textit{Computationally efficient} if it is copy efficient and $F_{|\phi\rangle}$ is computationally efficient.
\item\textit{$\eta$-simulable} if there exists a classical PPT algorithm $\mathcal{M}_C(1^n,1^t)\rightarrow s\in\{0,1\}^{g(t,n)}$ such that 
\begin{equation}
d_{\mathrm{TV}}(\mathcal{M}_C(1^n,1^t),\mathcal{M}(|\phi\rangle^{\otimes t})) \leq \eta(n,t),
\end{equation}
for all $t\in\mathbb{N}$, where $\mathcal{M}_C(1^n,1^t)$ and $\mathcal{M}(|\phi\rangle^{\otimes t})$ are understood as distributions over $\{0,1\}^{g(t,n)}$. 
\end{enumerate} 
\end{definition}

\begin{remark}[Oracle access to the target state] We note that recent works on state certification assume different models of oracle access to the target state. For example, Ref.~\cite{HuangPreskillSoleimanifar2025} assumes an oracle which when queried on a bit string $x$, provides the amplitude $\langle x|\phi\rangle$, Ref.~\cite{coladangelo2026powerbasesrobustcopyoptimal} considers an oracle which can provide amplitudes in either the computational or Hadamard basis, and Ref.~\cite{gupta2025singlequbitmeasurementssufficecertify} considers an oracle which when queried on any sequence of single-qubit measurements and outcomes provides the probability of that specific outcome sequence. Whenever a ``measure first, ask later'' state certification protocol is efficient from $\mathsf{O}$-access \textit{and} that oracle can be efficiently classically simulated (with respect to $n$) for the specific target state, then the protocol will be computationally efficient. 
\end{remark}

We now want to generalize the definition above to a ``measure first, ask later'' state certification protocol for a \textit{set} of states. To this end, it will be helpful to note that we can combine the two conditions given in the definition above into the single condition that
\begin{equation}
\underset{s\gets \mathcal{M}(|\psi\rangle^{\otimes t})}{\mathrm{Pr}}\left[\begin{cases} F_{|\phi\rangle}(s,\epsilon,\delta)=\mathsf{Accept}\text{ if } |\psi\rangle = |\phi\rangle \\  F_{|\phi\rangle}(s,\epsilon,\delta)=\mathsf{Reject}\text{  if }|\langle \psi|\phi\rangle|^2 < 1-\epsilon \end{cases}\right] \geq 1-\delta.
\end{equation}
With this in hand, we can finally define the notion of a ``measure first, ask later'' state certification protocol for \textit{simultaneously} certifying a \textit{set} of states $\{|\phi_k\rangle\,|\,k\in [K]\}$ from a \textit{single set} of measurement outcomes $s\gets\mathcal{M}(|\psi\rangle^{\otimes t})$.

\begin{definition}[``Measure first, ask later'' state certification protocol for $\{ |\phi_k\rangle\}$]\label{def:measure-first-ask-later}
Let $\{|\phi_k\rangle\,|\, k\in \mathbb{K}\}$ be a finite set of $n$-qubit states. Consider a QPT algorithm $\mathcal{M}(|\psi\rangle^{\otimes t})\rightarrow s\in\{0,1\}^{g(t,n)}$ and let $\mathcal{F}:\{0,1\}^*\times \{0,1\}^*\times \{0,1\}^*\times \{0,1\}^* \rightarrow \{\mathsf{Accept},\mathsf{Reject}\}$ be such that $\mathcal{F}(\tilde{k},\cdot,\cdot,\cdot) = \mathsf{Reject}$ for all $\tilde{k}\notin \mathbb{K}$. We say that $(\mathcal{M},\mathcal{F})$ is a ``measure first, ask later'' state certification protocol for $\{|\phi_k\rangle\}$ from $t = t(n,\epsilon,\delta,|\mathbb{K}|)$ copies if, for all states $|\psi\rangle$ and all $\epsilon,\delta\in (0,1)$ with finite binary encodings:
\begin{equation}
\underset{s\gets \mathcal{M}(|\psi\rangle^{\otimes t})}{\mathrm{Pr}}\left[\forall \, k\in \mathbb{K} \, \begin{cases} \mathcal{F}(k,s,\epsilon,\delta)=\mathsf{Accept}\text{ if } |\psi\rangle = |\phi_k\rangle \\  \mathcal{F}(k,s,\epsilon,\delta)=\mathsf{Reject}\text{ if } |\langle \psi|\phi_k\rangle|^2 < 1-\epsilon \end{cases}\right] \geq 1-\delta.
\end{equation}
We say that $(\mathcal{M},\mathcal{F})$ is:
\begin{enumerate}
\item \textit{Copy efficient} if $t=\mathrm{poly}(n,\epsilon^{-1},\log \delta^{-1},\log |\mathbb{K}|)$ is sufficient.
\item \textit{Computationally efficient from $\mathsf{O}$-access} 
if it is copy efficient and the function $\mathcal{F}$ is computationally efficient given access to an oracle $\tilde{\mathsf{O}}$ which when queried on $(k,x)$ returns the response from querying $\mathsf{O}(|\phi_k\rangle)$ on $x$ for some oracle $\mathsf{O}(|\phi_k\rangle)$ providing some notion of classical access to $|\phi_k\rangle$.
\item \textit{Computationally efficient} if it is copy efficient and $\mathcal{F}$ is computationally efficient.
\item $\eta$-simulable if there exists a classical PPT algorithm $\mathcal{M}_C(k,1^n,1^t)\rightarrow s\in\{0,1\}^{g(t,n)}$ such that
\begin{equation}
d_\mathrm{TV}(\mathcal{M}_C(k,1^n,1^t),\mathcal{M}(|\phi_k\rangle^{\otimes t})) \leq \eta(|k|,n,t)
\end{equation}
for all $k\in\mathbb{K}$ and $t\in\mathbb{N}$, where again $\mathcal{M}_C(k,1^n,1^t)$ and $\mathcal{M}(|\phi_k\rangle^{\otimes t})$ are understood as distributions over $\{0,1\}^{g(t,n)}$.
\end{enumerate}
\end{definition}

Before we continue it is worth making a few remarks on the above definition.

\begin{remark}[On the relation to classical shadows]\label{rem:classical-shadows} We note that global Clifford shadows~\cite{Huang_2020} provides a paradigmatic \textit{copy efficient} ``measure first, ask later'' state certification protocol for \textit{any} set of states. However, for arbitrary sets of states this protocol will \textit{not} be computationally efficient or $\eta$-simulable for meaningfully small values of $\eta$. Indeed, the term "measure first, ask later" has been deliberately borrowed from existing literature on randomized measurement protocols~\cite{elben2023randomized} to signify that Definition~\ref{def:measure-first-ask-later} is meant to provide a generalization and abstraction of ``classical shadow type" state certification protocols.
\end{remark}

\begin{remark}[Relevant parameter ranges]\label{rem:paramter-ranges} Let us assume that the pair $(\mathcal{M},\mathcal{F})$ is a measurement protocol for the set of $n$-qubit states $\{|\phi_k\rangle\,|\,k\in\,\mathbb{K}\}$:
\begin{enumerate}
    \item In many settings, one will typically be happy with both $\epsilon$ and $\delta$ being some small constants. However, in the cryptographic context we consider here, where we would like failure to occur with negligible probability, we will usually fix $\epsilon=1/8$ and $\delta=2^{-n}$. Additionally, we will work only with sets of states of size $|\mathbb{K}|=2^{O(\poly(n))}$. In this case, we have that $t=\mathrm{poly}(n,\epsilon^{-1},\log \delta^{-1},\log |\mathbb{K}|)$ implies $t=O(\poly(n))$ and that computational efficiency of $\mathcal{F}$ implies it can be evaluated in time $O(\mathrm{poly}(n))$.
    \item In order to be PPT, the algorithm $\mathcal{M}_C$ appearing in the definition of $\eta$-simulability should have runtime at most $\mathrm{poly}(|k|,n,t)$ on input $(k,1^n,1^t)$. However, in the parameter ranges mentioned above -- i.e. where $t=O(\mathrm{poly}(n))$ and $K=2^{O(\poly(n))}$ -- the runtime of $\mathcal{M}_C$ will be $O(\mathrm{poly}(n))$.
\end{enumerate}
\end{remark}

\begin{remark}[Refinements of ``measure first, ask later'' state certification protocols]\label{remark:single-copy} We note that in Definition~\ref{def:measure-first-ask-later}, we have placed no restrictions on the algorithm $\mathcal{M}$ beyond the fact that it is efficient. However, one can easily define the notion of a single-copy, or non-adaptive, ``measure first, ask later'' state certification protocol, by placing the appropriate restriction on $\mathcal{M}$. Intuitively, the more restrictions one places on $\mathcal{M}$, the harder it should be to construct such a state certification protocol for a given set of states. However, perhaps surprisingly, recent work has shown that (given access to an amplitude oracle) almost all quantum states can be certified by a ``measure first, ask later'' state certification protocol with \textit{single qubit non-adaptive measurements}~\cite{HuangPreskillSoleimanifar2025}. 
\end{remark}

\begin{remark}[On $\eta$-simulability] We note that $\eta$-simulability means that for any $k\in\mathbb{K}$, one can efficiently sample (up to $\eta$ accuracy) from the distribution over measurement outcomes obtained by measuring $|\phi_k\rangle^{\otimes t}$ via $\mathcal{M}$. While the utility or relevance of this property may not be as immediately clear as the efficiency properties, we will see in the following sections that simulability of a measurement protocol is what allows us to obtain a OWF from the OWP constructed via "measure first, ask later" state certification protocols.
\end{remark}

\section{Certifiable Microcrypt primitives}\label{s:certifiable-microcrypt}

Equipped with the notion of a ``measure first, ask later'' state certification protocol for a set of states we can finally define \textit{certifiable} variants of OWSG and PRSG. 

\begin{definition}[Certifiable one-way state generator (C-OWSG) and certifiable pseudorandom state generator (C-PRSG)]\label{def:certifiable-microcrypt-primitives}
Consider a one-way state generator $({\sf KeyGen},{\sf StateGen}, {\sf Ver})$ (or pseudorandom state generator $({\sf KeyGen},{\sf StateGen})$),  with $\mathbb K_\lambda\coloneqq{\rm supp}({\sf KeyGen}(1^\lambda))$ and $\ket{\phi_k} = {\sf StateGen}(k)$ for all $k\in \mathbb{K}_\lambda$. Given a ``measure first, ask later'' state certification protocol $(\mathcal{M},\mathcal{F})$ for the set of states $\{|\phi_k\rangle\,|\,k\in \mathbb{K}_\lambda\}$, we say that the tuple $\big({\sf KeyGen},{\sf StateGen}, {\sf Ver}),(\mathcal{M},\mathcal{F})\big)$ (or tuple $\big({\sf KeyGen},{\sf StateGen}),(\mathcal{M},\mathcal{F})\big)$) is: 
\begin{enumerate}
\item A \textit{certifiable} one-way state generator (pseudorandom state generator) if $(\mathcal{M},\mathcal{F})$ is copy efficient.
\item An \textit{$\eta$-simulable certifiable} one-way state generator  (pseudorandom state generator)  if $(\mathcal{M},\mathcal{F})$ is copy efficient and $\eta$-simulable.
\item An \textit{efficiently certifiable} one-way state generator (pseudorandom state generator) if $(\mathcal{M},\mathcal{F})$ is computationally efficient. 
\item An \textit{$\eta$-simulable efficiently certifiable} one-way state generator (pseudorandom state generator) if $(\mathcal{M},\mathcal{F})$ is computationally efficient and $\eta$-simulable. 
\end{enumerate}
\end{definition}

A-priori the utility of distinguishing carefully between all the certifiable OWSG variants defined above may not be clear. However, we show in the following sections that each of the above variants can be used to construct different cryptographic primitives. In particular:

\begin{enumerate}
\item Certifiable OWSGs can be used to construct OWPs (see Theorem~\ref{thm:owp-from-certifiable-owsg}).
\item Efficiently certifiable OWSGs can be used to construct EV-OWPs (see Corollary~\ref{cor:ev-owp-from-eff-certifiable-owsg}).
\item $1/3$-simulable certifiable OWSGs can be used to construct OWFs via distributional OWFs (see Theorem~\ref{thm:owf-via-sim-certOWSG}).
\item $\eta$-simulable efficiently certifiable OWSGs can be used to directly construct OWFs (without going via distributional OWFs) whenever $\eta(|k|,n,t)=\mathsf{negl}(\lambda)$ for $|k|,n$ and $t$ which are all $O(\mathrm{poly}(\lambda))$. 
\end{enumerate}
Given the above, we see that the same OWSGs can be used to construct a variety of different cryptographic primitives (in different cryptographic worlds), depending on the properties of the "measure first, ask later" state certification protocol with which it is equipped! As such, all the variants of certifiable OWSG defined above really should be considered as different primitives.

\begin{remark}[All OWSG are certifiable] Recall from  Remark~\ref{rem:classical-shadows} that global Clifford shadows provides a copy-efficient ``measure first, ask later'' state certification protocol for \textit{any} set of states. Therefore, \textit{any} OWSG (PRSG) equipped with the state certification protocol provided by global Clifford shadows is immediately a certifiable OWSG (PRSG). Using the results mentioned above (in particular Theorem~\ref{thm:owp-from-certifiable-owsg}) this means that one can construct a OWP from any OWSG, by using the global Clifford shadows state certification protocol. Indeed, this is precisely how Khurana and Tomer have proven that OWSG can be used to construct OWP~\cite{khurana2024commitmentsquantumonewayness}.
\end{remark}

\begin{remark}[Further variants of certifiable primitives] As per Remark~\ref{remark:single-copy}, it is clear that one can define restricted variations of ``measure first, ask later'' state certification protocols, and using these one could define variants of certifiable OWSG (PRSG) such as single-copy and non-adaptive certifiable and efficiently certifiable OWSG and PRSG. Given that the more restrictions one places on the state certification protocol the harder it is to construct the primitive, one could conjecture that it is possible to construct more powerful cryptographic protocols from restricted variants of the certifiable microcrypt primitives we introduce here.
\end{remark}

Finally, with the above established, we make the observation that variants of certifiable \textit{standard} PRSG (i.e.,  with super-logarithmic stretch) can be used to construct variants of certifiable OWSG. This follows directly from the fact that given a standard PRSG $(\mathsf{KeyGen},\mathsf{StateGen})$ (i.e.,  with stretch $n(\lambda)=\omega(\log \lambda)$) the tuple $(\mathsf{KeyGen},\mathsf{StateGen},\mathsf{Ver})$ with the standard verification algorithm \textit{is} a OWSG~\cite{Cavalar_2025}. Specifically, as an immediate corollary of this prior result, together with the definitions of certifiable PRSG and OWSG, we have the following:

\begin{corollary}[Certifiable PRSG imply certifiable OWSG]\label{thm:cert-prsg-imply-cert-owsg}
Let $(\mathsf{KeyGen},\mathsf{StateGen})$ be a \textit{standard} pseudorandom state generator -- i.e.,  with stretch $n(\lambda) = \omega(\log \lambda)$. Let $\mathsf{Ver}$ be as per Definition~\ref{def:owsg}. Then:
\begin{enumerate}
    \item If $((\mathsf{KeyGen},\mathsf{StateGen}),(\mathcal{M},\mathcal{F}))$ is a (efficiently) certifiable PRSG then $((\mathsf{KeyGen},\mathsf{StateGen},\mathsf{Ver}),(\mathcal{M},\mathcal{F}))$ is a (efficiently) certifiable OWSG
    \item If $((\mathsf{KeyGen},\mathsf{StateGen}),(\mathcal{M},\mathcal{F}))$ is an $\eta$-simulable (efficiently) certifiable PRSG then $((\mathsf{KeyGen},\mathsf{StateGen},\mathsf{Ver}),(\mathcal{M},\mathcal{F}))$ is an $\eta$-simulable (efficiently) certifiable OWSG
\end{enumerate}
\end{corollary}

\section{One-way puzzles from certifiable one-way state generators}\label{s:owp-from-cert-owsg}

Equipped with the tools of the previous section, we start here by proving that one can construct a OWP from \textit{any} certifiable OWSG. This generalizes and abstracts the construction of OWP from OWSG by Khurana and Tomer~\cite{khurana2024commitmentsquantumonewayness}, which utilized the \textit{specific} ``measure first, ask later'' state certification protocol provided by global Clifford shadows. In particular, the theorem we provide below simply makes it clear that one can replace the global Clifford shadows protocol used in Ref.~\cite{khurana2024commitmentsquantumonewayness} with \textit{any} copy-efficient ``measure first, ask later'' state certification protocol for the output states of the starting OWSG. As already mentioned, this is helpful because if the state certification protocol satisfies additional properties -- such as being \textit{computationally efficient} or \textit{$\eta$-simulable} -- then the OWP will also inherit these properties. As such, the result below provides a route to the construction of \textit{variants} of OWPs -- by utilizing \textit{tailored} state certification protocols for the output states of the input OWSG --  which one could  not hope to obtain by only ever using global Clifford shadows to construct OWPs from OWSGs.

\begin{theorem}[One-way puzzles from certifiable one-way state generators]\label{thm:owp-from-certifiable-owsg} Given a certifiable one-way state generator, one can construct a one-way puzzle.
\end{theorem}

\begin{proof}Given a C-OWSG $(({\sf KeyGen},{\sf StateGen}, {\sf Ver}),(\mathcal{M},\mathcal{F}))$ with copy efficient ``measure first, ask later'' state certification protocol $(\mathcal{M},\mathcal{F})$ from $t(n,\epsilon,\delta,|\mathbb{K}|) = \mathrm{poly}(n,1/\epsilon,\log(1/\delta),\log|\mathbb{K}|)$ copies, consider the following construction.

\begin{construction}[OWP from C-OWSG]\label{con:OWP-from-OWSG-state certification} Given a C-OWSG $({\sf KeyGen},{\sf StateGen}, {\sf Ver})$ with copy efficient ``measure first, ask later'' state certification protocol $(\mathcal{M},\mathcal{F})$ from $t(n,\epsilon,\delta,|\mathbb{K}|) = \mathrm{poly}(n,1/\epsilon,\log(1/\delta),\log|\mathbb{K}|)$ copies

\begin{enumerate}
\item Define $\mathbb K_\lambda={\rm supp}({\sf KeyGen}(1^\lambda))$.
\item For all $k\in\mathbb K_\lambda$\,, define $\ket{\phi_k}$ as the $n=\poly(\lambda)$-qubit state $\ket{\phi_k} = {\sf StateGen}(k)$.
\end{enumerate}
By definition, $(\mathcal{M},\mathcal{F})$ is a copy efficient ``measure first, ask later'' state certification protocol for $\{|\phi_k\rangle\}_{k\in\mathbb K_\lambda}$. 
With this in hand, define the algorithms $({\sf SampPuzz},{\sf VerPuzz})$ as follows:
\begin{enumerate}
\item $\sf SampPuzz$. On input $1^\lambda$:
\begin{enumerate}
\item Run $\mathsf{KeyGen}(1^\lambda)$ and obtain a key $k\in \mathbb{K}_\lambda$.
\item Run $\mathsf{StateGen}(k)$ $t=t(n,\epsilon=1/8,\delta = 2^{-\lambda}, |\mathbb K_\lambda|)$ times, and construct $|\phi_k\rangle^{\otimes t}$.
\item Run $\mathcal{M}(\ket{\phi_k}^{\otimes t})$ and obtain $s$.
\item Output $(k,s)$.
\end{enumerate}
\item  $\sf VerPuzz$ is defined via
\begin{equation*}
    {\sf VerPuzz}(k,s)=\begin{cases}
        1,&\mathcal{F}(k,s,1/8,2^{-\lambda}) =\mathsf{Accept},\\
        0,&\mathcal{F}(k,s,1/8,2^{-\lambda}) =\mathsf{Reject}.
    \end{cases}
\end{equation*}
\end{enumerate}
Define the output OWP as the tuple $({\sf SampPuzz},{\sf VerPuzz})$.
\end{construction}

We now prove that the algorithms $({\sf SampPuzz},{\sf VerPuzz})$ defined in Construction~\ref{con:OWP-from-OWSG-state certification} constitute a OWP.

To this end,  we start by proving that $\sf SampPuzz$ is a QPT algorithm with respect to $\lambda$. To do this note:
\begin{enumerate}
\item As a consequence of the fact that $\sf{KeyGen}$ is a QPT algorithm one has $|\mathbb{K}_\lambda| = 2^{\mathrm{poly}(\lambda)}$ and $\log |\mathbb K_\lambda| =\poly (\lambda)$.
\item As a consequence of the fact that $\sf{StateGen}$ is a QPT algorithm one has $n=\mathrm{poly}(\lambda)$.
\item Therefore via the fact that $(\mathcal{M},\mathcal{F})$ is a copy efficient measurement protocol from 
\begin{equation}
t(n,\epsilon,\delta,|\mathbb{K}|) = \mathrm{poly}(n,1/\epsilon,\log(1/\delta),\log|\mathbb{K}|)
\end{equation}
copies,  we have that $t(n,1/8,2^{-\lambda},|\mathbb{K}_\lambda|) = \poly(\lambda)$.
\item Together with the assumption that $\mathcal{M}$ is a QPT algorithm (implicit in the assumption that $(\mathcal{M},\mathcal{F})$ is a state certification protocol), this implies that $\sf SampPuzz$ is a QPT algorithm.
\end{enumerate}

Now, let us establish the correctness and security of the pair $({\sf SampPuzz},{\sf VerPuzz})$. In what follows, for convenience we will drop the $\epsilon$ and $\delta$ arguments from $\mathcal{F}$ with the understanding that $\mathcal{F}(k,s) \coloneqq \mathcal{F}(k,s,1/8,2^{-\lambda})$ for all $(k,s)$. With this in mind, we start by noting that with $t = t(n,\epsilon=1/8,\delta = 2^{-\lambda},|\mathbb{K}_\lambda|)$ one has 

\begin{equation}\label{eq:certification-in-owp}
\underset{s\gets \mathcal{M}(|\psi\rangle^{\otimes t})}{\mathrm{Pr}}\left[\forall \, k\in \mathbb{K}_\lambda \, \begin{cases} \mathcal{F}(k,s)=\mathsf{Accept}\text{ if } |\psi\rangle = |\phi_k\rangle \\  \mathcal{F}(k,s)=\mathsf{Reject}\text{ if } |\langle \psi|\phi_k\rangle|^2 < 7/8 \end{cases}\right] > 1-2^{-\lambda}.
\end{equation}
We can now establish correctness and security.

{\em Correctness}: Using the properties of the state certification protocol $(\mathcal{M},\mathcal{F})$ we have:
    \begin{align*}
    \underset{(k,s)\,\gets\,{\sf SampPuzz}(1^\lambda)}{\rm Pr}\big[{\sf VerPuzz}(k,s)=1\big]&= \underset{\substack{
    k\,\gets\,{\sf KeyGen}(1^\lambda)\,\,\\
    s\,\gets\,{\mathcal M}({\ket{\phi_k}^t})
    }}{\rm Pr}\big[\mathcal{F}(k,s)=\mathsf{Accept}\big] \\
    &= \sum_{k\in \mathbb{K}_\lambda}\left[\mathrm{Pr}(k)\left(\underset{s\,\gets\,{\mathcal M}({\ket{\phi_k}^t})
    }{\rm Pr}\big[\mathcal{F}(k,s)=\mathsf{Accept}\big]\right)\right] \\
    &\geq \sum_{k\in \mathbb{K}_\lambda}\left[\mathrm{Pr}(k)\left(\underset{s\,\gets\,{\mathcal M}({\ket{\phi_k}^t})
    }{\mathrm{Pr}}\left[\forall \, j\in \mathbb{K}_\lambda \, \begin{cases} \mathcal{F}(j,s)=\mathsf{Accept}\text{ if } |\phi_k\rangle = |\phi_j\rangle \\  \mathcal{F}(j,s)=\mathsf{Reject}\text{ if } |\langle \phi_k|\phi_j\rangle|^2 < 7/8 \end{cases}\right]\right)\right]\\
    &> \sum_{k\in \mathbb{K}_\lambda}\left[\mathrm{Pr}(k)(1-2^{-\lambda})\right]\\
    &= 1- 2^{-\lambda}\,.
    \end{align*}

{\em Security}: Note that
    \begin{align}
    \underset{\substack{
    (k,s)\,\gets\,{\sf SampPuzz}(1^\lambda)\,\,\\
    k'\,\gets\,{\mathcal A}(s)
    }}{\rm Pr}\big[{\sf VerPuzz}(k',s)=1\big]&=\underset{\substack{
    (k,s)\,\gets\,{\sf SampPuzz}(1^\lambda)\,\,\\
    k'\,\gets\,{\mathcal A}(s)
    }}{\rm Pr}\big[\mathcal{F}(k',s)=\mathsf{Accept}\big]\\
    &=\underset{\substack{
    (k,s)\,\gets\,{\sf SampPuzz}(1^\lambda)\,\,\\
    k'\,\gets\,{\mathcal A}(s)
    }}{\rm Pr}\left[A(k,s,k')\lor B(k,s,k')\right],\nonumber
    \end{align}
    where
    \begin{equation}
    \begin{aligned}
        A(k,s,k')&\coloneqq\left[\mathcal{F}(k',s)=\mathsf{Accept}\right]\land\left[|\braket{\phi_{k'}}{\phi_k}|^2<7/8\right],\\
        B(k,s,k')&\coloneqq\left[\mathcal{F}(k',s)=\mathsf{Accept}\right]\land\left[|\braket{\phi_{k'}}{\phi_k}|^2\ge7/8\right].
    \end{aligned}
    \end{equation}
Now, using the guarantee of the state certification protocol $(\mathcal{M},\mathcal{F})$, one has
\begin{align}
\underset{\substack{
    (k,s)\,\gets\,{\sf SampPuzz}(1^\lambda)\,\,\\
    k'\,\gets\,{\mathcal A}(s)
    }}{\rm Pr}\big[A(k,s,k')\big] &= \underset{\substack{
    k\,\gets\,{\sf KeyGen}(1^\lambda)\,\,\\
    s\,\gets\,{\mathcal M}({\ket{\phi_k}^t})\\k'\gets\mathcal{A}(s)
    }}{\rm Pr}\Big[\left[\mathcal{F}(k',s)=\mathsf{Accept}\right]\land\left[|\braket{\phi_{k'}}{\phi_k}|^2<7/8\right]\Big]\\
      \nonumber
    &\leq \underset{\substack{
    k\,\gets\,{\sf KeyGen}(1^\lambda)\,\,\\
    s\,\gets\,{\mathcal M}({\ket{\phi_k}^t})
    }}{\rm Pr}\Big[\exists\, k'\in \mathbb{K}_\lambda \left(\left[\mathcal{F}(k',s)=\mathsf{Accept}\right]\land\left[|\braket{\phi_{k'}}{\phi_k}|^2<7/8\right]\right)\Big]\\
      \nonumber
    &\leq 2^{-\lambda}.
\end{align}
Additionally, 
\begin{align}
\underset{\substack{
    (k,s)\,\gets\,{\sf SampPuzz}(1^\lambda)\,\,\\
    k'\,\gets\,{\mathcal A}(s)
    }}{\rm Pr}\big[B(k,s,k')\big] &= \underset{\substack{
    k\,\gets\,{\sf KeyGen}(1^\lambda)\,\,\\
    s\,\gets\,{\mathcal M}({\ket{\phi_k}^t})\\k'\gets\mathcal{A}(s)
    }}{\rm Pr}\Big[\left[\mathcal{F}(k',s)=\mathsf{Accept}\right]\land\left[|\braket{\phi_{k'}}{\phi_k}|^2\ge7/8\right]\Big]\\
    &\leq \underset{\substack{
    k\,\gets\,{\sf KeyGen}(1^\lambda)\,\,
      \nonumber
      \\
     \nonumber
    s\,\gets\,{\mathcal M}({\ket{\phi_k}^t})\\k'\gets\mathcal{A}(s)
    }}{\rm Pr}\left[|\braket{\phi_{k'}}{\phi_k}|^2\ge7/8\right]\\
    &\coloneqq \beta.  \nonumber
\end{align}
We now show that $\beta\le{\rm negl}(\lambda)$. To see this, recall from Definition~\ref{def:owsg} that ${\sf Ver}(k,\ket\psi)$ corresponds to the outcome of the projective measurement $\{\ketbra{\phi_k}{\phi_k},\mathbbm{1}-\ketbra{\phi_k}{\phi_k}\}$ applied to the state $\ket\psi$. Then, for all QPT algorithms $\mathcal{A}$ one has
\begin{equation}
{\rm negl}(\lambda)\ge\underset{\substack{
    k\,\gets\,{\sf KeyGen}(1^\lambda)\,\,\\
    k'\,\gets\,\mathcal A\circ\mathcal M(\ket{\phi_k}^{\otimes t})
    }}{\rm Pr}\Big[{\sf Ver}(k',\ket{\phi_k})=1\Big]=\underset{\substack{
    k\,\gets\,{\sf KeyGen}(1^\lambda)\,\,\\
    k'\,\gets\,\mathcal A\circ\mathcal M(\ket{\phi_k}^{\otimes t})
    }}{\mathbb E}\Big[|\braket{\phi_{k'}}{\phi_k}|^2\Big]\ge\frac{7}{8}\beta
\end{equation}
where the first inequality follows from the security of the OWSG, and the final inequality follows from Markov's inequality. Therefore, $\beta\le{\rm negl}(\lambda)$. By a union bound, we then have
\begin{equation}
\underset{\substack{
    (k,s)\,\gets\,{\sf SampPuzz}(1^\lambda)\,\,\\
    k'\,\gets\,{\mathcal A}(s)
    }}{\rm Pr}\big[{\sf VerPuzz}(k',s)=1\big] \leq \mathrm{negl}(\lambda).
\end{equation}
\end{proof}

With the above established, we obtain as promised an immediate corollary that if the input OWSG is \textit{efficiently} certifiable, then Construction~\ref{con:OWP-from-OWSG-state certification} actually yields an \textit{efficiently verifiable} OWP. Specifically, we have the following.

\begin{corollary}[Efficiently verifiable one-way puzzle from efficiently certifiable one-way state generator]\label{cor:ev-owp-from-eff-certifiable-owsg} Given an efficiently certifiable one-way state generator one can construct an efficiently verifiable one-way puzzle.
\end{corollary}

\section{One-way functions from simulable certifiable one-way state generators}\label{s:owf-from-simulable-owsg}

We have already seen in the previous section that given an (efficiently) certifiable OWSG, one can construct an (efficiently verifiable) OWP. In this section, we show the following:
\begin{enumerate}
\item If one is given a $1/3$-simulable certifiable OWSG (with a classical $\mathsf{KeyGen}$ algorithm), then one can construct a quantum-secure OWF from the OWP obtained from the OWSG. More specifically, we will show that one can construct a suitably accurate classical sampling algorithm for the OWP, and therefore via the results of Corollary~\ref{cor:classical-sampling-implies-qsecDOWF} one can construct a quantum-secure distributional OWF, from which via Theorem~\ref{thm:qsecDWOFsimplyqsecOWFS} one can construct a quantum-secure OWF.
\item If one is given an $\eta$-simulable \textit{efficiently} certifiable OWSG (with a classical $\mathsf{KeyGen}$ algorithm), for a function $\eta$ which is negligible in a sense to be made precise, then one can construct a simple quantum-secure OWF directly from the OWP obtained by the OWSG, \textit{without} having to go via a distributional OWF. In particular, this OWF avoids the complexities incurred in the previous construction due to first compiling the distributional OWF into a weak OWF, and then compiling this weak OWF into a OWF. Essentially we show here that if one has a measurement protocol which is both efficiently certifiable, and $\eta$ simulable for a ``negligible'' $\eta$, then one can obtain a simpler OWF.
\end{enumerate}

\subsection{OWFs from simulable certifiable OWSGs via distributional OWFs}\label{ss:owf-via-distributional-owfs}

We prove in this section that given a $1/3$-simulable certifiable OWSG, with a classical PPT key generation algorithm $\mathsf{KeyGen}$, one can construct a quantum-secure OWF. 

\begin{theorem}[OWF from $1/3$-simulable certifiable OWSG with PPT $\mathsf{KeyGen}$]\label{thm:owf-via-sim-certOWSG} Given a 1/3-simulable certifiable OWSG $((\mathsf{KeyGen},\mathsf{StateGen},\mathsf{Ver}),(\mathcal{M},\mathcal{F}))$, in which $\mathsf{KeyGen}$ is a classical PPT algorithm, one can construct a quantum-secure one-way function.
\end{theorem}

\begin{proof} The idea of the proof is the following:
\begin{enumerate}
\item We know from Theorem~\ref{thm:owp-from-certifiable-owsg} that given a certifiable OWSG $((\mathsf{KeyGen},\mathsf{StateGen},\mathsf{Ver}),(\mathcal{M},\mathcal{F}))$, one can construct a OWP $(\mathsf{SampPuzz},\mathsf{VerPuzz})$ via Construction~\ref{con:OWP-from-OWSG-state certification}.
\item Therefore, if one can use the additional 1/3-simulability of $(\mathcal{M},\mathcal{F})$ to construct a classical PPT algorithm $\mathsf{SampPuzz_C}$ satisfying
\begin{equation}\label{eq:sampling-accuracy}
d_{\mathrm{TV}}(\mathsf{SampPuzz}(1^\lambda),\mathsf{SampPuzz}_C(1^\lambda)) \leq 1/3,
\end{equation}
then it follows from Corollary~\ref{cor:classical-sampling-implies-qsecDOWF} that one can construct an explicit quantum-secure distributional OWF.
\item Given the quantum-secure distributional OWF, it then follows from Theorem~\ref{thm:qsecDWOFsimplyqsecOWFS} that one can construct a quantum-secure OWF.
\end{enumerate}
As such, we show that one can construct a classical PPT algorithm satisfying Eq.~\eqref{eq:sampling-accuracy}. To  this end, recall from Construction~\ref{con:OWP-from-OWSG-state certification} that $\mathsf{SampPuzz}$ is defined as follows, on input $1^\lambda$:
\begin{enumerate}
\item $k\gets\mathsf{KeyGen}(1^\lambda)$ (where $\mathsf{KeyGen}$ is assumed to be classical PPT).
\item $t\gets t(n(\lambda),\epsilon=1/8,\delta = 2^{-\lambda}, |\mathbb K_\lambda|)$, where $t(n,\epsilon,\delta, |\mathbb{K}_\lambda|)$ is the copy complexity of $(\mathcal{M},\mathcal{F})$. 
\item $|\phi_k\rangle \gets \mathsf{StateGen}(k)$ 
\item $s\gets \mathcal{M}(|\phi_k\rangle^{\otimes t})$.
\item Output $(k,s)$.
\end{enumerate}
The natural idea is then to construct $\mathsf{SampPuzz}_C$ by simply replacing steps 3 and 4 above with the implementation of $\mathcal{M}_C(k,1^n,1^t)$, where $\mathcal{M}_C$ is the $\eta$-accurate simulation of $\mathcal{M}$ guaranteed by $\eta$-simulability of $(\mathcal{M},\mathcal{F})$. Specifically, define $\mathsf{SampPuzz}_C$, on input $1^\lambda$, via:
\begin{enumerate}
\item $k\gets\mathsf{KeyGen}(1^\lambda)$. 
\item $t\gets t(n(\lambda),\epsilon=1/8,\delta = 2^{-\lambda}, |\mathbb K_\lambda|)$.
\item $s\gets \mathcal{M}_{\cal C}(k,1^{n(\lambda)},1^t)$.
\item Output $(k,s)$.
\end{enumerate}

Given this, we start by proving that $\mathsf{SampPuzz}_C$ is a PPT algorithm. To see this:
\begin{enumerate}
\item Note that we have assumed that $\mathsf{KeyGen}$ is a classical PPT algorithm, and recall from Section~\ref{ss:microcrypt-primitives} that efficiency of both $\mathsf{KeyGen}$ and $\mathsf{StateGen}$ implies that $\mathbb{K}_\lambda\subseteq\{0,1\}^{O(\mathrm{poly}(\lambda))}$, and therefore both $|k| = O(\mathrm{poly}(\lambda))$ and $n(\lambda)=O(\mathrm{poly}(\lambda))$.
\item With the above, it follows from the copy efficiency of $(\mathcal{M},\mathcal{F})$ that $t(n(\lambda),\epsilon=1/8,\delta = 2^{-\lambda}, |\mathbb K_\lambda|) = O(\mathrm{poly}(\lambda))$.
\item Therefore, it follows from the fact that $\mathcal{M}_C$ is a PPT algorithm with a runtime on input $(k,1^n,1^t)$ of \\$O(\mathrm{poly}(|k|,n,t)) = O(\mathrm{poly}(\lambda))$.
\item As a result, the runtime of $\mathsf{SampPuzz}_C$ is $O(\mathrm{poly}(\lambda))$, and it is also a  PPT algorithm.
\end{enumerate}
With this established, we now have to show that $\mathsf{SampPuzz}_C$ is sufficiently accurate. To do this, recall that 
 $1/3$-simulability of $(\mathcal{M},\mathcal{F})$ enforces
\begin{equation}
d_\mathrm{TV}(\mathcal{M}_C(k,1^n,1^t),\mathcal{M}(|\phi_k\rangle^{\otimes t})) \leq 1/3,
\end{equation}
for all $(k,n,t)$, from which it follows that
\begin{equation}
d_\mathrm{TV}(\mathsf{SampPuzz}(1^\lambda),\mathsf{SampPuzz}_C(1^\lambda)) \leq 1/3,
\end{equation}
by virtue of the fact that $k$ is sampled identically in both $\mathsf{SampPuzz}$ and $\mathsf{SampPuzz}_C$.
\end{proof}

\subsection{Direct OWF construction via efficient verifiability}\label{ss:owf-directly-via-EV}

In the previous section, we have shown that given a 1/3-simulable certifiable OWSG $((\mathsf{KeyGen},\mathsf{StateGen},\mathsf{Ver}),(\mathcal{M},\mathcal{F}))$, with classical PPT $\mathsf{KeyGen}$, one can construct a quantum-secure OWF. However, the construction of this OWF involves multiple steps:
\begin{enumerate}
\item The construction of a OWP $(\mathsf{SampPuzz},\mathsf{VerPuzz})$, via Construction~\ref{con:OWP-from-OWSG-state certification}.
\item The construction of a $1/3$-accurate classical PPT sampling algorithm $\mathsf{SampPuzz}_C$. 
\item The construction of a quantum-secure \textit{distributional} OWF $\{f_\lambda\}$ from $\mathsf{SampPuzz_C}$ (As per Corollary~\ref{cor:classical-sampling-implies-qsecDOWF}).
\item The construction of a quantum-secure one-way function $\{g_\lambda\}$ via the construction underlying Theorem~\ref{thm:qsecDWOFsimplyqsecOWFS}. In  particular, this construction proceeds via first constructing a quantum-secure \textit{weak} OWF $\{h_\lambda\}$, and then constructing the quantum-secure OWF $\{g_\lambda\}$ from $\{h_\lambda\}$.
\end{enumerate}
In this section, we show that if one starts from an $\eta$-simulable \textit{efficiently certifiable} OWSG (with classical PPT $\mathsf{KeyGen}$), for a function $\eta$ which is negligible in a sense that we will make precise shortly, then a simple modification of the quantum-secure distributional OWF constructed in Step 3 above is immediately a quantum-secure OWF. In particular,  one can obtain a relatively simple OWF without the chain of compilations 
\begin{equation}
\text{distributional OWF $\rightarrow$ weak OWF $\rightarrow$ OWF}
\end{equation}
that was required in the previous section. To do this:
\begin{enumerate}
\item We start by showing in Theorem~\ref{thm:owf-from-classical-evowp} that given an \textit{efficiently verifiable} OWP with a \textit{negligibly} inaccurate classical sampling algorithm, a simple modification of the distributional OWF defined in Theorem~\ref{thm:KTowf} and Corollary~\ref{cor:classical-sampling-implies-qsecDOWF} is immediately a quantum-secure OWF.
\item We already know from Corollary~\ref{cor:ev-owp-from-eff-certifiable-owsg} that one can construct an EV-OWP from an efficiently certifiable OWSG. So, we simply show (using essentially the proof of Theorem~\ref{thm:owf-via-sim-certOWSG}) that when the input efficiently certifiable OWSG is also $\eta$-simulable, for appropriate $\eta$, then one can construct a negligibly inaccurate classical PPT sampling algorithm for the EV-OWP, and therefore by the previous point, directly construct a quantum-secure OWF.
\end{enumerate}

To achieve point 1 above, it is helpful to understand why the arguments establishing \textit{distributional} one-wayness of the function family $\{f_\lambda\}$ in Theorem~\ref{thm:KTowf}, fail to prove that the same function is a OWF. To this end, recall from Theorem~\ref{thm:KTowf} that given a OWP $(\mathsf{Samp},\mathsf{Ver})$ and an efficient deterministic \textit{classical} algorithm $\mathsf{Samp}_C$ satisfying
\begin{equation}
d_{\mathrm{TV}}\left(\mathsf{Samp}(1^\lambda),\mathsf{Samp}_C(1^\lambda,r)_{r\gets\{0,1\}^{m(\lambda)}}\right) \leq \frac{1}{3},
\end{equation}
with $\mathsf{Samp}_C(1^\lambda,r) = (k_{\lambda}(r),s_\lambda(r))$, then the function family $\{f_\lambda\}$ considered by Khurana and Tomer (and shown to be a \textit{distributional} OWF) is defined via $f_\lambda(r) = s_\lambda(r)$. 

As an illustration, let's now imagine we have the even stronger assumption that
\begin{equation}\label{eq:perfect-tv}
\mathsf{Samp}(1^\lambda) =\mathsf{Samp}_C(1^\lambda,r)_{r\gets\{0,1\}^{m(\lambda)}},
\end{equation}
so that \textit{correctness} of the OWP becomes
\begin{equation}\label{eq:new-correctness}
\underset{(k,s)\gets\mathsf{Samp}(1^\lambda)}{\mathrm{Pr}}\left[\mathsf{Ver}(k,s)=1\right] = \underset{r\gets\{0,1\}^{m(\lambda)} }{\mathrm{Pr}}\left[\mathsf{Ver}(k_\lambda(r)_,s_\lambda(r))=1\right] \geq 1-\mathsf{negl},
\end{equation}
and try to prove that $\{f_{\lambda}\}$ is in fact a standard OWF. The natural approach is to assume that there exists some successful adversary $\mathcal{A}$ for the OWF, and show that this can be used to construct a successful adversary $\mathcal{B}$ for the underlying OWP. So, let's try to do this, and to make our job even easier, let's assume we have an extremely strong adversary for the OWF, which \textit{always} succeeds, i.e., a PPT algorithm $\mathcal{A}$ satisfying
\begin{equation}\label{eq:owf-adversary}
\underset{\substack{r\gets\{0,1\}^{m(\lambda)}\\r'\gets\mathcal{A}(1^\lambda,f_\lambda(r))}}{\mathrm{Pr}}\left[r'\in f^{-1}_{\lambda}(f_\lambda(r))\right] = \underset{\substack{r\gets\{0,1\}^{m(\lambda)}\\r'\gets\mathcal{A}(1^\lambda,s_\lambda(r))}}{\mathrm{Pr}}\left[r'\in s^{-1}_{\lambda}(s_\lambda(r))\right] =1.
\end{equation}
Then, we define our candidate OWP adversary $\mathcal{B}$ such that on input $(1^\lambda, s)$ with $(k,s)\gets\mathsf{Samp}(1^\lambda)$:
\begin{enumerate}
\item $r'\gets \mathcal{A}(1^\lambda,s)$,
\item $k'\gets k_\lambda(r')$.
\end{enumerate}
Indeed, note that if $\mathcal{A}$ is such that for all $r\in\{0,1\}^{m(\lambda)}$ one has $\mathcal{A}(1^\lambda,s_\lambda(r)) = r$, then using Eq.~\eqref{eq:new-correctness}, one has 
\begin{align}
\underset{(k,s)\gets\mathsf{Samp}(1^\lambda) }{\mathrm{Pr}}\left[\mathsf{Ver}(\mathcal{B}(1^\lambda,s),s)=1\right] &= \underset{r\gets\{0,1\}^{m(\lambda)} }{\mathrm{Pr}}\left[\mathsf{Ver}(\mathcal{B}(1^\lambda,s_\lambda(r)),s_\lambda(r))=1\right]\\
&= \underset{r\gets\{0,1\}^{m(\lambda)} }{\mathrm{Pr}}\left[\mathsf{Ver}(k_\lambda(\mathcal{A}(1^\lambda,s_\lambda(r))),s_\lambda(r))=1\right]\\
&=\underset{r\gets\{0,1\}^{m(\lambda)} }{\mathrm{Pr}}\left[\mathsf{Ver}(k_\lambda(r),s_\lambda(r))=1\right]\\
&\geq 1-\mathsf{negl}(\lambda),
\end{align}
i.e., $\mathcal{B}$ is a (very) succesful adversary for the OWP. As such, it seems intuitive/plausible that one could use Eq.~\eqref{eq:owf-adversary} to show that
\begin{align}
\underset{(k,s)\gets\mathsf{Samp}(1^\lambda) }{\mathrm{Pr}}\left[\mathsf{Ver}(\mathcal{B}(1^\lambda,s),s)\right] &= \underset{r\gets\{0,1\}^{m(\lambda)} }{\mathrm{Pr}}\left[\mathsf{Ver}(\mathcal{B}(1^\lambda,s_\lambda(r)),s_\lambda(r))\right] \\
&=\underset{\substack{r\gets\{0,1\}^{m(\lambda)} \\r'\gets\mathcal{A}(1^\lambda,s_\lambda(r))}}{\mathrm{Pr}}\left[\mathsf{Ver}(k_\lambda(r'),s_\lambda(r))\right]\\
&\geq \tilde{g}(\lambda)
\end{align}
for some non-negligible function $\tilde{g}$. Unfortunately, however, there is a problem! To see this, let's define the sets
$\mathrm{Good}_\lambda$ and $\mathrm{Bad}_\lambda$ via 
\begin{align}
\mathrm{Good}_\lambda &= \{r\,|\, \mathsf{Ver}(k_\lambda(r),s_\lambda(r))=1\},\\
\mathrm{Bad}_\lambda &= \{r\,|\, \mathsf{Ver}(k_\lambda(r),s_\lambda(r))=0\}.
\end{align}
Note that the correctness condition of the OWP then becomes
\begin{equation}
\underset{r\gets\{0,1\}^{m(\lambda)} }{\mathrm{Pr}}\left[r\in \mathrm{Good}_\lambda\right] \geq 1-\mathsf{negl}.
\end{equation}
In particular, it can be the case that $\mathrm{Bad}_\lambda$ is non-empty. Let's assume that this is the case, and note that nothing prevents an adversary $\mathcal{A}$ which, as illustrated in Figure~\ref{fig:owftoowp}, in addition to satisfying Eq.~\eqref{eq:owf-adversary} is fine-tuned so that with high probability it also outputs some $r'\in\mathsf{Bad}_\lambda$, i.e. 
\begin{equation}\label{eq:badoutput}
\underset{\substack{r\gets\{0,1\}^{m(\lambda)}\\r'\gets\mathcal{A}(1^\lambda,s_\lambda(r))} }{\mathrm{Pr}}\left[r'\in \mathrm{Bad}_\lambda\right] \geq 1-\mathsf{negl}.
\end{equation}
Then, using both Eq~\eqref{eq:owf-adversary} and Eq.~\eqref{eq:badoutput} we have
\begin{align}
\underset{(k,s)\gets\mathsf{Samp}(1^\lambda) }{\mathrm{Pr}}\left[\mathsf{Ver}(\mathcal{B}(1^\lambda,s),s)=1\right] &= \underset{r\gets\{0,1\}^{m(\lambda)} }{\mathrm{Pr}}\left[\mathsf{Ver}(\mathcal{B}(1^\lambda,s_\lambda(r)),s_\lambda(r))=1\right]\\
&=\underset{r\gets\{0,1\}^{m(\lambda)} }{\mathrm{Pr}}\left[\mathsf{Ver}(k_\lambda(\mathcal{A}(1^\lambda,s_\lambda(r))),s_\lambda(r))=1\right]\\
&= \underset{\substack{r\gets\{0,1\}^{m(\lambda)}\\r'\gets\mathcal{A}(1^\lambda,s_\lambda(r))} }{\mathrm{Pr}}\left[\mathsf{Ver}(k_\lambda(r'),s_\lambda(r))=1\right]\\
&= \underset{\substack{r\gets\{0,1\}^{m(\lambda)}\\r'\gets\mathcal{A}(1^\lambda,s_\lambda(r))} }{\mathrm{Pr}}\left[\mathsf{Ver}(k_\lambda(r'),s_\lambda(r'))=1\right]\\
&=\underset{\substack{r\gets\{0,1\}^{m(\lambda)}\\r'\gets\mathcal{A}(1^\lambda,s_\lambda(r))} }{\mathrm{Pr}}\left[r'\in \mathsf{Good}_\lambda\right]\\
&<\mathsf{negl}(\lambda),
\end{align}
i.e., $\mathcal{B}$ is \textit{not} a successful adversary for the OWP, and our proof strategy fails! As such, we see that if we want to build an adversary for the OWP from an adversary for the OWF, it would be helpful if we could somehow ``force'' the OWF adversary to always output some $r'\in\mathrm{Good}_\lambda$. Luckily, this can be achieved by simply considering the slightly modified function family $\{F_\lambda:\{0,1\}^{m(\lambda)}\rightarrow\{0,1\}^*\}_{\lambda\in\mathbb{N}}$ defined via
\begin{equation}
F_\lambda(r) = \begin{cases} (0,s_\lambda(r)) & r\in\mathrm{Good}_\lambda,
\\
(1,r) & r\in\mathrm{Bad}_\lambda,
\end{cases}
\end{equation}
In particular, with this construction, if $r\in\mathrm{Good}_{\lambda}$, then on input $(0,s_\lambda(r))$, a successful OWF adversary has to output some $r'\in\mathsf{Good}_\lambda$, otherwise $F_\lambda(r) \neq F_\lambda(r')$. One might immediately object that on inputs $r\in \mathrm{Bad}_\lambda$, the function is trivial to invert, but this is not problematic given that correctness enforces that this happens only with negligible probability. However, it is important to note that in order for $\{F_\lambda\}$ to be efficient to compute, we have to also insist that $\mathsf{Ver}$ is efficient -- i.e., that the underlying OWP is an \textit{efficiently verifiable} OWP.

\begin{figure}
    \centering
     \includegraphics[width=0.55\linewidth]{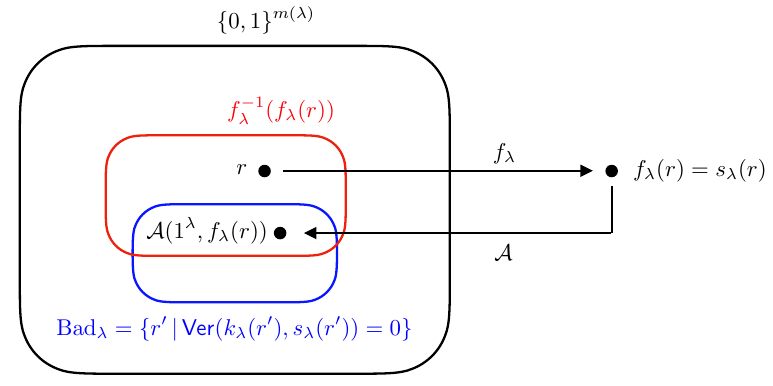} 
    \caption{An illustration of a ``problematic'' adversary $\mathcal{A}$ for the function $\{f_\lambda = s_\lambda\}$, which on input $s_\lambda(r)$ outputs some $r'\in f_\lambda^{-1}(f_\lambda(r))\cap \mathrm{Bad}_\lambda$.} 
    \label{fig:owftoowp} 
\end{figure}

Using these ideas, we then have the following:

\begin{theorem}[One-way functions from classical efficiently verifiable one-way puzzles]\label{thm:owf-from-classical-evowp}
Let $m:\mathbb{N}\rightarrow\mathbb{N}$ be some fixed polynomial. Given an efficiently verifiable one-way puzzle $(\mathsf{Samp},\mathsf{Ver})$, assume there exists an efficient deterministic classical algorithm $\mathsf{Samp}_\mathrm{C}$ such that
\begin{equation}
d_\mathrm{TV}\left(\mathsf{Samp}(1^\lambda),\mathsf{Samp}_{\mathrm{C}}(1^\lambda,r)_{r\gets\{0,1\}^{m(\lambda)}}\right)\leq \mathsf{negl}(\lambda)
\end{equation}
where again $\mathsf{Samp}_{\mathrm{C}}(1^\lambda,r)_{r\gets\{0,1\}^{m(\lambda)}} = (k_\lambda(r),s_\lambda(r))_{r\gets\{0,1\}^{m(\lambda)}}$ is understood as the distribution over outputs of $\mathsf{Samp}_{\mathrm{C}}(1^\lambda,r)$ with respect to bitstrings $r\in \{0,1\}^{m(\lambda)}$ drawn uniformly at random. Then, the family of functions $\{F_\lambda:\{0,1\}^{m(\lambda)}\rightarrow\{0,1\}^*\}_{\lambda\in\mathbb{N}}$ defined via
\begin{equation}
F_\lambda(r) = \begin{cases} (0,s_\lambda(r)) & r\in\mathrm{Good}_\lambda,
\\
(1,r) & r\in\mathrm{Bad}_\lambda,
\end{cases}
\end{equation}
is a quantum-secure one-way function. 
\end{theorem}

\begin{proof} Let $(\mathsf{Samp},\mathsf{Ver})$ be an EV-OWP. We start with a simple "robustness lemma", whose proof we defer to Appendix~\ref{app:robustness-proof}, showing that replacing $\mathsf{Samp}$ with any negligibly close approximation in TV distance yields a new OWP.

\begin{lemma}[Robustness of one-way puzzles]\label{lem:owp-robustness}
Let $(\mathsf{Samp},\mathsf{Ver})$ be a one-way puzzle, and let $\mathsf{Samp}'$ be any QPT or PPT algorithm satisfying 
\begin{equation}
\epsilon(\lambda) = d_\mathrm{TV}(\mathsf{Samp}(1^\lambda),\mathsf{Samp}'(1^\lambda)) \leq \mathsf{negl}(\lambda). 
\end{equation}
Then, $(\mathsf{Samp'},\mathsf{Ver})$ is a one-way puzzle.
\end{lemma}
With this in hand, it follows immediately that if, as per the theorem statement, $\mathsf{Samp}_{\mathrm{C}}$ satisfies 
\begin{equation}
d_\mathrm{TV}(\mathsf{Samp}(1^\lambda),\mathsf{Samp}_{\mathrm{C}}(1^\lambda,r)_{r\gets\{0,1\}^{m(\lambda)}})\leq \mathsf{negl}(\lambda)
\end{equation}
then $(\mathsf{Samp}_\mathrm{C},\mathsf{Ver})$ is an (efficiently verifiable) OWP. Now, as already noted, the efficiency of $F_\lambda$ follows from the efficiency of $\mathsf{Samp}_{\mathrm{C}}$ and $\mathsf{Ver}$. To obtain a contradiction, let's assume that $\{F_\lambda\}$ is \textit{not} a quantum-secure OWF. In this case, there exists some QPT inversion algorithm $\mathcal{A}$ satisfying 
\begin{equation}
\underset{r\gets\{0,1\}^{m(\lambda)}}{\mathrm{Pr}}\left[F_\lambda\left(\mathcal{A}(1^\lambda,F_\lambda(r))\right) = F_\lambda(r)\right] \geq \epsilon(\lambda)
\end{equation}
for some non-negligible $\epsilon$. Alternatively, if we define the event 
\begin{equation}
\mathrm{Succ}_\mathcal{A} = \left\{F_\lambda \left(\mathcal{A}(1^\lambda,F_\lambda(r))
\right) = F_\lambda(r)\right\},
\end{equation}
then
\begin{equation}
\underset{r\gets\{0,1\}^{m(\lambda)}}{\mathrm{Pr}}[\mathrm{Succ}_\mathcal{A}] \geq \epsilon(\lambda).
\end{equation}
Using this, we now construct a QPT solver $\mathcal{B}$ for the OWP $(\mathsf{Samp}_{\mathrm{C}},\mathsf{Ver})$. In particular, define $\mathcal{B}$ as follows:
\begin{enumerate}
\item On input $s$ and $1^\lambda$, obtain $r'\gets \mathcal{A}(1^\lambda,(0,s))$.
\item Output $k' \gets k_\lambda(r')$.
\end{enumerate}Note that in the case that $s= s_\lambda(r)$ for some $r\in \mathrm{Good}_\lambda$, then by construction of the OWF, if $\mathcal{A}$ succeeds in Step 1, then we know that $s_\lambda(r') = s_\lambda(r)$ \textit{and} that  $r'\in \mathrm{Good}_\lambda$, and therefore that
\begin{equation}
\mathsf{Ver}\left(\mathcal{B}(s_\lambda(r)),s_\lambda(r)\right) = \mathsf{Ver}\left(k_\lambda(r'),s_\lambda(r)\right) = \mathsf{Ver}\left(k_\lambda(r'),s_\lambda(r')\right) = 1.
\end{equation}
With this in mind, note that using the correctness condition of $(\mathsf{Samp}_{\mathrm{C}},\mathsf{Ver})$, namely that 
\begin{equation}
\underset{r\gets\{0,1\}^{m(\lambda)}}{\mathrm{Pr}}\left[\mathsf{Ver}\left(k_\lambda(r),s_\lambda(r)\right)=1\right] = \underset{r\in\{0,1\}^{m(\lambda)}}{\mathrm{Pr}}[r\gets \mathrm{Good}_\lambda] \geq 1- \delta(\lambda)
\end{equation}
for some negligible $\delta$, we have that 
\begin{align}
\underset{r\gets\{0,1\}^{m(\lambda)}}{\mathrm{Pr}}[\mathsf{Ver}(\mathcal{B}(s_\lambda(r)),s_\lambda(r)) =1]&\geq \mathrm{Pr}[(r\in \mathrm{Good}_\lambda) \land \mathrm{Succ}_\mathcal{A}] \\
\nonumber
&= \mathrm{Pr}[\mathrm{Succ}_\mathcal{A}] - \mathrm{Pr}[(r\notin \mathrm{Good}_\lambda)\land \mathrm{Succ}_\mathcal{A} ]\\
\nonumber
&\geq \mathrm{Pr}[\mathrm{Succ}_\mathcal{A}] - \mathrm{Pr}[(r\notin \mathrm{Good}_\lambda)]\\
\nonumber
&\geq \epsilon(\lambda) - \delta(\lambda),
\nonumber
\end{align}
which is non-negligible, contradicting the assumption that $(\mathsf{Samp}_{\mathrm{C}},\mathsf{Ver})$ is a OWP.
\end{proof}

With this established, we can now prove the following result:

\begin{theorem}[OWF directly from negl-simulable efficiently certifiable OWSG ]\label{thm:owf-via-sim-effcertOWSG} Given an $\eta$-simulable efficiently certifiable one-way state generator $((\mathsf{KeyGen},\mathsf{StateGen},\mathsf{Ver}),(\mathcal{M},\mathcal{F}))$, in which $\mathsf{KeyGen}$ is a classical PPT algorithm, if 
\begin{equation}\label{eq:negl-condition}
\sup_{k\in\mathbb{K}_\lambda}\eta\Big(|k|, n(\lambda), t\big(n(\lambda),\epsilon=1/8,\delta=2^{-\lambda}, |\mathbb{K}_\lambda|\big)\Big) = \mathsf{negl}(\lambda),
\end{equation}
then one can construct an efficient deterministic classical algorithm $\mathsf{Samp}_\mathrm{C}$, with $\mathsf{Samp}_{\mathrm{C}}(1^\lambda,r) = (k_\lambda(r),s_\lambda(r))$,  such that
\begin{equation}
d_\mathrm{TV}\left(\mathsf{SampPuzz}(1^\lambda),\mathsf{SampPuzz}_{\mathrm{C}}(1^\lambda,r)_{r\gets\{0,1\}^{m(\lambda)}}\right)\leq \mathsf{negl}(\lambda),
\end{equation}
where $(\mathsf{SampPuzz},\mathsf{VerPuzz})$ is the EV-OWP constructed from $((\mathsf{KeyGen},\mathsf{StateGen},\mathsf{Ver}),(\mathcal{M},\mathcal{F}))$ via Construction~\ref{con:OWP-from-OWSG-state certification}.
Then, the family of functions $\{F_\lambda:\{0,1\}^{m(\lambda)}\rightarrow\{0,1\}^*\}_{\lambda\in\mathbb{N}}$ defined via
\begin{equation}
F_\lambda(r) = \begin{cases} (0,s_\lambda(r)) & r\in\mathrm{Good}_\lambda,
\\
(1,r) & r\in\mathrm{Bad}_\lambda,
\end{cases}
\end{equation}
is a quantum-secure one-way function. 
\end{theorem}

\begin{proof} Define the algorithm $\mathsf{SampPuzz}_C$ exactly as per the proof of Theorem~\ref{thm:owf-via-sim-certOWSG}. This algorithm is PPT by the same arguments given in the proof of Theorem~\ref{thm:owf-via-sim-certOWSG}. Also as per the proof of Theorem~\ref{thm:owf-via-sim-certOWSG}, it follows from the $\eta$-simulability of $(\mathcal{M},\mathcal{F})$ that 
\begin{equation}
d_\mathrm{TV}(\mathcal{M}_C(k,1^n,1^t),\mathcal{M}(|\phi_k\rangle^{\otimes t})) \leq \eta(|k|,n,t),
\end{equation}
for all $(k,n,t)$. However, in this case, this implies that
\begin{align}
d_\mathrm{TV}(\mathsf{SampPuzz}(1^\lambda),\mathsf{SampPuzz}_C(1^\lambda)) &\leq \mathbb{E}_{k\sim\mathsf{KeyGen}(1^\lambda)}\left[\eta\left(|k|,n(\lambda),t(\lambda)\right)\right]
\\
&\leq \sup_{k\in\mathbb{K}_\lambda}\left[\eta\left(|k|,n(\lambda),t(\lambda)\right)\right]\\
&=\mathsf{negl(\lambda)},
\end{align}
where $t(\lambda)=t\big(n(\lambda),\epsilon=1/8,\delta=2^{-\lambda}, |\mathbb{K}_\lambda|\big)$
and the last inequality follows from our assumption on $\eta$. By pulling out the randomness in $\mathsf{SampPuzz}_C$\,, one obtains an efficient classical deterministic algorithm which satisfies the claim of the theorem. The fact that the function family $\{F_{\lambda}\}$ is a quantum-secure OWF then follows immediately from Theorem~\ref{thm:owf-from-classical-evowp}.
\end{proof}

\section{Hamiltonian phase state assumptions imply one-way functions}\label{s:HPS-imply-one-way}

In this section, we prove the following statement:
\begin{theorem}[Search HPS assumption implies one-way functions]\label{thm:HPS-implies-owf} If the Search HPS assumption (Assumption~\ref{ass:search-HPS}) is true, then one can construct a quantum-secure one-way function.
\end{theorem}

As discussed earlier, the Decision HPS assumption implies the Search HPS assumption (constructively), and therefore a simple corollary of the above theorem is that one can also construct a quantum-secure OWF if the Decision HPS assumption is true.

Given the tools, of the previous sections, we note that the following logic would suffice for proving Theorem~\ref{thm:HPS-implies-owf}.

\begin{enumerate}
\item Under the Search HPS assumption, we know that we can construct (via Construction~\ref{con:OWSG-from-HPS}) a OWSG  with output states in $\mathrm{HPS}_{q,m,n}$ for some $q(n) = 2^{O(\mathrm{poly}(n))}$ and $m(n)=O(\mathrm{poly}(n))$. Note that this OWSG has a classical key generation algorithm $\mathsf{KeyGen}$.
\item Given Theorem~\ref{thm:owf-via-sim-certOWSG}, we know that from any $1/3$-simulable certifiable OWSG (with classical $\mathsf{KeyGen}$), we can construct a quantum-secure OWF. 
\item Therefore, to prove Theorem~\ref{thm:HPS-implies-owf}, it is sufficient to show that there exists a $1/3$-simulable copy efficient "measure first, ask later" state certification protocol for $\mathrm{HPS}_{q,m,n}$.
\end{enumerate}
However, as discussed in Section~\ref{ss:owf-directly-via-EV}, the OWF obtained via the above proof would be compiled from a distributional OWF. If we can prove that there exists an $\eta$-simulable \textit{computationally} efficient "measure first, ask later" state certification protocol for $\mathrm{HPS}_{q,m,n}$, where $\eta$ is negligible in the sense of Theorem~\ref{thm:owf-via-sim-effcertOWSG}, then one can \textit{directly} construct a OWF, without going via a distributional OWF. As such, instead of proving that there exists a $1/3$-simulable copy efficient state certification protocol for $\mathrm{HPS}_{q,m,n}$, we prove the stronger statement that there exists an $\eta$-simulable \textit{computationally} efficient state certification protocol, for $\eta$ sufficient for Theorem~\ref{thm:owf-via-sim-effcertOWSG}. More specifically:
\begin{enumerate}
\item In Section~\ref{ss:phase-state certification}, we introduce the "measure first, ask later" state certification protocol from Ref.~\cite{HuangPreskillSoleimanifar2025} and prove that it is computationally efficient for any set of phase states with efficiently computable phase functions.
\item In Section~\ref{ss:classic-sim-of-phase-state certification}, we prove that for any set of phase states with efficiently computable phase functions, the same state certification protocol is also $\eta$-simulable for a function $\eta$ satisfying $\eta(|k|,n,t)=\mathsf{negl}(\lambda)$ whenever $t=O(\mathrm{poly}(n))$.
\item In Section~\ref{ss:HPS-theorem-proof}, we then put the pieces together to prove Theorem~\ref{thm:HPS-implies-owf}, as per the high-level sketch above.
\end{enumerate}

\subsection{Computationally efficient certification of phase states}\label{ss:phase-state certification}

In this section, we present the recent "measure first, ask later" state certification protocol from Ref.~\cite{HuangPreskillSoleimanifar2025}, and show that it is indeed computationally efficient for any set of phase states with efficiently computable phase functions (as a simple corollary of results already proven in Ref.~\cite{HuangPreskillSoleimanifar2025}). To present this certification protocol in the language of Definition~\ref{def:mfal-for-a-single-state}, we start by defining an intermediate measurement protocol $\hat{\mathcal{M}}$, which acts on a single copy of an unknown state $|\psi\rangle$:

\begin{algorithm}[H]
\caption{: $\hat{\mathcal{M}}$ -- Single-copy measurement protocol implicit in Protocol 2 of Ref.~\cite{HuangPreskillSoleimanifar2025}.}\label{alg:single-copy}
\begin{algorithmic}[1]
\Require $n$-qubit state $|\psi\rangle$
\State Choose a single index $i\in [n]\coloneqq \{1,\ldots,n\}$.
\State Measure all qubits \textit{except} qubit $i$ in the $Z$ basis, and denote the outcome with $z \in \{0,1\}^{n-1}$.
\State Choose $M\in \{X,Y,Z\}$ randomly, and measure qubit $i$ in the $M$-basis. Denote the outcome as $b\in\{0,1\}$.
\Ensure $(i,z,M,b)$.
\end{algorithmic}
\end{algorithm}

To set conventions, we label the $+1$ Pauli eigenstate by $b=0$ and the $-1$ eigenstate with $b=1$. We stress that the above measurement protocol $\hat{\mathcal{M}}$ is \textit{target} independent -- i.e.,  it has no dependence on the target state $|\phi\rangle$ -- and is therefore "measure first, ask later" in the sense of Section~\ref{ss:state certification-prelim}. Given this single-copy intermediate measurement protocol $\hat{\mathcal{M}}$, we then construct the measurement protocol $\mathcal{M}$ by applying algorithm $\hat{\mathcal{M}}$ independently on $t$ copies of the unknown state $|\psi\rangle$ as follows:

\begin{algorithm}[H]
\caption{: $\mathcal{M}$ -- Full measurement protocol implicit in Protocol 2 in Ref.~\cite{HuangPreskillSoleimanifar2025}}\label{alg:HKP-protocol-M}
\begin{algorithmic}[1]
\Require $|\psi\rangle^{\otimes t}$
\State Create an empty list $s = []$.
\For{$j \in [t]\coloneqq \{1,\ldots,t\}$,}
    \State $s^{(j)} \gets\hat{\mathcal{M}}(|\psi\rangle)$
    \State Append $s^{(j)}$ to $s$
\EndFor
\Ensure $s = [s^{(1)}, \ldots, s^{(t)}]$.
\end{algorithmic}
\end{algorithm}

Now, for any fixed target state $|\phi\rangle$, we can define the classical post-processing function $F_{\ket{\phi}}$ required in Definition~\ref{def:mfal-for-a-single-state}.  At a high level, $F_{|\phi\rangle}$ compares two objects extracted from each sample $s^{(j)}$: a single-qubit classical shadow $\sigma^{(j)}$, which is a linear (and target-independent) estimator of the true post-measurement state on qubit $i$, and a candidate single-qubit state $|\Psi^{(j)}\rangle$, reconstructed using the target state $\ket\phi$ via the amplitude oracle. The overlap between these two, averaged over $t$ rounds, is then used to determine whether to accept or reject.

To make this precise, recall that each $s^{(j)}$ is in fact a tuple   $s^{(j)}=(i,z,M,b)\in[n]\times\{0,1\}^{n-1}\times\{X,Y,Z\}\times\{0,1\}$. Furthermore, let's use the notation $z(j)$ for $j\in[n-1]$ to denote the $j$'th element of $z\in \{0,1\}^{n-1}$, and for any $i\in [n]$ and bit $b\in \{0,1\}$, we define $z\|b_i \in \{0,1\}^n$ as the bitstring defined from $z$ via the insertion of bit $b$ at index $i$ -- i.e., via 
\begin{equation}\label{eq:z-insertion}
z\|b_i(j) = \begin{cases}z(j) &\text{ if } j< i ,\\
b &\text{ if } j= i,\\
z(j-1) &\text{ if } j > i.
\end{cases}
\end{equation}
With this in hand, we then define a single-qubit classical shadow $\sigma^{(j)}$ from $s^{(j)}$ as
\begin{equation}
s^{(j)}=(i,z,M,b)\to\sigma^{(j)}=\begin{cases}
    3\,H|b\rangle\langle b|H-{\mathbb I}_2,&M=X,\\[2mm]
    3\,S H|b\rangle\langle b|HS^\dagger-{\mathbb I}_2,&M=Y,\\[2mm]
    3\,|b\rangle\langle b|-{\mathbb I}_2,&M=Z,
\end{cases}
\end{equation}
and also the single qubit state $\ket{\Psi^{(j)}}$ from $s^{(j)}$ as:
\begin{equation}
s^{(j)}=(i,z,M,b)\to|\Psi^{(j)}\rangle=\frac{\alpha^{(j)}_0|0\rangle+\alpha^{(j)}_1|1\rangle}{\sqrt{|\alpha^{(j)}_0|^2+|\alpha^{(j)}_1|^2}}
\end{equation}
where $\alpha_b^{(j)}=\langle z\|b_i|\phi\rangle$ can be obtained from the amplitude oracle ${\mathrm O}(\ket{\phi})$.  Finally, we then define the overlap
\begin{equation}
\omega^{(j)} = \mathrm{Tr}\left(\sigma^{(j)} |\Psi^{(j)}\rangle\langle \Psi^{(j)}| \right).
\end{equation}
With this, we can now define the function $F_{|\phi\rangle}$ formally in Algorithm~\ref{alg:measurement-protocol-HPS}:

\begin{algorithm}[H]
\caption{: $F_{|\phi\rangle}$ implicit in Protocol 2 (with $r=1$) in Ref.~\cite{HuangPreskillSoleimanifar2025}}\label{alg:measurement-protocol-HPS}
\begin{algorithmic}[1]
\Require Amplitude oracle $\mathsf{O}(|\phi\rangle)$, relaxation time $\tau$, $\mathcal{M}$'s output $s = [s^{(1)}, \ldots, s^{(t)}]$, and tolerance $\epsilon>0$.
\State Create an empty list $\bf{\omega} = []$.
\For{$j \in [t]$,}
    \State Calculate the shadow $\sigma^{(j)}$ from $s^{(j)}$.
    \State Using $\mathsf{O}(|\phi\rangle)$ and $s^{(j)}$, calculate the state $|\Psi^{(j)}\rangle$. 
    \State Calculate the overlap $\omega^{(j)}$.
    \State Append $\omega^{(j)}$ to $\omega$.
\EndFor
\State Calculate the mean
\begin{equation*}
\hat{\omega}[|\phi\rangle] \coloneqq \frac{1}{t}\sum_{j=1}^t \omega^{(j)}
\end{equation*}
\Ensure If $\hat{\omega}[|\phi\rangle] \geq 1-(3\epsilon)/(4\tau)$ output $\mathsf{Accept}$, else output $\mathsf{Reject}$.
\end{algorithmic}
\end{algorithm}

\begin{remark}[Efficiency of $F_{|\phi\rangle}$ for phase states]\label{rem:efficiency-for-phase-states}
We emphasize that Algorithm~\ref{alg:measurement-protocol-HPS} as stated requires amplitude-oracle access $\mathsf{O}(|\phi\rangle)$ in Line 4. However, for any phase state $|\phi_f\rangle$ with efficiently computable phase function $f$, the oracle $\mathsf{O}(|{\phi_f}\rangle)$ can itself be simulated in polynomial time, by directly evaluating $\langle x|\phi_f\rangle = 2^{-n/2}\exp(if(x))$. 
\end{remark}

Now, given the measurement protocol $(\mathcal{M},F_{|\phi\rangle})$ for a \textit{single} target state $|\phi\rangle$, we can easily construct the measurement protocol $(\mathcal{M},\mathcal{F})$ for a set of states $\{|\phi_k\rangle\,|\,k\in \mathbb{K}\}$ (i.e.,  satisfying Definition~\ref{def:measure-first-ask-later}) by defining 
\begin{equation}\label{eq:post-processing-def}
\mathcal{F}(k,\cdot,\cdot,\cdot) = \begin{cases}
F_{|\phi_k\rangle}(\cdot,\cdot,\cdot) &\forall\, k\in\mathbb{K}, \\
\mathsf{Reject} &\text{otherwise}.
\end{cases}
\end{equation}
For the case of phase states, we obtain the following as a straightforward corollary of Theorem 6 in Ref.~\cite{HuangPreskillSoleimanifar2025}:

\begin{theorem}[Certifiability of phase states -- Corollary of Theorem 6 in Ref.~\cite{HuangPreskillSoleimanifar2025}]\label{thm:HPSforaset}
For any set of phase states $\mathrm{PS}(\{f_k\}_{k\in\mathbb{K}}) = \{|\phi_k\rangle\,|\, k\in \mathbb{K}\}$, define $\mathcal{F}$ as per Eq.~\eqref{eq:post-processing-def}. For any state $|\phi_k\rangle$, denote by $\mathsf{O}(|\phi_k\rangle)$ the oracle which gives access to amplitudes of $|\phi_k\rangle$. Then, $(\mathcal{M},\mathcal{F})$ with $\tau= n$ is a ``measure first, ask later'' state certification protocol for $\{|\phi_k\rangle\}$, which is computationally efficient from $\mathsf{O}$-access, with copy complexity
\begin{equation}\label{eq:samp-complexity-certification}
t(n,\epsilon,\delta,|\mathbb{K}|)=O\left(\frac{n^2}{\epsilon^2}\log\left(\frac{|\mathbb{K}|}{\delta}\right)\right).
\end{equation}
Moreover, if the phase functions $\{f_{|\phi_k\rangle}\}_{k\in\mathbb{K}}$ are efficiently computable and membership in $\mathbb{K}$ can be efficiently decided, then $(\mathcal{M},\mathcal{F})$ with $\tau=n$ is in fact a computationally efficient ``measure first, ask later'' state certification protocol for $\{|\phi_k\rangle\}$. 
\end{theorem}

For completeness, we provide a proof of Theorem~\ref{thm:HPSforaset} in Appendix~\ref{app:proof-certification}.

\subsection{Classical simulability of phase state certification}\label{ss:classic-sim-of-phase-state certification}

In the previous section, we established that the state certification protocol from Ref.~\cite{HuangPreskillSoleimanifar2025} is \textit{computationally efficient} for any set of phase states with efficiently computable phase functions. In this section, we show that for any set of phase states with efficiently computable phase functions, this state certification protocol is also $\eta$-simulable, for a function $\eta(|k|,n,t)$ which is negligible with respect to $n$ whenever $t$ is at most a polynomial function of $n$. 

\begin{theorem}[$\mathsf{negl}$-simulability of Ref.~\cite{HuangPreskillSoleimanifar2025}'s state certification protocol for phase states]\label{thm:there-is-classical-sim}  Let $\mathrm{PS}(\{f_k\}_{k\in\mathbb{K}}) = \{|\phi_k\rangle\,|\, k\in \mathbb{K}\}$ be a set of phase states with efficiently computable phase functions $\{f_k\}_{k\in\mathbb{K}}$, and let $(\mathcal{M},\mathcal{F})$ be the ``measure first, ask later" state certification protocol from Theorem~\ref{thm:HPSforaset}. Then, $(\mathcal{M},\mathcal{F})$ is $\eta$-simulable for some function $\eta$ satisfying $\eta(|k|,n,t) = \mathsf{negl}(n)$ whenever $t = \mathrm{poly}(n)$. More specifically, 
there exists a classical PPT algorithm $\mathcal{M}_C$ satisfying
\begin{equation}
d_\mathrm{TV}(\mathcal{M}_C(k,1^n,1^t),\mathcal{M}(|\phi_k\rangle^{\otimes t})) \leq \eta(|k|,n,t)
\end{equation}
for all $k\in\mathbb{K}$ and $n,t\in\mathbb{N}$, where $\eta$ is such that whenever $t= O(\mathrm{poly}(n))$, then $\eta(|k|,n,t) = \mathsf{negl}(n)$.
\end{theorem}

\begin{proof}
Recall from Algorithm~\ref{alg:HKP-protocol-M} that $\mathcal{M}$ works by repeatedly calling the \textit{single-copy} measurement protocol $\hat{\mathcal{M}}$ defined in Algorithm~\ref{alg:single-copy}. As such, we start by providing a classical randomized algorithm $\hat{\mathcal{M}}_C(k,n,l)\rightarrow s$ for approximately sampling within error $1/2^l$ from the distribution defined by $\hat{\mathcal{M}}(|\phi_k\rangle)$ for each $n$-qubit state $|\phi_k\rangle$ defined by $k\in\mathbb K$. Given this, we then define $\mathcal{M}_C$ by simply replacing the calls to $\hat{\mathcal{M}}$ (on $|\phi_k\rangle$) in Algorithm~\ref{alg:HKP-protocol-M} with calls to $\hat{\mathcal{M}}_C$ (on input $k$). 

To this end, we require some notation. In particular, given any bitstring $z\in \{0,1\}^{n-1}$, index $i\in [n]$ and bit $b\in\{0,1\}$, recall from Eq.~\eqref{eq:z-insertion} the definition of $z\|b_i \in \{0,1\}^n$.
Moreover, given any $k\in \mathbb{K}$, define
\begin{equation}
g_k(z,i) = f_k(z\|1_i) - f_k(z\|0_i),
\end{equation}
where $f_k$ is the phase function from Eq.~\eqref{eq:def-fk}. Using this, for any $k\in \mathbb{K}$, we define the distribution $\tilde{\mathcal{M}}_{\mathrm{C}}(k,n)$ as the distribution sampled from via the following algorithm, where we set the convention that $\mathrm{Bern}(p)$ is the bit distribution satisfying $\mathrm{Pr}[b=1]=p$:

\begin{algorithm}[H]
\caption{Sampler for $\tilde{\mathcal{M}}_{\mathrm{C}}(k,n)$}\label{alg:classical-simulation}
\begin{algorithmic}[1]
\Require $k\in \mathbb{K}$, $n\in\mathbb{N}$
\State Choose a single index $i\in [n]$ uniformly at random.
\State Sample $z \in \{0,1\}^{n-1}$ from the uniform distribution.
\State Choose $M\in \{X,Y,Z\}$ randomly.
\If{$M=Z$,}
\State $b\gets\mathrm{Bern}(1/2)$
\ElsIf{$M=X$,}
\State $b\gets
\mathrm{Bern}\left[1-\cos^2(g_k(z,i)/2)\right]$
\ElsIf{$M=Y$,}
\State $b\gets
\mathrm{Bern}\big[1-\cos^2\big(\frac{g_k(z,i)-\pi/2}{2}\big)\big]$
\EndIf
\Ensure $(i,z,M,b)$.
\end{algorithmic}
\end{algorithm}

We now have the following lemma:

\begin{lemma}[Exact classical simulation of the single-copy measurement distribution]\label{lem:exact-but-inefficient} For all $k\in \mathbb{K}$, one has that
\begin{equation}
\tilde{\mathcal{M}}_{\mathrm{C}}(k,n) = \hat{\mathcal{M}}(|\phi_k\rangle),
\end{equation}
where both $\tilde{\mathcal{M}}_{\mathrm{C}}(k,n)$ and $\hat{\mathcal{M}}(|\phi_k\rangle)$ are understood as distributions.
\end{lemma}

\begin{proof}
Given that $|\phi_k\rangle$ is a phase state, we have that measuring any register of $|\phi_k\rangle$ in the computational basis yields a uniformly random bit. As such, the bitstring $z\in \{0,1\}^{n-1}$ output by $\hat{\mathcal{M}}(|\phi_k\rangle)$ is uniformly random, as is the bitstring $z\in\{0,1\}^{n-1}$ output by $\hat{\mathcal{M}}_{\mathrm{C}}(k,n)$. Now, note that the post-measurement state after obtaining $i$ in Step 1, and $z\in \{0,1\}^{n-1}$ in Step 2 of Algorithm $\hat{\mathcal{M}}(|\phi_k\rangle)$ is 
\begin{equation}
|\phi_k(z)\rangle = \frac{1}{\sqrt{2}}\left(|0\rangle + e^{ig_k(z,i)}|1\rangle\right)
\end{equation}
where as before $g_k(z,i) = f_k(z\|1_i) - f_k(z\|0_i)$ and $f_k$ is the phase function from Eq.~\eqref{eq:def-fk}. Therefore, we have the following for the bit $b$ output by Step 3 of $\hat{\mathcal{M}}(|\phi_k\rangle)$, conditioned on obtaining $i$ and $z$ in Steps 1 and 2:
\begin{enumerate}
\item If $M=Z$, then $\mathrm{Pr}[b=0|z,i] = \mathrm{Pr}[b=1|z,i] = 1/2$.
\item If $M=X$, then
\begin{equation}
\mathrm{Pr}[b|z,i] = \begin{cases}
\cos^2(g_k(z,i)/2) &\text{ if }  b=0, \\
1 -\cos^2(g_k(z,i)/2) &\text{ if }  b=1.
\end{cases}
\end{equation}
\item If $M= Y$, then
\begin{equation}
\mathrm{Pr}[b|z,i] = \begin{cases}
\cos^2\left(\frac{g_k(z,i)-\pi/2}{2}\right) &\text{ if } b=0, \\
1 -\cos^2\left(\frac{g_k(z,i)-\pi/2}{2}\right) &\text{ if }  b=1.
\end{cases}
\end{equation}
\end{enumerate}
But the above distributions are precisely the distributions sampled from in Step 4 of $\tilde{\mathcal{M}}_{\mathrm{C}}(k,n)$, and both $i$ and $M$ are drawn identically in both algorithms. 
\end{proof}

Unfortunately, $\tilde{\mathcal{M}}_{\mathrm{C}}(k,n)$ cannot be sampled from \textit{exactly} in polynomial time due to the Bernoulli sampling steps in Lines 7 and 9 of Algorithm~\ref{alg:classical-simulation}. However, using the efficient computability of the phase functions, together with standard sampling techniques, one can sample from a negligibly close approximation to the desired Bernoulli distributions in polynomial time. More specifically, we have the following lemma:

\begin{lemma}[Efficient approximate Bernoulli sampling]\label{lem:eff-approx-Bernoulli} 
Assume the phase functions $\{f_k\}_{k\in\mathbb{K}}$ are efficiently computable. Then, there exist randomized algorithms $\mathcal{B}_1$ and $\mathcal{B}_2$, which on input $z\in\{0,1\}^n$, $i\in [n]$, $k\in\mathbb{K}$ and $l$, run in time $O(\mathrm{poly}(n,|k|,l))$ and satisfy
\begin{align}
&d_\mathrm{TV}\left(\mathcal{B}_1(z,i,k,l),
\mathrm{Bern}\left[1-\cos^2(g_k(z,i)/2)\right]\right) \leq 2^{-l}\label{eq:B1-close},\\
&d_\mathrm{TV}\left(\mathcal{B}_2(z,i,k,l), 
\mathrm{Bern}\Big[1-\cos^2\big(\frac{g_k(z,i)-\pi/2}{2}\big)\Big]\right) \leq 2^{-l}.\label{eq:B2-close}
\end{align}
\end{lemma}

\begin{proof} On input $z,i,k,l$, algorithm $\mathcal{B}_1$ does the following:
\begin{enumerate}
\item $\hat{f}_k(z\|1_i,l)\gets \mathcal{A}(k, z\|1_i,l)$ and $\hat{f}_k(z\|0_i,l)\gets \mathcal{A}(k, z\|0_i,l)$ where $\mathcal{A}$ is the algorithm implied via efficient computability of $\{f_k\}_{k\in\mathbb{K}}$.
\item Compute $\hat{g}_k(z,i,l):=\hat{f}_k(z\|1_i,l)-\hat{f}_k(z\|0_i,l)$.
\item Compute an $l$-bit dyadic rational $\hat c(z,i,k,l)$ satisfying
\begin{equation}\label{eq:hat-c-def}
\left|\hat c(z,i,k,l)-\cos^2\left(\frac{\hat g_k(z,i,l)}{2}\right)\right|<\frac1{2^{l+1}}.
\end{equation}
\item Output: sample $b\gets {\rm Bern}[1-\hat c(z,i,k,l)]$.
\end{enumerate}

Clearly, the output of the randomized algorithm $\mathcal B_1$ is that of a Bernoulli distribution ${\mathcal B}_1(z,i,k,l)={\rm Bern}(1-\hat c(z,i,k,l))$ with parameter $1-\hat c(z,i,k,l)$. To compute such a parameter, one (1st) reads $k$, (2nd) calls the corresponding ${\mathcal A}_k$ twice, and (3rd) uses efficiently computable operations (e.g., $\cos$ acting on a length-$l$ binary fraction) to obtain the length-$l$ binary fraction representation of $\hat c(z,i,k,l)$. Finally, having the length-$l$ binary fraction representation of $\hat c$, one can sample from ${\rm Bern}(1-\hat c)$ by flipping $l$ fair coins. Therefore, $\mathcal{B}_1$ has runtime $\mathrm{poly}(n,|k|,l)$.

Let us now show Eq.~\eqref{eq:B1-close}. We have
\begin{eqnarray}
\begin{aligned}
d_\mathrm{TV}\left(\mathcal{B}_1(z,i,k,l),\mathrm{Bern}\left[1-\cos^2(g_k(z,i)/2)\right]\right)&=d_\mathrm{TV}\left({\rm Bern}[1-\hat c(z,i,k,l)],\mathrm{Bern}\left[1-\cos^2(g_k(z,i)/2)\right]\right)\\
&=\left|\hat c(z,i,k,l)-\cos^2\left(\frac{g_k(z,i)}{2}\right)\right|\,,
\end{aligned}
\end{eqnarray}
and by the triangular inequality,
\begin{equation}
\left|\hat c(z,i,k,l)-\cos^2\left(\frac{g_k(z,i)}{2}\right)\right|\le\left|\hat c(z,i,k,l)-\cos^2\left(\frac{\hat g_k(z,i,l)}{2}\right)\right|+\left|\cos^2\left(\frac{\hat g_k(z,i,l)}{2}\right)-\cos^2\left(\frac{g_k(z,i)}{2}\right)\right|\le\frac1{2^l}\,,
\end{equation}
where we have used Eq.~\eqref{eq:hat-c-def} in the first term, the fact that $\cos^2$ has Lipschitz constant $1$, and that $|\hat g_k(z,i,l)-g_k(z,i)|\le 1/2^l$ by Eq.~\eqref{eq:hat-f-def}. That is, the second term is bounded as
\begin{equation}
\left|\cos^2\left(\frac{\hat g_k(z,i,l)}{2}\right)-\cos^2\left(\frac{g_k(z,i)}{2}\right)\right|\le\frac12\left|\hat g_k(z,i,l)-g_k(z,i)\right|\le\frac1{2^{l+1}}\,.
\end{equation}

The proof for $\mathcal B_2$ is essentially the same but replacing $\hat c$ with $\hat d(z,i,k,l)$ such that
\begin{equation}\label{eq:hat-d-def}
\left|\hat d(z,i,k,l)-\cos^2\left(\frac{\hat g_k(z,i,l)-\pi/2}{2}\right)\right|<\frac1{2^{l+1}}\,.
\end{equation}
Therefore, the runtime of $\mathcal B_2$ is again $\mathrm{poly}(n,|k|,l)$, and its output is that of a Bernoulli distribution ${\mathcal B}_2(z,i,k,l)={\rm Bern}[1-\hat d(z,i,k,l)]$ with parameter $1-\hat d(z,i,k,l)$. Using the triangular inequality and same arguments above,
\begin{equation}
d_\mathrm{TV}\left(\mathcal{B}_2(z,i,k,l),
\mathrm{Bern}\Big[1-\cos^2\big(\frac{g_k(z,i)-\pi/2}{2}\big)\Big]\right)=\left|\hat d(z,i,k,l)-\cos^2\left(\frac{g_k(z,i)-\pi/2}{2}\right)\right|\le\frac1{2^l}\,,
\end{equation}
and the lemma is proven.
\end{proof}
With the above in hand, we can now straightforwardly define the algorithm $\hat{\mathcal{M}}_C(k,n,l)$ (see Algorithm~\ref{alg:classical-simulation-approx}) by replacing the exact Bernoulli sampling in $\tilde{\mathcal{M}}_C(k,n)$ with the approximate Bernoulli sampling via $\mathcal{B}_1$ and $\mathcal{B}_2$. 

\begin{algorithm}[H]
\caption{Sampler for $\hat{\mathcal{M}}_{\mathrm{C}}(k,n,l)$}\label{alg:classical-simulation-approx}
\begin{algorithmic}[1]
\Require $k\in \mathbb{K}$, $n,l\in\mathbb{N}$
\State Choose a single index $i\in [n]$ uniformly at random.
\State Sample $z \in \{0,1\}^{n-1}$ from the uniform distribution.
\State Choose $M\in \{X,Y,Z\}$ randomly.
\If{$M=Z$,}
\State $b\gets \mathrm{Bern}(1/2)$ 
\ElsIf{$M=X$,}
\State $b\gets \mathcal{B}_1(z,i,k,l)$
\ElsIf{$M=Y$,}
\State $b\gets \mathcal{B}_2(z,i,k,l)$
\EndIf
\Ensure $(i,z,M,b)$.
\end{algorithmic}
\end{algorithm}
In particular, we have the following lemma:
\begin{lemma}[Efficient single-copy negligible error simulation]\label{lem:efficient-approximation}For all $k\in \mathbb{K}$ and $l\in\mathbb{N}$, one has
\begin{equation}
d_\mathrm{TV}(\hat{\mathcal{M}}_{\mathrm{C}}(k,n,l),\tilde{\mathcal{M}}_{\mathrm{C}}(k,n)) \leq 2^{-l}.
\end{equation}
Moreover, $\hat{\mathcal{M}}_C(k,n,l)$ runs in time $\mathrm{poly}(n,l)$.
\end{lemma}
\begin{proof}
First, it is clear that the runtime of $\hat{\mathcal{M}}_C(k,n,l)$ is that of $\mathcal B_1$ and $\mathcal B_2$, that is, $\mathrm{poly}(n,|k|,l)$. Moreover,
\begin{equation}
\begin{aligned}
d_\mathrm{TV}(\hat{\mathcal{M}}_{\mathrm{C}}(k,n,l),\tilde{\mathcal{M}}_{\mathrm{C}}(k,n))=\frac{1}{3}0&+\frac{1}{3}\frac{1}{n2^{n-1}}\sum_{i=1}^n\sum_{z\in\{0,1\}^{n-1}}d_{\rm TV}\left(\mathcal B_1(z,i,k,l),\mathrm{Bern}\left[1-\cos^2(g_k(z,i)/2)\right]\right)\\
&+\frac{1}{3}\frac{1}{n2^{n-1}}\sum_{i=1}^n\sum_{z\in\{0,1\}^{n-1}}d_{\rm TV}\left(\mathcal B_2(z,i,k,l),\mathrm{Bern}\left[1-\cos^2\big(g_k(z,i)/2-\pi/4\big)\right]\right)\\
&\le\frac{2}{3}2^{-l}\le2^{-l}
\end{aligned}
\end{equation}
where we have used the fact that both algorithms are the same when $M=Z$ for all $i,z$, and Eqs.~\eqref{eq:B1-close} and~\eqref{eq:B2-close} for the cases $M=X$ and $M=Y$, respectively.
\end{proof}

Finally, we can put everything together and define the claimed algorithm $\mathcal{M}_C$ as follows:
\begin{algorithm}[H]
\caption{: $\mathcal{M}_C(k,1^n,1^t)$ -- PPT $\mathsf{negl}$-approximate sampler for $\mathcal{M}(|\phi_k\rangle^{\otimes t})$}\label{alg:PPT-approx-M}
\begin{algorithmic}[1]
\Require $k,1^n,1^t$
\State Create an empty list $s = []$.
\For{$j \in [t]\coloneqq \{1,\ldots,t\}$,}
    \State $s^{(j)} \gets\hat{\mathcal{M}}_C(k,n,n)$
    \State Append $s^{(j)}$ to $s$
\EndFor
\Ensure $s = [s^{(1)}, \ldots, s^{(t)}]$.
\end{algorithmic}
\end{algorithm}

By virtue of Lemma~\ref{lem:exact-but-inefficient} and Lemma~\ref{lem:efficient-approximation}, we have that
\begin{equation}
d_\mathrm{TV}(\mathcal{M}_C(k,1^n,1^t),\mathcal{M}(|\phi_k\rangle^{\otimes t}))\leq t2^{-n}.
\end{equation}
Therefore, for all $t=\mathrm{poly}(n)$, we then immediately have
\begin{equation}
d_\mathrm{TV}(\mathcal{M}_C(k,1^n,1^t),\mathcal{M}(|\phi_k\rangle^{\otimes t})) \leq \mathsf{negl}(n)
\end{equation}
as claimed. Moreover, $\mathcal{M}_C(k,1^n,1^t)$ just calls $t$ times the algorithm $\hat{\mathcal{M}}_C(k,n,l)$ with $l=n$. Since the latter runs in time $\poly(n,|k|,l)$, it follows that $\mathcal{M}_C(k,1^n,1^t)$ runs in time $\poly(n,|k|,t)$, that is, polynomial in the size of its input. This concludes the proof of Theorem~\ref{thm:there-is-classical-sim}.
\end{proof}

\subsection{Proof of Theorem~\ref{thm:HPS-implies-owf}}\label{ss:HPS-theorem-proof}

Given Theorem~\ref{thm:HPSforaset} and Theorem~\ref{thm:there-is-classical-sim}, we can now straightforwardly prove Theorem~\ref{thm:HPS-implies-owf}.

\begin{proof}[Proof (Theorem~\ref{thm:HPS-implies-owf})] Set $n=\lambda$. Let the functions $m(\lambda) =\mathrm{\poly(\lambda)}$ and $q(\lambda) =2^{O(\mathrm{poly}(\lambda))}$ be as per the Search HPS assumption. Let $(\mathsf{KeyGen},\mathsf{StateGen},\mathsf{Ver})$ be the OWSG obtained from  Construction~\ref{con:OWSG-from-HPS} under the Search HPS assumption. Note that $\mathbb{K}_\lambda \coloneqq \mathrm{supp}(\mathsf{KeyGen}(\lambda)) =  \mathbb{K}_{q,m,\lambda}$. Let $(\mathcal{M},\mathcal{F})$ be the state certification protocol from Theorem~\ref{thm:HPSforaset} for the set of phase states $\mathrm{HPS}_{q,m,\lambda} = \{|\phi_k\,|\,k\in\mathbb{K}_\lambda\}$, with copy complexity 
\begin{equation}\label{eq:samp-complexity-certification-final-proof}
t(\lambda,\epsilon,\delta,|\mathbb{K}|)=O\left(\frac{\lambda^2}{\epsilon^2}\log\left(\frac{|\mathbb{K}|}{\delta}\right)\right).
\end{equation}
We prove Theorem~\ref{thm:HPS-implies-owf} by:
\begin{enumerate}
\item Showing that $(\mathcal{M},\mathcal{F})$ is \textit{efficiently computable}.
\item Showing that $(\mathcal{M},\mathcal{F})$ is $\eta$-simulable for a function $\eta$ satisfying
\begin{equation}
\sup_{k\in\mathbb{K}_\lambda}\left[\eta\Big(|k|, n(\lambda), t\big(n(\lambda),\epsilon=1/8,\delta=2^{-\lambda}, |\mathbb{K}_\lambda|\big)\Big)\right] = \mathsf{negl}(\lambda).
\end{equation}
\end{enumerate}
Given both the above, it then follows from Theorem~\ref{thm:owf-via-sim-effcertOWSG} that one can directly construct a quantum-secure OWF from the $\eta$-simulable efficiently certifiable OWSG $((\mathsf{KeyGen},\mathsf{StateGen},\mathsf{Ver}),(\mathcal{M},\mathcal{F}))$.

\textbf{(1) Efficient computability:} This follows immediately from Theorem~\ref{thm:HPSforaset}, together with the observation in Remark~\ref{remark:eff-computability-HPS} that the phase functions defining $\mathrm{HPS}_{q,m,\lambda}$ are efficiently computable.

\textbf{(2) $\eta$-simulability:} From Theorem~\ref{thm:there-is-classical-sim}, it follows that $(\mathcal{M},\mathcal{F})$ is $\eta$-simulable for $\eta$ satisfying
\begin{equation}
\eta(|k|,\lambda,t) = \mathsf{negl}(\lambda)
\end{equation}
whenever $t = O(\mathrm{poly}(\lambda))$. Now, using Eq.~\eqref{eq:samp-complexity-certification-final-proof}, the fact that $n=\lambda$ and the fact that $|\mathbb{K}_\lambda| = 2^{O(\mathrm{poly}(\lambda))}$, we have
\begin{equation}
t\big(n(\lambda),\epsilon=1/8,\delta=2^{-\lambda}, |\mathbb{K}_\lambda|\big)=O(\mathrm{poly}(\lambda)),
\end{equation}
and therefore,
\begin{equation}
\sup_{k\in\mathbb{K}_\lambda}\left[\eta\Big(|k|, n(\lambda), t\big(n(\lambda),\epsilon=1/8,\delta=2^{-\lambda}, |\mathbb{K}_\lambda|\big)\Big)\right] =\mathsf{negl(\lambda)}.
\end{equation}
\end{proof}

\begin{remark}[Alternative state certification methods]\label{rem:alternative}
We note that the protocol of~\cite{HuangPreskillSoleimanifar2025} is not the only computationally efficient and $\eta$-simulable state certification protocol available for (Hamiltonian) phase states. The recent work of~\cite{park2026samplehardwareefficientfidelityestimation} implicitly provides an alternative method that achieves \textit{constant} copy complexity with respect to $n$ ---an improvement over the $O(n^2/\epsilon^2)$ scaling of Theorem~\ref{thm:HPSforaset}--- at the cost of requiring more elaborate, though still target-independent, measurements rather than single-qubit measurements alone. Using this state certification protocol in place of the one from~\cite{HuangPreskillSoleimanifar2025} would therefore also allow us to construct a OWF from the search HPS assumption, with different properties to the one we have obtained implicitly in our proof of Theorem~\ref{thm:HPS-implies-owf}. We note that more recent state certification protocols which work for \textit{all} states using only single-qubit measurements~\cite{gupta2025singlequbitmeasurementssufficecertify,Li_2026} are \textit{not} applicable to our setting, as they require target-state-dependent adaptive measurements -- i.e., they are not "measure first, ask later" as required here.
\end{remark}

\begin{remark}[Explicit OWF]\label{rem:explicit-owf} In the proof of Theorem~\ref{thm:HPS-implies-owf}, the OWF $\{F_\lambda\}$, whose existence has been proven, has been left implicit. In Appendix~\ref{app:explicit-OWF}, we provide an explicit specification of this OWF. We note here however that the domain of $F_\lambda$ depends on the \textit{copy-complexity} $t$ of the state certification protocol $(\mathcal{M},\mathcal{F})$ appearing above. In particular, this fact provides a motivation for the use of alternative state certification protocols with \textit{constant} copy complexity, as discussed in Remark~\ref{rem:alternative} above.
\end{remark}

\section{Towards concrete instantiations of efficiently verifiable one-way puzzles}\label{s:towards-concrete-instantiations}

The results of the previous sections give a general recipe for constructing EV-OWP from PRSG/OWSG together with ``measure first, ask later'' state certification protocols. Using this recipe, we see that we can construct an EV-OWP which provides a genuine instantiation of Quantumania from any family of states $\{\ket{\phi_k}\}_{k\in\mathbb K}$ satisfying the following three properties: 

\begin{enumerate}
\item \textbf{PRSG/OWSG:} There exists a distribution over the family which yields either a standard PRSG or a OWSG (under a plausible assumption).
\item \textbf{Certifiability:} The family admits a computationally efficient ``measure first, ask later'' state certification protocol $(\mathcal{M}, \mathcal{F})$ in the sense of Definition~\ref{def:measure-first-ask-later}. 
\item \textbf{Independence from OWF:} The assumption under which the family is a PRSG or OWSG can be true independently of whether one-way functions exist. In particular:
\begin{enumerate}
\item The assumption in Point 1 \textit{cannot} be that one-way functions exist.
\item The measurement protocol $(\mathcal{M},\mathcal{F})$ should not be $1/3$-simulable -- i.e., the joint distribution over $(k,s)$ induced by sampling a key $k$ and running $\mathcal M$ on $\ket{\phi_k}^{\otimes t}$ should not be classically simulable to within constant total variation distance. If it is $1/3$-simulable, then by Theorem~\ref{thm:owf-via-sim-certOWSG}, the EV-OWP built from the PRSG/OWSG together with the state certification protocol implies a one-way function.
\end{enumerate}
\end{enumerate}
We note that simply satisfying Properties 1 and 2 will yield an EV-OWP (via Corollary~\ref{cor:ev-owp-from-eff-certifiable-owsg}), however Property 3 is necessary for this EV-OWP to be genuinely in Microcrypt! In particular, our work highlights that if one is to instantiate Quantumania via this route, then one requires a family of states satisfying a delicate balance: It should be sufficiently complex to yield an OWSG (or a PRSG) without using OWFs, yet simple enough to admit a computationally efficient ``measure first, ask later'' state certification protocol, while at the same time remaining sufficiently intricate that the certification protocol itself cannot be efficiently classically simulated. In summary, they are families of states that are {\it hard enough to learn, structured enough to certify, yet not so structured that the certification protocol becomes classically simulable}. As such, the natural question raised by this work is the following:

\begin{question} Does there exist a family of states satisfying all three properties above?
\end{question}

We stress that, prior to this work, Hamiltonian phase states were precisely such a candidate! However, we show here that they fail to satisfy Property 3. Given this, to the best of our knowledge, no known family of states currently satisfies all three properties simultaneously. Below, we discuss the shortcomings of plausible candidates, with a partial summary given in Table~\ref{tab:landscape}.

{\bf OWF phase states:} The original PRSG construction of Ji, Liu and Song~\cite{PRSGdefinition} is obtained by considering a family of phase states with phases computed via a quantum-secure one-way function. As such, this family of states immediately cannot satisfy Property 3. 

{\bf Hamiltonian phase states.} As discussed in Remark~\ref{remark:efficient-preparation-HPS}, under either the Search (or Decision) HPS assumption, Hamiltonian phase states satisfy Property 1. Additionally, Theorem~\ref{thm:HPSforaset} (leveraging the state certification protocol of Ref.~\cite{HuangPreskillSoleimanifar2025}) shows that they satisfy Property 2. However, we have shown in Theorem~\ref{thm:there-is-classical-sim} that they do not satisfy Property 3!

{\bf Low stabilizer rank states.} Superpositions of few stabilizer states -- i.e.,  states with low \textit{stabilizer rank} -- are potentially plausible candidates for satisfying Property 1. In particular, despite recent work showing the learnability of states with bounded stabilizer \textit{extent}~\cite{srini-tomography}, it remains an open question whether states of low stabilizer \textit{rank} can be learned efficiently from copies (see Question 2 in Ref.~\cite{anshu2023surveycomplexitylearningquantum}). Additionally, by exploiting the stabilizer structure of these states, one can show that classical shadows provides a computationally efficient ``measure first, ask later'' state certification protocol -- i.e., that these states satisfy Property 2. Unfortunately, however, it appears that, as was the case for Hamiltonian phase states, the same structure that facilitates efficient certification \textit{also} allows one to classically simulate the quantum part of the measurement process, and therefore, these states fail to satisfy Property 3. As such, states with low stabilizer rank provide a second example of a family of states that plausibly satisfies Properties 1 and 2, but not Property 3.

{\bf Random circuits.}  Output states of (log depth) random brickwork  circuits are plausibly hard to learn or clone~\cite{fefferman2025hardnesslearningquantumcircuits}, giving good evidence for Property 1 (without assuming OWF in the construction). However, no computationally efficient ``measure first, ask later'' certification protocol is known for this family, and consequently, at present, they do not satisfy Property 2, and Property 3 cannot be assessed. We note that the lack of a computationally efficient ``measure first, ask later'' state certification protocol for these states is precisely why proposed QCCC digital signature schemes based on this class of states lack efficient classical certification~\cite{niroula2026digitalsignaturesclassicalshadows}.

\begin{table}[t]
\centering
\begin{adjustbox}{width=1\textwidth}
\begin{tabular}{lccc}
\toprule
State family &
PRSG/OWSG &
Efficient
measure-first certification &
Independent of OWF\\
\midrule

OWF phase states
& Yes (assuming OWF)
& Yes
& No \\

Hamiltonian phase states
    & Assumed
    & Yes
    & No \\

Low-rank stabilizer states
    & Plausible$^*$
    & Yes
    & No \\

Random brickwork circuits
    & Plausible
    & Unknown
    & Unknown \\

\bottomrule
\end{tabular}
\end{adjustbox}
\caption{A summary of the extent to which a variety of existing families of states satisfy the properties necessary for instantiating an EV-OWP which is genuinely in Microcrypt. Existing candidate state families
populate different regions of this landscape, but no known family currently
satisfies all of the desired properties. (*) Open problem listed as Question 2 in~\cite{anshu2023surveycomplexitylearningquantum}. }
\label{tab:landscape}
\end{table}

\textbf{PRSGs from Forrelation:} Ref.~\cite{Kretschmer_2023} has constructed an oracle relative to which \textit{$t$-Forrelation} states can be used to construct a \textit{single-copy} PRSG and $\mathsf{P}=\mathsf{NP}$ -- i.e.,  in this oracle world the single-copy PRSG does not imply OWF. Unfortunately, single-copy PRSG are only able to instantiate Nanocrypt, and to the best of our knowledge are too weak to construct OWPs and therefore do not suffice for our purpose. However, Ref.~\cite{Kretschmer_2023} has also outlined a path for constructing multi-copy PRSG, assuming a strong conjecture on properties of $t$-Forrelation states. In light of this, understanding the certifiability of such states, assuming the required property of $t$-Forrelation states provides an interesting route to better understanding whether such states could also provide an instantiation of EV-OWP in this oracle world.

{\bf Tensor network states:} In one-dimension, quantum states with polynomial bond dimension are efficiently learnable~\cite{Lanyon_2017}, and therefore not suitable candidates for the construction of PRSG or OWSG. However, on arbitrary graphs this is no longer immediately the case, and such states may potentially be complex enough for the construction of PRSG or OWSG, while retaining sufficient structure for computationally efficient "measure first, ask later" state certification.

\begin{remark}[On candidate ensembles for OWPs] In this section we have focused on candidate ensembles for the instantiation of EV-OWPs (and therefore Quantumania). However, if one is interested in constructing OWPs (and instantiating Countcrypt) the same list of criteria apply, but with only \textit{copy} efficiency of the state certification protocol required (as opposed to \textit{computational} efficiency). However, because \textit{every} state ensemble admits a copy-efficient "measure first, ask later" state certification protocol (global Clifford shadows) the tools and techniques we develop here cannot offer more insight into the suitability of any particular state ensemble -- i.e. tailored state certification protocols cannot add anything over classical shadows.
\end{remark}

\newpage
\appendix

\section{Proof of Corollary~\ref{cor:classical-sampling-implies-qsecDOWF}}\label{app:owp-to-doOWF}

As mentioned in the main text, this proof proceeds identically to the proof of Theorem~\ref{thm:KTowf} (Claim D.1 in~\cite{khurana2025founding}), after establishing that from a successful quantum adversary for a distributional OWF (as per Definition~\ref{def:qsecDOWF}) you can get the same guarantee you would get from a successful classical adversary (as per Definition~\ref{def:dist-one-way-function}).

To be more precise, let's assume that $\{f_\lambda\}_\lambda$ is \textit{not} a quantum-secure distributional OWF. Under this assumption, one has that
\begin{equation}
\overline{F}_\mathcal{A}(\lambda) > 1-\frac{1}{\lambda^2}
\end{equation}
for infinitely many $\lambda$. If we use the notation $z\gets \rho^{\mathcal{A}}_{f_\lambda}(r)$ to denote $z$ obtained by measuring $\rho^{\mathcal{A}}_{f_\lambda}(r)$ in the computational basis, then it follows from the Fuchs-van de Graaf inequality and Jensen's inequality that
\begin{align}
d_\mathrm{TV}\left( (r,f_\lambda(r))_{r\gets{0,1}^{m(\lambda)}}, (z,f_\lambda(r))_{\substack{r\gets{0,1}^{m(\lambda)}\\z\gets \rho^{\mathcal{A}}_{f_\lambda}(r)}} \right) &=\mathbb{E}_{r\gets\{0,1\}^{m(\lambda)}}\left[d_\mathrm{TV}\left(r'\gets f^{-1}_{\lambda}(f_\lambda(r)),z\gets\rho^{\mathcal{A}}_{f_\lambda(r)} \right)\right] \\
&\leq \mathbb{E}_{r\gets\{0,1\}^{m(\lambda)}}\left[\frac{1}{2}\|\sigma_{f_\lambda(r)} - \rho^{\mathcal{A}}_{f_\lambda(r)} \|_1\right]\\
&\leq \mathbb{E}_{r\gets\{0,1\}^{m(\lambda)}}\left[\sqrt{1-F_\mathrm{sq}(\sigma_{f_\lambda(r)},\rho^\mathcal{A}_{f_\lambda(r)})}\right]\\
&\leq \sqrt{1-\overline{F}_\mathcal{A}(\lambda)} \\
&< \frac{1}{\lambda}.
\end{align}
Now, define $\tilde{\mathcal{A}}$ as the QPT algorithm which on input $(1^\lambda,y)$ does the following:
\begin{enumerate}
\item Obtain $\rho^{\mathcal{A}}_{y} \gets \mathcal{A}(1^\lambda,|y\rangle\langle y|)$.
\item Output $z\gets \rho^{\mathcal{A}}_{y}$.
\end{enumerate}
Then, it follows from the above that
\begin{equation}\label{eq:second-guarantee}
d_\mathrm{TV}\left( (r,f_\lambda(r))_{r\gets{0,1}^{m(\lambda)}}, (z,f_\lambda(r))_{\substack{r\gets{0,1}^{m(\lambda)}\\z\gets \tilde{\mathcal{A}}(1^\lambda,f_{\lambda}(r))}}\right) \leq \sqrt{1-\overline{F}_\mathcal{A}(\lambda)} < \frac{1}{\lambda}.
\end{equation}
The rest of the proof is now identical to the proof of Claim D.1 in~\cite{khurana2025founding}, but we reproduce it for convenience. In particular, we use $\tilde{\mathcal{A}}$ to construct an adversary $\mathcal{B}$ for the OWP $(\mathsf{Samp},\mathsf{Ver})$ from which $\{f_\lambda\}$ was constructed. To do this, recall that we have defined the existence of an efficient deterministic classical algorithm $\mathsf{Samp}_C$ satisfying
\begin{equation}
d_\mathrm{TV}\left(\mathsf{Samp}(1^\lambda),\mathsf{Samp}_{\mathrm{C}}(1^\lambda,r)_{r\gets\{0,1\}^{m(\lambda)}}\right) \leq 1/3,
\end{equation}
where $\mathsf{Samp}_{\mathrm{C}}(1^\lambda,r) = (k_\lambda(r),s_\lambda(r))$. Using this, define the QPT puzzle adversary $\mathcal{B}$ on input $(1^\lambda,s)$ as follows:
\begin{enumerate}
\item $z\gets \tilde{\mathcal{A}}(1^\lambda,s)$.
\item $(k_\lambda(z),s_\lambda(z))\gets \mathsf{Samp}_{\mathrm{C}}(1^\lambda,z)$
\item Output $k_\lambda(z)$.
\end{enumerate}
Now, let $(k,s)\gets\mathsf{Samp}(1^\lambda)$. It follows from Eq.~\eqref{eq:second-guarantee} that we have
\begin{equation}
d_\mathrm{TV}\left((k_\lambda(r),s_\lambda(r))_{r\gets\{0,1\}^{m(\lambda)}}, (\mathcal{B}(1^\lambda,s_\lambda(r)),s_\lambda(r))_{r\gets\{0,1\}^{m(\lambda)}} \right) <\frac{1}{\lambda}
\end{equation}
and from the $1/3$-correctness of $\mathsf{Samp}_C$ that
\begin{equation}
d_\mathrm{TV}\left((k,s),(\mathcal{B}(1^\lambda,s),s)\right) \leq \frac{2}{3} + \frac{1}{\lambda}.
\end{equation}
Then, the correctness of the OWP implies
\begin{equation}
\mathrm{Pr}_{(k,s)\rightarrow\mathsf{Samp}(1^\lambda)}\left[\mathsf{Ver}(\mathcal{B}(1^\lambda,s),s)=1\right] \geq \frac{1}{3} - \frac{1}{\lambda} - \mathsf{negl}(\lambda)
\end{equation}
for infinitely many $\lambda$, which contradicts the assumed security of the OWP.

\section{Physical implementation of Hamiltonian phase states}
\label{app:experimental-implementation}

Hamiltonian phase states are attractive as candidate cryptographic states not only because of their conjectured computational properties, but also because they can be prepared using simple commuting quantum circuits. In this appendix, we briefly discuss how the operations required to prepare such states can be realized in common experimental architectures. This physical accessibility makes the obstruction established in this work particularly relevant: it applies to a concrete and experimentally motivated family of quantum states rather than to a purely abstract cryptographic construction.

As discussed in Remark~\ref{remark:efficient-preparation-HPS}, Hamiltonian phase states can be prepared using \emph{instantaneous quantum polynomial-time} 
(IQP) circuits. The matrix ${\bf A}$ specifies the structure and support of the Ising terms in the generating Hamiltonian. More 
precisely, for $A=(A_1,\ldots,A_n)\in\mathbb{F}2^n$, the corresponding operation 
is a multi-qubit Ising phase rotation of the 
form
\begin{equation}
U(\theta,A)
:=
\exp\left(
\iu\theta
\bigotimes_{j=1}^{n} Z_j^{A_j}
\right),
\label{eq}
\end{equation}
where $\theta\in\mathbb{R}$ and the Hamming weight $|A|$ 
determines the number of qubits on which the operation acts non-trivially.

For $|A|=2$, Eq.~\eqref{eq} is a two-qubit $ZZ$ rotation. Such 
entangling phase operations are natural in Rydberg-atom platforms, where strong dipole--dipole interactions and the Rydberg-blockade mechanism enable programmable interactions between pairs of atoms \cite{PhysRevLett.85.2208,RevModPhys.82.2313,PhysRevA.110.032619,Jandura,LevinePRL19,Wassner2026holonomicquantum}. Suitable pulse sequences can also generate multi-qubit controlled-phase operations within a blockade region, including three-qubit gates \cite{Jandura}. The precise range and connectivity of the available interactions depend on the geometry 
of the atomic array, the blockade radius, and the employed control protocol.

More general instances of Eq.~\eqref{eq} can be synthesized from 
two-qubit entangling gates and single-qubit rotations. Let $w=|A|$, and choose one of the $w$ qubits as a parity accumulator. A sequence of $w-1$ CNOT gates computes the parity of the participating qubits onto the accumulator. Applying a single-qubit $Z$ rotation to the accumulator and subsequently reversing the CNOT sequence implements the desired $w$-qubit phase rotation. Thus, in the standard auxiliary-system-free 
construction, the operation requires
$2(w-1)$
CNOT gates, together with 
one single-qubit $Z$ rotation. Equivalently, the construction can be expressed in terms of controlled-$Z$ gates by conjugating the appropriate target qubits with Hadamard gates. The required parity ladders can be parallelized when the hardware connectivity permits, thereby reducing their circuit depth.

Rydberg-ion systems provide another possible architecture. Strong dipolar interactions between Rydberg-excited ions can facilitate multi-qubit phase operations, potentially allowing some of the required interactions to be implemented using short pulse sequences \cite{RydbergIons}. In superconducting-qubit platforms, single-qubit rotations and local two-qubit entangling gates are routinely available \cite{SuperconductingQubits}. Higher-weight and geometrically non-local phase rotations must generally be compiled into the native gate set, with an overhead determined by the connectivity of the device.

Finally, the state certification protocol requires only single-qubit Pauli measurements after compu\-ta\-tio\-nal-basis measurements of the remaining qubits. These operations, as well as single-qubit readout, are standard capabilities of the architectures discussed above. We emphasize, however, that the role of this protocol in the present work is primarily conceptual: its measurement statistics on Hamiltonian phase states can be efficiently simulated classically, and it is precisely this simulability that allows us to prove that the Hamiltonian phase state assumptions imply one-way functions.

\section{Proof of Theorem~\ref{thm:HPSforaset}}\label{app:proof-certification}

We start by stating and proving Theorem 6 from Ref.~\cite{HuangPreskillSoleimanifar2025}, as the proof of Theorem~\ref{thm:HPSforaset} only requires a slight modification. To do this, consider the following Algorithm, which is Protocol 2 from Ref.~\cite{HuangPreskillSoleimanifar2025}, expressed in the language of Section~\ref{ss:state certification-prelim}:

\begin{algorithm}[H]
\caption{: Protocol 2 ($m=1$) from Ref.~\cite{HuangPreskillSoleimanifar2025}}\label{alg:HKP-protocol}
\begin{algorithmic}[1]
\Require Amplitude oracle $\mathsf{O}(|\phi\rangle)$, $t$ copies of unknown state $|\psi\rangle^{\otimes t}$ relaxation time $\tau$ and tolerance $\epsilon$.
\State $s \gets \mathcal{M}(|\psi\rangle^{\otimes t})$
\State $o\gets F_{|\phi\rangle}(s, \tau,\epsilon)$
\Ensure $o\in\{\mathsf{Accept},\mathsf{Reject}\}$.
\end{algorithmic}
\end{algorithm}

We then have the following:

\begin{theorem}[Theorem 6 from Ref.~\cite{HuangPreskillSoleimanifar2025}]\label{thm:thm6HPS}Given any target state $|\phi\rangle$ with relaxation time $\tau$, error $\epsilon>0$ and failure probability $\delta$, running Algorithm~\ref{alg:HKP-protocol} with
\begin{equation}
t = 64\frac{\tau^2}{\epsilon^2}\log\left(\frac{2}{\delta}\right)
\end{equation}
ensures that, with probability at least $1-\delta$:
\begin{enumerate}
\item If $|\langle \psi|\phi\rangle|^2 > 1-\epsilon/(2\tau)$ then $o=\mathsf{Accept}$,
\item If $|\langle \psi|\phi\rangle|^2 < 1-\epsilon$ then $o=\mathsf{Reject}$.
\end{enumerate}
\end{theorem}
\begin{proof} We start by introducing some notation. In particular 
\begin{enumerate}
\item We use the notation $\mathbb{E}[\omega](|\psi\rangle,|\phi\rangle)$ to denote the \textit{shadow overlap} between $|\psi\rangle$ and $|\phi\rangle$ defined via
\begin{equation}
\mathbb{E}[\omega](|\psi\rangle,|\phi\rangle) = \langle \psi|L(|\phi\rangle)|\psi\rangle,
\end{equation}
where $L(|\phi\rangle)$ is the operator defined in 
Eq.\ (21) of Ref.~\cite{HuangPreskillSoleimanifar2025} (note that we use $|\phi\rangle$ to denote the target state here, but in~\cite{HuangPreskillSoleimanifar2025} the target state is denoted with $|\psi\rangle$).
\item We use the notation $\hat{\omega}_t(|\psi\rangle,|\phi\rangle)$ to denote the output of line 7 of Algorithm~\ref{alg:measurement-protocol-HPS} for $F_{|\phi\rangle}$ when run on input $s\gets\mathcal{M}(|\psi\rangle^{\otimes t})$.
\end{enumerate}
With the notation above, the proof of Theorem~\ref{thm:thm6HPS} proceeds in two steps:

\begin{enumerate}
\item Show that given $\mathbb{E}[\omega](|\psi\rangle,|\phi\rangle)$, one can decide whether $|\langle \psi|\phi\rangle|^2 > 1-\epsilon/(2\tau)$ or $|\langle \psi|\phi\rangle|^2 < 1-\epsilon$.
\item Bound the sample complexity necessary for $\hat{\omega}_t(|\psi\rangle,|\phi\rangle)$ to be sufficiently well concentrated around $\mathbb{E}[\omega](|\psi\rangle,|\phi\rangle)$.
\end{enumerate}
With this in mind, the first tool is the following:
\begin{lemma}[Adapted from Theorem 4 from Ref.~\cite{HuangPreskillSoleimanifar2025}]\label{lem:shadow-norm} Given a state $|\phi\rangle$ with relaxation time $\tau$, one has the following:
\begin{enumerate}
 \item If $|\langle \psi|\phi\rangle|^2 < 1-\epsilon$ then  $\mathbb{E}[\omega](|\psi\rangle,|\phi\rangle)< 1-\epsilon/\tau$,
 \item If $|\langle \psi|\phi\rangle|^2 > 1-\epsilon/(2\tau)$ then  $\mathbb{E}[\omega](|\psi\rangle,|\phi\rangle)\geq 1-\epsilon/(2\tau)$,
\end{enumerate}
\end{lemma}
As such, if one wants to decide whether $|\langle \psi|\phi\rangle|^2 < 1-\epsilon$ or $|\langle \psi|\phi\rangle|^2 > 1-\epsilon/(2\tau)$ its sufficient to check whether $\mathbb{E}[\omega](|\psi\rangle,|\phi\rangle)< 1-\epsilon/\tau$ or $\mathbb{E}[\omega](|\psi\rangle,|\phi\rangle)\geq 1-\epsilon/(2\tau)$.

With this established, we now want to understand the concentration of $\hat{\omega}_t(|\psi\rangle,|\phi\rangle)$ around $\mathbb{E}[\omega](|\psi\rangle,|\phi\rangle)$. To this end, note that Eq. (30) of Ref.~\cite{HuangPreskillSoleimanifar2025} can be written as
\begin{equation}\label{eq:concentration-inequality}
\underset{s\,\gets\,{\mathcal{M}}\left(\ket\psi^{\otimes t}\right)}{\mathrm{Pr}}\left[\left|\hat{\omega}_t(|\psi\rangle,|\phi\rangle) - \mathbb{E}[\omega](|\psi\rangle,|\phi\rangle)\right| > \tilde{\epsilon}\right]\leq 2e^{-t\tilde{\epsilon}^2/2},
\end{equation}
in the case $m=1$.  From the above, we then get the following lemma:
\begin{lemma}[Theorem 5 from Ref.~\cite{HuangPreskillSoleimanifar2025}]
For any fixed state $|\phi\rangle$, one has
\begin{equation}
\underset{s\,\gets\,{\mathcal{M}}\left(\ket\psi^{\otimes t}\right)}{\mathrm{Pr}}\left[\left|\hat{\omega}_t(|\psi\rangle,|\phi\rangle) - \mathbb{E}[\omega](|\psi\rangle,|\phi\rangle)\right| > \tilde{\epsilon}\right]\leq \delta
\end{equation}
provided
\begin{equation}
t \geq 2\frac{1}{\tilde{\epsilon}^2}\log\left(\frac{1}{\delta}\right).
\end{equation}
\end{lemma}
As a corollary of the above two lemmas, 
we see that if we set $\tilde{\epsilon} = \epsilon/(4\tau)$ then whenever 
\begin{equation}
t \geq 2\frac{1}{\tilde{\epsilon}^2}\log\left(\frac{1}{\delta}\right) = 64\frac{\tau^2}{\epsilon^2}\log\left(\frac{1}{\delta}\right)
\end{equation}
one has, with probability at least $1-\delta$, that:
\begin{enumerate}
 \item If $|\langle \psi|\phi\rangle|^2 < 1-\epsilon$ then  $\hat{\omega}_t(|\psi\rangle,|\phi\rangle)< 1-(3\epsilon)/(4\tau)$,
 \item If $|\langle \psi|\phi\rangle|^2 > 1-\epsilon/(2\tau)$ then  $\hat{\omega}_t(|\psi\rangle,|\phi\rangle)\geq 1-(3\epsilon/(4\tau)$.
\end{enumerate}
Now, the theorem statement follows from the fact that the output of $F_{|\phi\rangle}(s)$ on $s\gets\mathcal{M}(|\psi\rangle^{\otimes t})$ is determined precisely by whether $\hat{\omega}_t(|\psi\rangle,|\phi\rangle)$ is greater than or less than $1-(3\epsilon/(4\tau)$.
\end{proof}

With this established, we note the following Lemma:
\begin{lemma}[Lemma 27 from Ref.~\cite{HuangPreskillSoleimanifar2025}] For any phase state $|\phi\rangle$ the relaxation time is given by $\tau=n$.
\end{lemma}
With the above, we are finally ready to prove Theorem~\ref{thm:HPSforaset};

\begin{proof}[Proof of Theorem~\ref{thm:HPSforaset}] In order for the state certification protocol to work \textit{simultaneously} for all states $|\phi_k\rangle$ one simply has to bound the sample complexity required for $\hat{\omega}_t(|\psi\rangle,|\phi_k\rangle)$ to be sufficiently well concentrated around $\mathbb{E}[\omega](|\psi\rangle,|\phi_k\rangle)$ \textit{for all} $|\phi_k\rangle$ simultaneously. To do this, one takes a union bound using Eq.~\eqref{eq:concentration-inequality}, and finds
\begin{equation}
\underset{s\,\gets\,{\mathcal{M}}\left(\ket\psi^{\otimes t}\right)}{\mathrm{Pr}}\left[\exists\, k : \left|\hat{\omega}_t(|\psi\rangle,|\phi_k\rangle) - \mathbb{E}[\omega](|\psi\rangle,|\phi_k\rangle)\right| > \tilde{\epsilon}\right] < |\mathbb{K}|2e^{-t\tilde{\epsilon}^2/2},
\end{equation}
from which it follows that
\begin{equation}
\underset{s\,\gets\,{\mathcal{M}}\left(\ket\psi^{\otimes t}\right)}{\mathrm{Pr}}\left[\exists\, k : \left|\hat{\omega}_t(|\psi\rangle,|\phi_k\rangle) - \mathbb{E}[\omega](|\psi\rangle,|\phi_k\rangle)\right| > \tilde{\epsilon}\right] < \delta,
\end{equation}
provided that $t\geq \frac{2}{\tilde{\epsilon}^2}\log\left(\frac{2|\mathbb{K}|}{\delta}\right)$. Again setting $\tilde{\epsilon} = \epsilon/(4\tau)$ and using Lemma~\ref{lem:shadow-norm} we then have
\begin{equation}
\underset{s\gets \mathcal{M}(|\psi\rangle^{\otimes t})}{\mathrm{Pr}}\left[\forall \, k\in \mathbb{K} \, \begin{cases} \hat{\omega}_t(|\psi\rangle,|\phi_k\rangle) > 1-(3\epsilon/(4\tau) \text{ if } |\langle\psi|\phi_k\rangle|^2 > 1-\epsilon/(2\tau) \\  \hat{\omega}_t(|\psi\rangle,|\phi_k\rangle) < 1-(3\epsilon/(4\tau)\text{ if } |\langle \psi|\phi_k\rangle|^2 < 1-\epsilon \end{cases}\right] > 1-\delta,
\end{equation}
provided that $t\geq 64\frac{\tau^2}{\epsilon^2}\log\left(\frac{2|\mathbb{K}|}{\delta}\right)$. By using the fact that $\tau=n$ for phase states, the theorem statement follows.
\end{proof}

\section{Proof of Lemma~\ref{lem:owp-robustness}}\label{app:robustness-proof}

\begin{proof}
We need to prove both correctness and security of $(\mathsf{Samp'},\mathsf{Ver})$. For both, we utilize the fact that for any function $h$ with $h(k,s)\in [0,1]$ for all $(k,s)$ we have
\begin{equation}
\left|\underset{(k,s)\gets\mathsf{Samp}'(1^\lambda)}{\mathbb{E}}[h(k,s)] - \underset{(k,s)\gets\mathsf{Samp}(1^\lambda)}{\mathbb{E}}[h(k,s)]\right| \leq \epsilon(\lambda).
\end{equation}

\textit{Correctness.} Using the fact that $\mathsf{Ver}(k,s)\in\{0,1\}$ 
we have
\begin{align}
\underset{(k,s)\gets\mathsf{Samp}'(1^\lambda)}{\mathrm{Pr}}[\mathsf{Ver}(k,s) = 1] &=\underset{(k,s)\gets\mathsf{Samp}'(1^\lambda)}{\mathbb{E}}[\mathsf{Ver}(k,s)] \\
\nonumber
&\geq \underset{(k,s)\gets\mathsf{Samp}(1^\lambda)}{\mathbb{E}}[\mathsf{Ver}(k,s)] - \epsilon(\lambda) \\
\nonumber
&=\underset{(k,s)\gets\mathsf{Samp}(1^\lambda)}{\mathrm{Pr}}[\mathsf{Ver}(k,s) = 1] - \epsilon(\lambda) \\
\nonumber
& \geq 1-\mathsf{negl}(\lambda) - \epsilon(\lambda)\\
\nonumber
&\geq 1-\mathsf{negl}(\lambda).
\nonumber
\end{align}

\textit{Security:} For any QPT adversary $\mathcal{A}(s)\rightarrow k'$ define the function 
\begin{equation}
G_\mathcal{A}(k,s) = \underset{k'\leftarrow\mathcal{A}(s)}{\mathrm{Pr}}[\mathsf{Ver}(k',s) = 1]\in [0,1].
\end{equation}
Then for any QPT adversary $\mathcal{A}$ we have
\begin{align}
    \underset{\substack{
    (k,s)\,\gets\,{\sf Samp'}(1^\lambda)\,\,\\
    k'\,\gets\,{\mathcal A}(s)
    }}{\rm Pr}\big[{\sf Ver}(k',s)=1\big] &= \underset{(k,s)\gets\mathsf{Samp}'(1^\lambda)}{\mathbb{E}}[G_\mathcal{A}(k,s)] 
    \\
    \nonumber&\leq\underset{(k,s)\gets\mathsf{Samp}(1^\lambda)}{\mathbb{E}}[G_\mathcal{A}(k,s)] + \epsilon(\lambda)\\
    \nonumber
    & = \underset{\substack{
    (k,s)\,\gets\,{\sf Samp}(1^\lambda)\,\,
    \\
    \nonumber
    k'\,\gets\,{\mathcal A}(s)
    }}{\rm Pr}\big[{\sf Ver}(k',s)=1\big] + \epsilon(\lambda) \\
    \nonumber
    &\leq \mathsf{negl}(\lambda)+ \epsilon(\lambda)\\
    \nonumber
    &\leq \mathsf{negl}(\lambda).
    \nonumber
\end{align}
\end{proof}

\section{The one-way function obtained from Search HPS}\label{app:explicit-OWF}

In the main text, the OWF $\{F_\lambda\}$ implied by Theorem~\ref{thm:HPS-implies-owf} (under the Search HPS assumption) has been left implicit. In this appendix we make this OWF explicit. To this end:
\begin{enumerate}
\item We start from the Search HPS assumption, with functions $q,m$ and distribution $\chi_{q,m,\lambda}$ over $\mathbb{K}_{q,m,\lambda}$. 
\item Let $(\mathsf{KeyGen},\mathsf{StateGen},\mathsf{Ver})$ be the OWSG implied by the Search HPS assumption (i.e. as per Construction~\ref{con:OWSG-from-HPS}). 
\item Let $(\mathcal{M},\mathcal{F})$ be the $\eta$-simulable computationally efficient state certification protocol for $\{|\phi_k\rangle\,|\,k\in\mathbb{K}_{q,m,\lambda}\}$ from $t$ copies given in Theorem~\ref{thm:HPSforaset}.
\item Let $(\mathsf{Samp},\mathsf{Ver})$ be the EV-OWP constructed from $((\mathsf{KeyGen},\mathsf{StateGen},\mathsf{Ver}),(\mathcal{M},\mathcal{F}))$ via Construction~\ref{con:OWP-from-OWSG-state certification}. In particular, recall that $\mathsf{Samp}(1^\lambda) \rightarrow (k,s)$ with  $k\in\mathbb{K}_{q,m,\lambda}$ and $s = [s^{(1)},\ldots s^{(t(\lambda))}]$ where
\begin{enumerate}
\item $t(\lambda) = t(\lambda,\epsilon=1/8,\delta = 2^{-\lambda},|\mathbb{K}_{q,m\lambda}|)$
\item  $s^{(j)} = (i^{(j)},z^{(j)},M^{(j)},b^{(j)})$ with $i^{(j)}\in [\lambda]$, $z^{(j)}\in\{0,1\}^{\lambda-1}$, $M^{(j)}\in \{X,Y,Z\}$ and $b^{(j)}\in \{0,1\}$
\end{enumerate}
For convenience we define $S(\lambda) \coloneqq  \{[\lambda]\times\{0,1\}^{\lambda-1}\times \{X,Y,Z\}\times\{0,1\}\}^{t(\lambda)}$ so that $s\in S(\lambda)$.
\item Let $\overline{\mathsf{Samp}}_C(1^\lambda)\rightarrow (k,s)\in \mathbb{K}_{q,m,\lambda}\times S(\lambda)$ be the PPT algorithm for approximately sampling from $\mathsf{Samp}(1^\lambda)$, which does the following on input $1^\lambda$:
\begin{enumerate}
\item  $k\gets \chi_{q,m,\lambda}$ 
\item $s= [s^{(1)},\ldots s^{(t(\lambda))}]\gets \mathcal{M}_C(k,1^\lambda,1^{t(\lambda)})$ where $\mathcal{M}_C$ is given in Algorithm~\ref{alg:PPT-approx-M}.
\end{enumerate}

\end{enumerate}
To define the OWF via the construction in Theorem~\ref{thm:owf-from-classical-evowp} we now need to ``pull out'' the randomness from $\overline{\mathsf{Samp}}_C$.  To do this, we construct efficiently computable deterministic functions 
\begin{enumerate}
\item $k_\lambda:\{0,1\}^{p_1(\lambda)}\rightarrow \mathbb{K}_{q,m,\lambda}$
\item $s_\lambda:\{0,1\}^{p_1(\lambda)}\times\big\{ [\lambda]\times\{0,1\}^{\lambda-1}\times \{X,Y,Z\}\big\}^{t(\lambda)}\times\big\{\{0,1\}^\lambda\big\}^{t(\lambda)}\rightarrow S(\lambda)$
\end{enumerate}
which satisfy
\begin{equation}
\big(k_\lambda(r)\big)_{r\gets\{0,1\}^{p_1(\lambda)}} = \chi_{q,m,\lambda}
\end{equation}
and 
\begin{equation}\label{eq:pulled-out}
\overline{\mathsf{Samp}}_C(1^\lambda) = \big(k_\lambda(r_0),s_{\lambda}(r_0,r_1,r_2)\big)_{\begin{subarray}l r_0\gets\{0,1\}^{p_1(\lambda)}\\r_1\gets\big\{ [\lambda]\times\{0,1\}^{\lambda-1}\times \{X,Y,Z\}\big\}^{t(\lambda)} \\r_2\gets \big\{\{0,1\}^\lambda\big\}^{t(\lambda)}\end{subarray}}
\end{equation}
To this end, note that it follows from the assumption that $\chi_{q,m,\lambda}$ can be efficiently sampled from (with respect to $\lambda$) that there exists a polynomial $p_1$ and an efficiently computable deterministic function $\tilde{\chi}_{q,m,\lambda}:\{0,1\}^{p_1(\lambda)}\rightarrow\mathbb{K}_{q,m,\lambda} $ satisfying
\begin{equation}
\big(\tilde{\chi}_{q,m,\lambda}(r)\big)_{r\gets\{0,1\}^{p_1(\lambda)}} = \chi_{q,m,\lambda}.
\end{equation}
Given this, we simply define $k_\lambda:\{0,1\}^{p_1(\lambda)}\rightarrow\mathbb{K}_{q,m,\lambda}$ via $k_\lambda(r) =\tilde{\chi}_{q,m,\lambda}(r)$. 

To define $s_\lambda$, we start by noting that $\mathcal{M}_C$ works by calling $\hat{\mathcal{M}}_C$ (Algorithm~\ref{alg:classical-simulation-approx}) iteratively, and therefore to pull the randomness out of $\mathcal{M}_C$ we have to pull the randomness out of $\hat{\mathcal{M}}_C$, To this, end let $\tilde{\mathcal{B}}_j:\{0,1\}^{\lambda-1}\times [\lambda]\times \mathbb{K}_{q,m,\lambda}\times \{0,1\}^\lambda\rightarrow \{0,1\}$ be the efficiently computable deterministic function satisfying
\begin{equation}\label{eq:ber-pulled-out}
\left(\tilde{\mathcal{B}}_j(z,i,k,r)\right)_{r\gets\{0,1\}^\lambda} = \mathcal{B}_j(z,i,k,\lambda)
\end{equation}
where $\mathcal{B}_j$ is as defined in Lemma~\ref{lem:eff-approx-Bernoulli}. With this, 
we then define 
\begin{equation}
\tilde{s}_\lambda: \mathbb{K}_{q,m,\lambda} \times [\lambda]\times\{0,1\}^{\lambda-1}\times \{X,Y,Z\}\times \{0,1\}^\lambda\rightarrow [\lambda]\times \{0,1\}^{\lambda-1}\times\{X,Y,Z\}\times \{0,1\}
\end{equation}
via
\begin{equation}
\tilde{s}_\lambda(k,i,z,M,r)) = \begin{cases} (i,z,M, r(1)) &\text{if }M=Z\\
(i,z,M, \tilde{\mathcal{B}}_1(z,i,k,r)) &\text{if }M=X\\
(i,z,M, \tilde{\mathcal{B}}_2(z,i,k,r)) &\text{if }M=Y
\end{cases}
\end{equation}
where $r(1)$ denotes the first bit of the string $r$. With these definitions, it then follows from Eq.~\eqref{eq:ber-pulled-out} and the definition of $\hat{\mathcal{M}}_C$ that
\begin{equation}
\hat{\mathcal{M}}_C(k,\lambda,\lambda) = \big(\tilde{s}_\lambda(k,i,z,M,r)\big)_{\begin{subarray}l
(i,z,m)\gets [\lambda]\times \{0,1\}^{\lambda -1}\times \{X,Y,Z\} \\
r\gets \{0,1\}^\lambda
\end{subarray}}
\end{equation}
In other words, $\tilde{s}_\lambda$ is the function obtained by pulling out the randomness from $\hat{\mathcal{M}}_C$. With this, for any
\begin{equation}
(r_0,r_1,r_2) \in \{0,1\}^{p_1(\lambda)}\times\big\{ [\lambda]\times\{0,1\}^{\lambda-1}\times \{X,Y,Z\}\big\}^{t(\lambda)}\times\big\{\{0,1\}^\lambda\big\}^{t(\lambda)}
\end{equation}
with
\begin{align}
r_1 &= \left((i^{(1))},z^{(1))},M^{(1)}),\ldots,   (i^{(t(\lambda))},z^{(t(\lambda))},M^{(t(\lambda))}) \right)\\
r_2 &= (r^{(1)},\ldots, r^{t(\lambda)})
\end{align}
we can finally define the function $s_\lambda:\{0,1\}^{p_1(\lambda)}\times\big\{ [\lambda]\times\{0,1\}^{\lambda-1}\times \{X,Y,Z\}\big\}^{t(\lambda)}\times\big\{\{0,1\}^\lambda\big\}^{t(\lambda)}\rightarrow S(\lambda)$ via
\begin{align}
s_{\lambda}(r_0,r_1,r_2) = \left(\tilde{s}_\lambda \big(k_\lambda(r_0),i^{(1)},z^{(1)},M^{(1)},r^{(1)}\big) ,\ldots, \tilde{s}_\lambda \big(k_\lambda(r_0),i^{(t(\lambda))},z^{(t(\lambda))},M^{(t(\lambda))},r^{(t(\lambda))}\big)\right),
\end{align}
which can be checked to satisfy Eq.~\eqref{eq:pulled-out}. With this established, let's define
\begin{equation}
R(\lambda) = \{0,1\}^{p_1(\lambda)}\times\big\{ [\lambda]\times\{0,1\}^{\lambda-1}\times \{X,Y,Z\}\big\}^{t(\lambda)}\times\big\{\{0,1\}^\lambda\big\}^{t(\lambda)}.
\end{equation}
Using the construction of Theorem~\ref{thm:owf-from-classical-evowp}, the OWF implied by Theorem~\ref{thm:HPS-implies-owf} is then
$\{F_\lambda: R(\lambda)\rightarrow \{0,1\}^*\}$  with 
\begin{equation}
F_\lambda(r) = 
\begin{cases}
(0,s_\lambda(r)) & r\in\mathrm{Good}_\lambda \\
(1,r) & r\in\mathrm{Bad}_\lambda
\end{cases}
\end{equation}
where
\begin{align}
\mathrm{Good}_\lambda &= \{r\,|\, \mathsf{Ver}(k_\lambda(r),s_\lambda(r)) = \mathcal{F}(k_\lambda(r),s_\lambda(r),1/8,2^{-\lambda})=1\},\\
\mathrm{Bad}_\lambda &= \{r\,|\, \mathsf{Ver}(k_\lambda(r),s_\lambda(r))=\mathcal{F}(k_\lambda(r),s_\lambda(r),1/8,2^{-\lambda})=0\}.
\end{align}
and $\mathcal{F}$ is the efficiently computable post-processing function of the state certification protocol, defined in Eq.~\eqref{eq:post-processing-def}. As mentioned in Remark~\ref{rem:explicit-owf}, we note that the domain $R(\lambda)$ depends on the copy complexity $t(\lambda)$, and that the length of strings in the co-domain could be reduced by considering state certification protocols with improved copy-complexity.

\printbibliography

@misc{srini-tomography,
      title={Tomography of quantum states with bounded extent}, 
      author={Srinivasan Arunachalam and Arkopal Dutt},
      year={2026},
      eprint={2606.07425},
      archivePrefix={arXiv},
      optprimaryClass={quant-ph},
      url={https://arxiv.org/abs/2606.07425}, 
}

@article{regevOnLattices2009,
author = {Regev, Oded},
title = {On lattices, learning with errors, random linear codes, and cryptography},
year = {2009},
issue_date = {September 2009},
publisher = {Association for Computing Machinery},
address = {New York, NY, USA},
volume = {56},
optnumber = {6},
optissn = {0004-5411},
doi = {10.1145/1568318.1568324},
journal = {J. ACM},
articleno = {34},
}

@misc{badescu2017quantumstatecertification,
      title={Quantum state certification}, 
      author={Costin B\u adescu and Ryan O'Donnell and John Wright},
      year={2017},
      eprint={1708.06002},
      archivePrefix={arXiv},
      primaryClass={quant-ph},
      url={https://arxiv.org/abs/1708.06002}, 
}

@misc{coladangelo2026powerbasesrobustcopyoptimal,
      title={The Power of Two Bases: Robust and copy-optimal certification of nearly all quantum states with few-qubit measurements}, 
      author={Andrea Coladangelo and Jerry Li and Joseph Slote and Ellen Wu},
      year={2026},
      eprint={2602.11616},
      archivePrefix={arXiv},
      primaryClass={quant-ph},
      url={https://arxiv.org/abs/2602.11616}, 
}

@inproceedings{chen2022exponential,
  title={Exponential separations between learning with and without quantum memory},
  author={Chen, Sitan and Cotler, Jordan and Huang, Hsin-Yuan and Li, Jerry},
  booktitle={2021 IEEE 62nd Annual Symposium on Foundations of Computer Science (FOCS)},
  pages={574--585},
  year={2022},
  organization={IEEE}
}

@misc{bostanci2025efficientquantumpseudorandomnesshamiltonian,
      title={Efficient quantum pseudorandomness from Hamiltonian phase states}, 
      author={John Bostanci and Jonas Haferkamp and Dominik Hangleiter and Alexander Poremba},
      year={2025},
      eprint={2410.08073},
      archivePrefix={arXiv},
      optprimaryClass={quant-ph},
      opturl={https://arxiv.org/abs/2410.08073}, 
}

@misc{khurana2024commitmentsquantumonewayness,
      title={Commitments from quantum one-wayness}, 
      author={Dakshita Khurana and Kabir Tomer},
      year={2024},
      eprint={2310.11526},
      archivePrefix={arXiv},
      optprimaryClass={quant-ph},
      opturl={https://arxiv.org/abs/2310.11526}, 
}

@misc{morimae2024onewaynessquantumcryptography,
      title={One-wayness in quantum cryptography}, 
      author={Tomoyuki Morimae and Takashi Yamakawa},
      year={2024},
      eprint={2210.03394},
      archivePrefix={arXiv},
      optprimaryClass={quant-ph},
      opturl={https://arxiv.org/abs/2210.03394}, 
}

@article{Cavalar_2025,
   title={On the computational hardness of quantum one-wayness},
   volume={9},
   optissn={2521-327X},
   opturl={http://dx.doi.org/10.22331/q-2025-03-27-1679},
   DOI={10.22331/q-2025-03-27-1679},
   journal={Quantum},
   publisher={Verein zur Forderung des Open Access Publizierens in den Quantenwissenschaften},
   author={Cavalar, Bruno and Goldin, Eli and Gray, Matthew and Hall, Peter and Liu, Yanyi and Pelecanos, Angelos},
   year={2025},
   optmonth=mar, pages={1679} }

@inbook{Morimae_2022,
   title={Quantum commitments and signatures without one-way functions},
   optISBN={9783031158025},
   optissn={1611-3349},
   opturl={http://dx.doi.org/10.1007/978-3-031-15802-5_10},
   DOI={10.1007/978-3-031-15802-5_10},
   booktitle={Advances in Cryptology – CRYPTO 2022},
   publisher={Springer Nature Switzerland},
   author={Morimae, Tomoyuki and Yamakawa, Takashi},
   year={2022},
   pages={269–295} }

@InProceedings{Kretschmer_2021,
  author =	{Kretschmer, William},
  title =	{{Quantum Pseudorandomness and Classical Complexity}},
  booktitle =	{16th Conference on the Theory of Quantum Computation, Communication and Cryptography (TQC 2021)},
  pages =	{2:1--2:20},
  series =	{Leibniz International Proceedings in Informatics (LIPIcs)},
  ISBN =	{978-3-95977-198-6},
  ISSN =	{1868-8969},
  year =	{2021},
  volume =	{197},
  editor =	{Hsieh, Min-Hsiu},
  publisher =	{Schloss Dagstuhl -- Leibniz-Zentrum f{\"u}r Informatik},
  address =	{Dagstuhl, Germany},
  URL =		{https://drops.dagstuhl.de/entities/document/10.4230/LIPIcs.TQC.2021.2},
  URN =		{urn:nbn:de:0030-drops-139975},
  doi =		{10.4230/LIPIcs.TQC.2021.2}
}

@InProceedings{PRSGdefinition,
author="Ji, Zhengfeng
and Liu, Yi-Kai
and Song, Fang",
editor="Shacham, Hovav
and Boldyreva, Alexandra",
title="Pseudorandom quantum states",
booktitle="Advances in Cryptology -- CRYPTO 2018",
year="2018",
publisher="Springer International Publishing",
address="Cham",
pages="126--152",
}

@misc{brakerski2020scalablepseudorandomquantumstates,
      title={Scalable pseudorandom quantum states}, 
      author={Zvika Brakerski and Omri Shmueli},
      year={2020},
      eprint={2004.01976},
      archivePrefix={arXiv},
      optprimaryClass={quant-ph},
      opturl={https://arxiv.org/abs/2004.01976}, 
}

@InProceedings{Prabhanjan,
author="Ananth, Prabhanjan
and Gulati, Aditya
and Qian, Luowen
and Yuen, Henry",
editor="Kiltz, Eike
and Vaikuntanathan, Vinod",
title="Pseudorandom (function-like) quantum state generators: New definitions and applications",
booktitle="Theory of Cryptography",
year="2022",
publisher="Springer Nature Switzerland",
address="Cham",
pages="237--265",
optISBN="978-3-031-22318-1"
}

@misc{ananth2023pseudorandomstringspseudorandomquantum,
      title={Pseudorandom strings from pseudorandom quantum states}, 
      author={Prabhanjan Ananth and Yao-Ting Lin and Henry Yuen},
      year={2023},
      eprint={2306.05613},
      archivePrefix={arXiv},
      optprimaryClass={quant-ph},
      opturl={https://arxiv.org/abs/2306.05613}, 
}

@inproceedings{Kretschmer_2023, 
series={STOC ’23},
   title={Quantum Cryptography in Algorithmica},
   opturl={http://dx.doi.org/10.1145/3564246.3585225},
   DOI={10.1145/3564246.3585225},
   booktitle={Proceedings of the 55th Annual ACM Symposium on Theory of Computing},
   publisher={ACM},
   author={Kretschmer, William and Qian, Luowen and Sinha, Makrand and Tal, Avishay},
   year={2023},
   optmonth=jun, pages={1589–1602},
   collection={STOC ’23} }

@inproceedings{Kretschmer_2025, series={STOC ’25},
   title={Quantum-computable one-way functions without one-way functions},
   opturl={http://dx.doi.org/10.1145/3717823.3718144},
   DOI={10.1145/3717823.3718144},
   booktitle={Proceedings of the 57th Annual ACM Symposium on Theory of Computing},
   publisher={ACM},
   author={Kretschmer, William and Qian, Luowen and Tal, Avishay},
   year={2025},
   optmonth=jun, pages={189–200},
   collection={STOC ’25} }

@inproceedings{khurana2025founding,
  title={Founding quantum cryptography on quantum advantage, or, towards cryptography from {\#P} hardness},
  author={Khurana, Dakshita and Tomer, Kabir},
  booktitle={Proceedings of the 57th Annual ACM Symposium on Theory of Computing},
  pages={178--188},
  year={2025}
}

@misc{microcryptzoo,
  author       = {Sattath, Or},
  title        = {{MicroCrypt Zoo}: An interactive visualization of quantum cryptographic primitives},
  year         = {2024},
  howpublished = {\url{https://sattath.github.io/microcrypt-zoo/}},
  note         = {Accessed: 2026-04-04}
}

@article{SuperconductingQubits,
title={Superconducting qubits: Current state of play},
author={Morten Kjaergaard and Mollie E. Schwartz and Jochen Braumüller and Philip Krantz and Joel I-Jan Wang and Simon Gustavsson and William D. Oliver},
journal={Ann. Rev. Cond. Matt. Phys.}, volume=11, pages={369-395}, year=2020,
DOI={10.1146/annurev-conmatphys-031119-050605}}

@article{RydbergIons,
title={{Trapped Rydberg ions: From spin chains to fast quantum gates}},
author={Markus Mueller and Lin-Mei Liang and Igor Lesanovsky and Peter Zoller},
journal={New J. Phys.}, volume=10, pages=093009, year=2008,
DOI={10.1088/1367-2630/10/9/093009}}

@inproceedings{ananth2022cryptography,
  title={Cryptography from pseudorandom quantum states},
  author={Ananth, Prabhanjan and Qian, Luowen and Yuen, Henry},
  booktitle={Annual International Cryptology Conference},
  pages={208--236},
  year={2022},
  organization={Springer}
}

@article{grilo2025quantum,
  title={Quantum pseudoresources imply cryptography},
  author={Grilo, Alex B. and Y{\'a}ng{\"u}ez, {\'A}lvaro},
  eprint    = {2504.15025},
  archiveprefix = {arXiv},
  year={2025}
}

@misc{anshu2023surveycomplexitylearningquantum,
      title={A survey on the complexity of learning quantum states}, 
      author={Anurag Anshu and Srinivasan Arunachalam},
      year={2023},
      eprint={2305.20069},
      archivePrefix={arXiv},
      primaryClass={quant-ph},
      url={https://arxiv.org/abs/2305.20069}, 
}

@misc{goldin2024countcrypt,
      title={CountCrypt: Quantum Cryptography between QCMA and PP}, 
      author={Eli Goldin and Tomoyuki Morimae and Saachi Mutreja and Takashi Yamakawa},
      year={2025},
      eprint={2410.14792},
      archivePrefix={arXiv},
      primaryClass={quant-ph},
      url={https://arxiv.org/abs/2410.14792}, 
}

@InProceedings{BCQ2023,
  author    = {Brakerski, Zvika and Canetti, Ran and Qian, Luowen},
  title     = {{On the computational hardness needed for quantum cryptography}},
  booktitle = {14th Innovations in Theoretical Computer Science Conference
               (ITCS 2023)},
  pages     = {24:1--24:21},
  series    = {Leibniz International Proceedings in Informatics (LIPIcs)},
  volume    = {251},
  editor    = {Tauman Kalai, Yael},
  publisher = {Schloss Dagstuhl -- Leibniz-Zentrum f{\"u}r Informatik},
  address   = {Dagstuhl, Germany},
  year      = {2023},
  optISBN      = {978-3-95977-263-1},
  optissn      = {1868-8969},
  doi       = {10.4230/LIPIcs.ITCS.2023.24},
  eprint    = {2209.04101},
  archiveprefix = {arXiv},
}

@article{Huang_2020,
   title={Predicting many properties of a quantum system from very few measurements},
   volume={16},
   optissn={1745-2481},
   opturl={http://dx.doi.org/10.1038/s41567-020-0932-7},
   DOI={10.1038/s41567-020-0932-7},
   number={10},
   journal={Nature Physics},
   publisher={Springer Science and Business Media LLC},
   author={Huang, Hsin-Yuan and Kueng, Richard and Preskill, John},
   year={2020},
   optmonth=jun, pages={1050–1057} }

@misc{niroula2026digitalsignaturesclassicalshadows,
      title={Digital signatures with classical shadows on near-term quantum computers}, 
      author={Pradeep Niroula and Minzhao Liu and Sivaprasad Omanakuttan and David Amaro and Shouvanik Chakrabarti and Soumik Ghosh and Zichang He and Yuwei Jin and Fatih Kaleoglu and Steven Kordonowy and Rohan Kumar and Michael A. Perlin and Akshay Seshadri and Matthew Steinberg and Joseph Sullivan and Jacob Watkins and Henry Yuen and Ruslan Shaydulin},
      year={2026},
      eprint={2602.04859},
      archivePrefix={arXiv},
      optprimaryClass={quant-ph},
      opturl={https://arxiv.org/abs/2602.04859}, 
}

@misc{fefferman2025hardnesslearningquantumcircuits,
      title={The hardness of learning quantum circuits and its cryptographic applications}, 
      author={Bill Fefferman and Soumik Ghosh and Makrand Sinha and Henry Yuen},
      year={2025},
      eprint={2504.15343},
      archivePrefix={arXiv},
      optprimaryClass={quant-ph},
      opturl={https://arxiv.org/abs/2504.15343}, 
}

@misc{QCCCCrypto,
      author = {Kai-Min Chung and Eli Goldin and Matthew Gray},
      title = {On central primitives for quantum cryptography with classical communication},
      howpublished = {Cryptology {ePrint} Archive, Paper 2024/356},
      year = {2024},
      opturl = {https://eprint.iacr.org/2024/356}
}

@misc{park2026samplehardwareefficientfidelityestimation,
      title={Sample- and hardware-efficient fidelity estimation by stripping phase-dominated magic}, 
      author={Guedong Park and Jaekwon Chang and Yosep Kim and Yong Siah Teo and Hyunseok Jeong},
      year={2026},
      eprint={2602.09710},
      archivePrefix={arXiv},
      optprimaryClass={quant-ph},
      opturl={https://arxiv.org/abs/2602.09710}, 
}

@misc{YihuiLSN1,
      title={The learning stabilizers with noise problem}, 
      author={Alexander Poremba and Yihui Quek and Peter Shor},
      year={2025},
      eprint={2410.18953},
      archivePrefix={arXiv},
      optprimaryClass={quant-ph},
      opturl={https://arxiv.org/abs/2410.18953}, 
}

@misc{YihuiLSN2,
      title={Post-quantum cryptography from quantum stabilizer decoding}, 
      author={Jonathan Z. Lu and Alexander Poremba and Yihui Quek and Akshar Ramkumar},
      year={2026},
      eprint={2603.19110},
      archivePrefix={arXiv},
      optprimaryClass={quant-ph},
      opturl={https://arxiv.org/abs/2603.19110}, 
}

@misc{KhesinLSN3,
      title={Average-case complexity of quantum stabilizer decoding}, 
      author={Andrey Boris Khesin and Jonathan Z. Lu and Alexander Poremba and Akshar Ramkumar and Vinod Vaikuntanathan},
      year={2025},
      eprint={2509.20697},
      archivePrefix={arXiv},
      optprimaryClass={quant-ph},
      opturl={https://arxiv.org/abs/2509.20697}, 
}

@misc{hiroka2025hardnessquantumdistributionlearning,
      title={Hardness of quantum distribution learning and quantum cryptography}, 
      author={Taiga Hiroka and Min-Hsiu Hsieh and Tomoyuki Morimae},
      year={2025},
      eprint={2507.01292},
      archivePrefix={arXiv},
      optprimaryClass={quant-ph},
      opturl={https://arxiv.org/abs/2507.01292}, 
}

@misc{hiroka2024computationalcomplexitylearningefficiently,
      title={Computational complexity of learning efficiently generatable pure states}, 
      author={Taiga Hiroka and Min-Hsiu Hsieh},
      year={2024},
      eprint={2410.04373},
      archivePrefix={arXiv},
      optprimaryClass={quant-ph},
      opturl={https://arxiv.org/abs/2410.04373}, 
}

@article{Hangleiter_2023,
   title={Computational advantage of quantum random sampling},
   volume={95},
   optissn={1539-0756},
   opturl={http://dx.doi.org/10.1103/RevModPhys.95.035001},
   DOI={10.1103/revmodphys.95.035001},
   number={3},
   journal={Rev. Mod. Phys.},
   publisher={American Physical Society (APS)},
   author={Hangleiter, Dominik and Eisert, Jens},
   year={2023},
   optmonth=jul }

@article{shepherd2009temporally,
  title={Temporally unstructured quantum computation},
  author={Shepherd, Dan and Bremner, Michael J.},
  journal={Proc. Roy. Soc. A},
  volume={465},
  number={2105},
  pages={1413--1439},
  year={2009},
  publisher={The Royal Society London}
}

@article{elben2023randomized,
  title={The randomized measurement toolbox},
  author={Elben, Andreas and Flammia, Steven T. and Huang, Hsin-Yuan and Kueng, Richard and Preskill, John and Vermersch, Benoit and Zoller, Peter},
  journal={Nature Rev. Phys.},
  volume={5},
  number={1},
  pages={9--24},
  year={2023},
  publisher={Nature Publishing Group UK London}
}

@article{HuangPreskillSoleimanifar2025,
  author  = {Huang, Hsin-Yuan and Preskill, John and Soleimanifar, Mehdi},
  title   = {Certifying almost all quantum states with few single-qubit measurements},
  journal = {Nat. Phys.},
  volume  = {21},
  pages   = {1834},
  year    = {2025},
  doi     = {10.1038/s41567-025-03025-1},
  eprint  = {2404.07281},
  archivePrefix = {arXiv},
  note    = {Also in Proc.\ FOCS 2024},
}

@misc{gupta2025singlequbitmeasurementssufficecertify,
      title={Few single-qubit measurements suffice to certify any quantum state}, 
      author={Meghal Gupta and William He and Ryan O'Donnell},
      year={2025},
      eprint={2506.11355},
      archivePrefix={arXiv},
      primaryClass={quant-ph},
      url={https://arxiv.org/abs/2506.11355}, 
}

@article{HelsenWalter2023,
  author  = {Helsen, Jonas and Walter, Michael},
  title   = {Thrifty shadow estimation: Reusing quantum circuits and bounding tails},
  journal = {Phys. Rev. Lett.},
  volume  = {131},
  number  = {24},
  pages   = {240602},
  year    = {2023},
  doi     = {10.1103/PhysRevLett.131.240602},
}

@article{GrewalIngram2024,
  author  = {Daniel Grier and Hakop Pashayan and Luke Schaeffer},
  title   = {Sample-optimal classical shadows for pure states},
  journal = {Quantum},
  volume  = {8},
  pages   = {1373},
  year    = {2024},
  doi     = {10.22331/q-2024-06-17-1373},
}

@article{KohGrewal2022,
  author  = {Koh, Dax Enshan and Grewal, Sabee},
  title   = {Classical shadows with noise},
  journal = {Quantum},
  volume  = {6},
  pages   = {776},
  year    = {2022},
  doi     = {10.22331/q-2022-08-16-776},
}

@article{ZhaoRubinMiyake2021,
  author  = {Zhao, Andrew and Rubin, Nicholas C. and Miyake, Akimasa},
  title   = {Fermionic partial tomography via classical shadows},
  journal = {Phys. Rev. Lett.},
  volume  = {127},
  number  = {11},
  pages   = {110504},
  year    = {2021},
  doi     = {10.1103/PhysRevLett.127.110504},
}

@article{WanHugginsLeeBabbush2023,
  author  = {Wan, Kianna and Huggins, William J. and Lee, Joonho and Babbush, Ryan},
  title   = {Matchgate shadows for fermionic quantum simulation},
  journal = {Commun. Math. Phys.},
  volume  = {404},
  pages   = {629--700},
  year    = {2023},
  doi     = {10.1007/s00220-023-04844-0},
  eprint  = {2207.13723},
  archivePrefix = {arXiv},
}

@article{BertoniEtAl2024ShallowShadows,
  author  = {Bertoni, Christian and Haferkamp, Jonas and Hinsche, Marcel and Ioannou, Marios and Eisert, Jens and Pashayan, Hakop},
  title   = {Shallow shadows: Expectation estimation using low-depth random {Clifford} circuits},
  journal={Phys. Rev. Lett.}, volume=133, pages={020602}, year=2024,
DOI={10.1103/PhysRevLett.133.020602}
}

@misc{ma2025constructrandomunitaries,
      title={How to construct random unitaries}, 
      author={Fermi Ma and Hsin-Yuan Huang},
      year={2025},
      eprint={2410.10116},
      archivePrefix={arXiv},
      optprimaryClass={quant-ph},
      opturl={https://arxiv.org/abs/2410.10116}, 
}

@misc{metger2024simpleconstructionslineardepthtdesigns,
      title={Simple constructions of linear-depth t-designs and pseudorandom unitaries}, 
      author={Tony Metger and Alexander Poremba and Makrand Sinha and Henry Yuen},
      year={2024},
      eprint={2404.12647},
      archivePrefix={arXiv},
      optprimaryClass={quant-ph},
      opturl={https://arxiv.org/abs/2404.12647}, 
}

@misc{brakerski2019pseudorandomquantumstates,
      title={(Pseudo) random quantum states with binary phase}, 
      author={Zvika Brakerski and Omri Shmueli},
      year={2019},
      eprint={1906.10611},
      archivePrefix={arXiv},
      optprimaryClass={quant-ph},
      opturl={https://arxiv.org/abs/1906.10611}, 
}

@inproceedings{Chen_2024,
   title={Efficient unitary designs from random sums and permutations},
   opturl={http://dx.doi.org/10.1109/FOCS61266.2024.00037},
   DOI={10.1109/focs61266.2024.00037},
   booktitle={2024 IEEE 65th Annual Symposium on Foundations of Computer Science (FOCS)},
   publisher={IEEE},
   author={Chen, Chi-Fang and Docter, Jordan and Xu, Michelle and Bouland, Adam and Brandao, Fernando G.~S.~L. and Hayden, Patrick},
   year={2024},
   optmonth=oct, pages={476–484} }

@article{bremner2011classical,
  title={Classical simulation of commuting quantum computations implies collapse of the polynomial hierarchy},
  author={Bremner, Michael J. and Jozsa, Richard and Shepherd, Dan J.},
  journal={Proc. 
  Roy. Soc. A},
  volume={467},
  number={2126},
  pages={459--472},
  year={2011},
  publisher={The Royal Society}
}

@article{Jandura,
title={Time-optimal two- and three-qubit gates for Rydberg atoms},
author={Sven Jandura and Guido Pupillo},
journal={Quantum}, 
volume=6, 
pages=712,
year=2022,
DOI={10.22331/q-2022-05-13-712}
}

@article{PhysRevLett.85.2208,
  title = {Fast quantum gates for neutral atoms},
  author = {Jaksch, Dieter and Cirac, J. Ignacio and Zoller, Peter and Rolston,  Steven L. and C\^ot\'e, Robin and Lukin, Mischa D.},
  journal = {Phys. Rev. Lett.},
  volume = {85},
  issue = {10},
  pages = {2208--2211},
  numpages = {0},
  year = {2000},
  optoptmonth = {Sep},
  publisher = {American Physical Society},
  doi = {10.1103/PhysRevLett.85.2208}
}

@article{PhysRevA.110.032619,
  title = {High-fidelity and robust controlled-$Z$ gates implemented with Rydberg atoms via echoing rapid adiabatic passage},
  author = {Xue, Ming and Xu, Shijie and Li, Xinwei and Li, Xiangliang},
  journal = {Phys. Rev. A},
  volume = {110},
  issue = {3},
  pages = {032619},
  numpages = {7},
  year = {2024},
  optoptmonth = {Sep},
  publisher = {American Physical Society},
  doi = {10.1103/PhysRevA.110.032619}
}

@article{RevModPhys.82.2313,
  title = {Quantum information with {Rydberg} atoms},
  author = {Saffman, Mark and Walker, Thad G. and M\o{}lmer, Klaus},
  journal = {Rev. Mod. Phys.},
  volume = {82},
  issue = {3},
  pages = {2313--2363},
  year = {2010},
  publisher = {American Physical Society},
  doi = {10.1103/RevModPhys.82.2313}
}

@article{LevinePRL19,
  title = {Parallel Implementation of High-Fidelity Multiqubit Gates with Neutral Atoms},
  author = {Levine, Harry and Keesling, Alexander and Semeghini, Giulia and Omran, Ahmed and Wang, Tout T. and Ebadi, Sepehr and Bernien, Hannes and Greiner, Markus and Vuleti\ifmmode \acute{c}\else \'{c}\fi{}, Vladan and Pichler, Hannes and Lukin, Mikhail D.},
  journal = {Phys. Rev. Lett.},
  volume = {123},
  issue = {17},
  pages = {170503},
  numpages = {6},
  year = {2019},
  optmonth = {Oct},
  publisher = {American Physical Society},
  doi = {10.1103/PhysRevLett.123.170503},
  url = {https://link.aps.org/doi/10.1103/PhysRevLett.123.170503}
}

@article{Wassner2026holonomicquantum,
  doi = {10.22331/q-2026-04-23-2080},
  title = {Holonomic quantum computation: a scalable adiabatic architecture},
  author = {Wassner, Clara and Guaita, Tommaso and Eisert, Jens and Carrasco, Jose},
  journal = {{Quantum}},
  optissn = {2521-327X},
  publisher = {{Verein zur F{\"{o}}rderung des Open Access Publizierens in den Quantenwissenschaften}},
  volume = {10},
  pages = {2080},
  optmonth = apr,
  year = {2026}
}

@inproceedings{impagliazzo1989one,
  title={One-way functions are essential for complexity based cryptography},
  author={Impagliazzo, Russell and Luby, Michael},
  booktitle={30th Annual Symposium on Foundations of Computer Science},
  pages={230--235},
  year={1989},
  organization={IEEE Computer Society}
}

@inbook{Behera_2025,
   title={A new world in the depths of Microcrypt: Separating OWSGs and quantum money from QEFID},
   ISBN={9783031910982},
   ISSN={1611-3349},
   url={http://dx.doi.org/10.1007/978-3-031-91098-2_2},
   DOI={10.1007/978-3-031-91098-2_2},
   booktitle={Advances in Cryptology – EUROCRYPT 2025},
   publisher={Springer Nature Switzerland},
   author={Behera, Amit and Malavolta, Giulio and Morimae, Tomoyuki and Mour, Tamer and Yamakawa, Takashi},
   year={2025},
   pages={23–52} }

@misc{PKE-deletion,
      author = {Fuyuki Kitagawa and Ryo Nishimaki and Takashi Yamakawa},
      title = {Publicly Verifiable Deletion from Minimal Assumptions},
      howpublished = {Cryptology {ePrint} Archive, Paper 2023/538},
      year = {2023},
      url = {https://eprint.iacr.org/2023/538}
}

@article{Li_2026,
   title={Universal and efficient quantum state verification via Schmidt decomposition and mutually unbiased bases},
   volume={10},
   ISSN={2521-327X},
   url={http://dx.doi.org/10.22331/q-2026-03-04-2011},
   DOI={10.22331/q-2026-03-04-2011},
   journal={Quantum},
   publisher={Verein zur Forderung des Open Access Publizierens in den Quantenwissenschaften},
   author={Li, Yunting and Zhu, Huangjun},
   year={2026},
   month=Mar, pages={2011} }

@misc{owsgevowpseperations,
      author = {John Bostanci and Boyang Chen and Barak Nehoran},
      title = {Oracle separation between quantum commitments and quantum one-wayness},
      howpublished = {Cryptology {ePrint} Archive, Paper 2024/1568},
      year = {2024},
      url = {https://eprint.iacr.org/2024/1568}
}

@article{Conrad_2026,
   title={Chasing shadows with Gottesman-Kitaev-Preskill codes},
   volume={10},
   ISSN={2521-327X},
   url={http://dx.doi.org/10.22331/q-2026-01-19-1973},
   DOI={10.22331/q-2026-01-19-1973},
   journal={Quantum},
   publisher={Verein zur Forderung des Open Access Publizierens in den Quantenwissenschaften},
   author={Conrad, Jonathan and Eisert, Jens and Flammia, Steven T.},
   year={2026},
   month=Jan, pages={1973} }

@article{Becker_2024,
   title={Classical shadow tomography for continuous variables quantum systems},
   volume={70},
   ISSN={1557-9654},
   url={http://dx.doi.org/10.1109/TIT.2024.3357972},
   DOI={10.1109/tit.2024.3357972},
   number={5},
   journal={IEEE Transactions on Information Theory},
   publisher={Institute of Electrical and Electronics Engineers (IEEE)},
   author={Becker, Simon and Datta, Nilanjana and Lami, Ludovico and Rouze, Cambyse},
   year={2024},
   month=May, pages={3427–3452} }

@misc{coladangelo2026robustquantumstatecertification,
      title={Robust quantum state certification and uncertainty principles for total influence}, 
      author={Andrea Coladangelo and Jerry Li and Joseph Slote},
      year={2026},
      eprint={2607.27184},
      archivePrefix={arXiv},
      optprimaryClass={quant-ph},
      url={https://arxiv.org/abs/2607.27184}, 
}

@misc{cojocaru2026equivalenceaveragecasehardnesslearning,
      title={Equivalence Between Average-Case Hardness of Learning and Cryptography for Mixed Quantum States}, 
      author={Alexandru Cojocaru and Laura Lewis},
      year={2026},
      eprint={2608.14331},
      archivePrefix={arXiv},
      primaryClass={quant-ph},
      url={https://arxiv.org/abs/2608.14331}, 
}

@article{Lanyon_2017,
   title={Efficient tomography of a quantum many-body system},
   volume={13},
   ISSN={1745-2481},
   url={http://dx.doi.org/10.1038/nphys4244},
   DOI={10.1038/nphys4244},
   number={12},
   journal={Nature Physics},
   publisher={Springer Science and Business Media LLC},
   author={Lanyon, B. P. and Maier, C. and Holzäpfel, M. and Baumgratz, T. and Hempel, C. and Jurcevic, P. and Dhand, I. and Buyskikh, A. S. and Daley, A. J. and Cramer, M. and Plenio, M. B. and Blatt, R. and Roos, C. F.},
   year={2017},
 pages={1158–1162} }

@book{goldreich2001foundations,
  title={Foundations of cryptography},
  author={Goldreich, Oded and others},
  volume={1},
  number={3},
  year={2001},
  publisher={Cambridge university press Cambridge}
}

@inproceedings{Radian_2019, series={AFT ’19},
   title={Semi-Quantum Money},
   url={http://dx.doi.org/10.1145/3318041.3355462},
   DOI={10.1145/3318041.3355462},
   booktitle={Proceedings of the 1st ACM Conference on Advances in Financial Technologies},
   publisher={ACM},
   author={Radian, Roy and Sattath, Or},
   year={2019},
   month=Oct, pages={132–146},
   collection={AFT ’19} }

@misc{kashefi2007statisticalzeroknowledgequantum,
      title={Statistical Zero Knowledge and quantum one-way functions}, 
      author={Elham Kashefi and Iordanis Kerenidis},
      year={2007},
      eprint={quant-ph/0511266},
      archivePrefix={arXiv},
      primaryClass={quant-ph},
      url={https://arxiv.org/abs/quant-ph/0511266}, 
}

@phdthesis{impagliazzo1992pseudo,
  title={Pseudo-random generators for cryptography and for randomized algorithms},
  author={Impagliazzo, Russell},
  year={1992},
  school={PhD thesis, University of California, Berkeley, 1992. http://cseweb. ucsd~…}
}

@article{Arapinis_2021,
   title={Quantum Physical Unclonable Functions: Possibilities and Impossibilities},
   volume={5},
   ISSN={2521-327X},
   url={http://dx.doi.org/10.22331/q-2021-06-15-475},
   DOI={10.22331/q-2021-06-15-475},
   journal={Quantum},
   publisher={Verein zur Forderung des Open Access Publizierens in den Quantenwissenschaften},
   author={Arapinis, Myrto and Delavar, Mahshid and Doosti, Mina and Kashefi, Elham},
   year={2021},
   pages={475} }
\end{document}